\documentclass[journal]{IEEEtran}
\ifCLASSINFOpdf
\else
\fi

\usepackage[ruled]{algorithm2e}
\usepackage{algorithmic}

\usepackage[hidelinks]{hyperref}
\usepackage{graphicx}
\graphicspath{{figures/}}

\usepackage{amsmath}
\usepackage{cases}
\usepackage{amssymb}
\usepackage{amsfonts}
\usepackage{mathtools}

\usepackage{amsthm}
\theoremstyle{plain}
\newtheorem{lemma}{Lemma}
\newtheorem{theorem}{Theorem}
\newtheorem{proposition}{Proposition}
\newtheorem{conjecture}{Conjecture}
\theoremstyle{definition}
\newtheorem{definition}{Definition}
\newtheorem{condition}{Condition}

\newcommand\abs[1]{\left|#1\right|}
\newcommand\where[2]{\left.#1\right|_{#2}}
\newcommand\difffrac[3][1]{
	\ifnum #1=1
		\frac{\mathrm{d} #2}{\mathrm{d} #3}
	\else
	\frac{{\mathrm{d}}^{#1} #2}{\mathrm{d} #3^{#1}}
	\fi
}

\newcommand\R{{\mathbb{R}}}
\newcommand\N{{\mathbb{N}}}

\newcommand{\Setminus}[2]{{\left.#1\middle\backslash #2\right.}}
\newcommand\manifold{{\mathcal{F}}}
\newcommand\problem{{\mathcal{P}}}
\newcommand\hamilton{{\mathcal{H}}}
\newcommand\swithchlaw{{\mathcal{S}}}

\newcommand\st{{\mathrm{s.t.}}}

\renewcommand\vector[1]{\boldsymbol{#1}}
\newcommand\ve{{\vector{e}}}
\newcommand\vf{{\vector{f}}}
\newcommand\vj{{\vector{j}}}

\newcommand\vt{{\vector{t}}}
\newcommand\vx{{\vector{x}}}
\newcommand\vy{{\vector{y}}}

\newcommand\vA{{\vector{A}}}
\newcommand\vB{{\vector{B}}}
\newcommand\vC{{\vector{C}}}
\newcommand\vD{{\vector{D}}}

\newcommand\vI{{\vector{I}}}
\newcommand\vJ{{\vector{J}}}
\newcommand\vM{{\vector{M}}}

\newcommand\vX{{\vector{X}}}
\newcommand\vlambda{{\vector{\lambda}}}
\newcommand\veta{{\vector{\eta}}}

\newcommand\vphi{{\vector{\phi}}}
\newcommand\vPhi{{\vector{\Phi}}}
\newcommand\vzero{{\vector{0}}}

\newcommand\tf{{t_\mathrm{f}}}
\newcommand\f{{\mathrm{f}}}
\newcommand\vxf{{\vx_\mathrm{f}}}
\newcommand\sgn{{\mathrm{sgn}}}
\newcommand\dom{{\mathrm{dom\,}}}
\newcommand\rank{{\mathrm{rank}}}

\usepackage{extarrows}
\usepackage{xcolor}
\newcommand\modify[1]{\textcolor{black}{#1}}

\newcommand\RomanNum[1]{\uppercase\expandafter{\romannumeral #1}}

\begin{document}
%
\title{Time-Optimal Control for High-Order Chain-of-Integrators Systems with Full State Constraints and Arbitrary Terminal States}
%
%
%
\author{Yunan~Wang,~
        Chuxiong~Hu,~\IEEEmembership{Senior~Member,~IEEE,}
        Zeyang~Li,
        Shize~Lin,
        Suqin~He,
        and~Yu~Zhu,~\IEEEmembership{Member,~IEEE}
\thanks{Corresponding author: Chuxiong Hu (e-mail: cxhu@tsinghua.edu.cn). The authors are all with the Department of Mechanical Engineering, Tsinghua University, Beijing 100084, China.}}

\maketitle

\begin{abstract}
    Time-optimal control for high-order chain-of-integrators systems with full state constraints and arbitrar\modify{ily} given terminal states remains a challenging problem in the optimal control theory domain, yet to be resolved. To enhance further comprehension of the problem, this paper establishes a novel notation system and theoretical framework, providing the switching manifold for high-order problems in the form of switching law\modify{s}. Through deriving properties of switching laws \modify{regarding} signs and dimension, this paper proposes a definite condition for time-optimal control. Guided by the developed theory, a trajectory planning method named the manifold-intercept method (MIM) is developed. The proposed MIM can plan time-optimal jerk-limited trajectories with full state constraints, and can also plan near-optimal \modify{non-chattering} higher-order trajectories with negligible extra motion time \modify{compared to optimal profiles}. Numerical results indicate that the proposed MIM outperforms all baselines in computational time, computational accuracy, and trajectory quality by a large gap.
\end{abstract}


%
\IEEEpeerreviewmaketitle


\section{Introduction}\label{sec:Introduction}

\IEEEPARstart{H}{igh-order} chain-of-integrators systems have been universally utilized in computer numerical control (CNC) machining \cite{wang2021local,wang2022real}, robotic motion control \cite{wang2022learning,zhao2020pareto,wang2021dynamics}, semiconductor device fabrication \cite{ni2017sinusoidal,li2018convergence}, and autonomous driving \cite{guler2016adaptive}. Trajectory planning has a significant influence on motion efficiency and accuracy in those scenarios \cite{wang2023optimization,wang2023slice}. However, time-optimal control for high-order chain-of-integrators systems with full state constraints and arbitrary terminal states remains a challenging and significant open problem in the optimal control theory domain, yet to be \modify{fully} resolved.

The time-optimal control problem for chain-of-integrators requires planning a trajectory with minimum motion time from a given initial state vector to a terminal state vector, where the system state vector is composed of components such as position, velocity, acceleration, and so forth. State constraints require that all system state components are limited by given upper bounds, while the initial states and the terminal states are specified arbitrarily. A suboptimal fifth order trajectory planned by the proposed Algorithm \ref{alg:CalMIM} is shown in Fig. \ref{fig:firstshow}. \modify{The input control, i.e., crackle, and} the system states, i.e., position, velocity, acceleration, jerk, and snap, are bounded by the given constraints. The \modify{control} is always \modify{at its} maximum, minimum, or zero \modify{value} along the planned trajectory, satisfying the Bang-Singular-Bang control law \cite{he2020time}.

\begin{figure}[!t]
    \centering
    \includegraphics[width=\columnwidth]{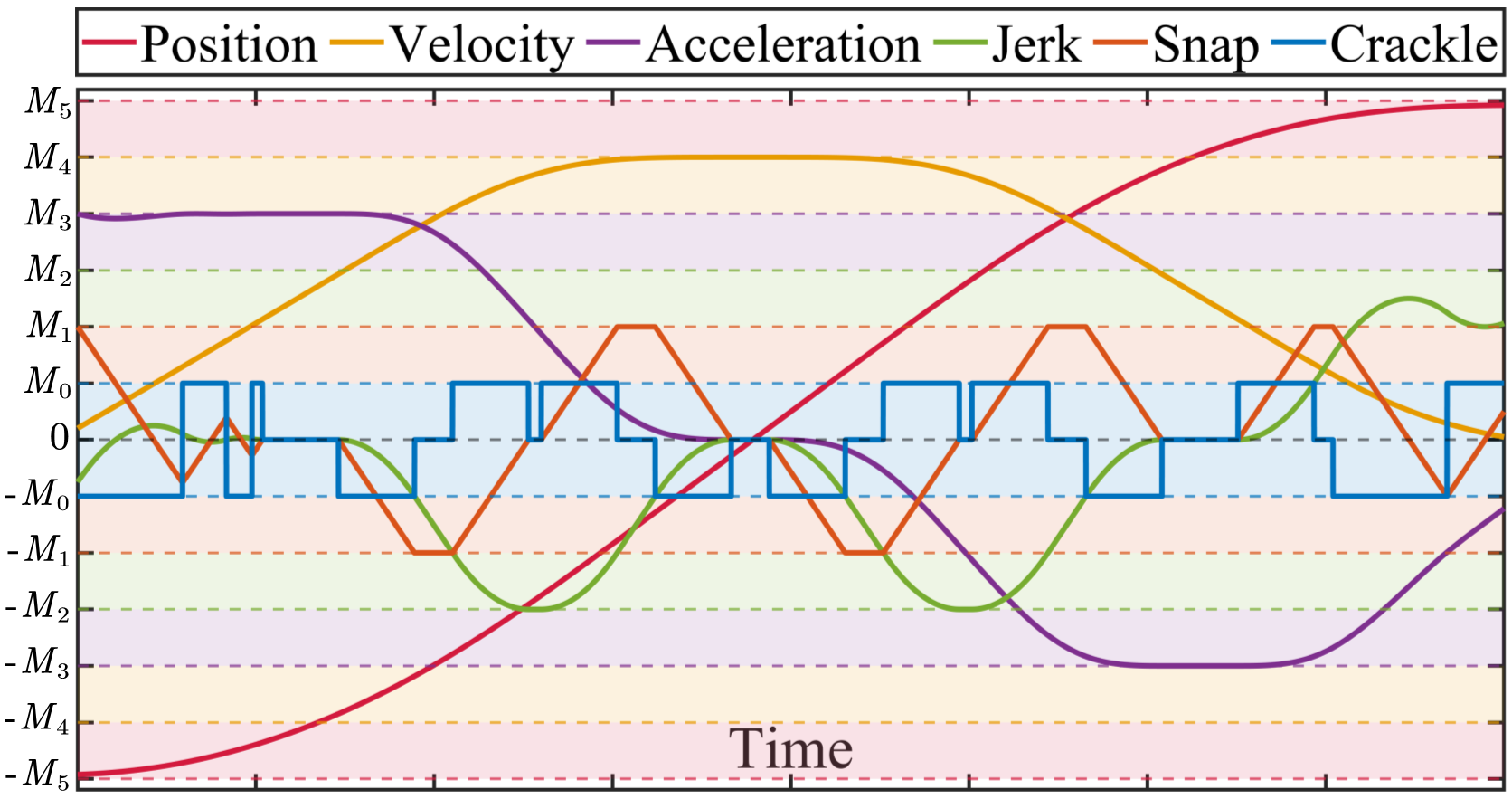}
    \caption{A fifth order trajectory planned by \modify{the} proposed method. \modify{$M_5$, $M_4$, $M_3$, $M_2$, $M_1$, and $M_0$ are the upper bounds of position, velocity, acceleration, jerk, snap, and crackle, respectively. The trajectory can be represented as $\underline{0}\overline{0}\left(\overline{3},2\right)\overline{0}\underline{0}\overline{03}\underline{01}\overline{0}\underline{2}\overline{01}\underline{0}\overline{4}\underline{01}\overline{0}\underline{2}\overline{01}\underline{0}\underline{3}\overline{01}\underline{0}\overline{0}$ in Section \ref{sec:CostateSystemBehaviorAnalysis}.}} 
    \label{fig:firstshow}
\end{figure}

In optimal control, numerous research has been conducted on time-optimal control for chain-of-integrators.
The problem with only input saturation has been fully solved based on Pontryagin's \modify{maximum} principle (PMP) \cite{hartl1995survey}, where theorems on the Bang-Bang control \cite{lee1967foundations} and the analytic expression of the optimal control \cite{bartolini2002time} are well-known.
Marchand et al. \cite{marchand2007global} proposed a discrete-time control law based on the above problem.
However, the time-optimal problem with input saturation and full state constraints remains unsolved, especially when the initial and terminal states are arbitrarily given, i.e., the terminal states are allowed to \modify{be} non-zero. The 2nd order problem is trivial \cite{ma2021optimal}, but 3rd or higher-order problems have not been solved well.
Haschke et al. \cite{haschke2008line} proposed an online time-optimal jerk-limited trajectory planning method without position constraints.
Kr{\"o}ger \cite{kroger2011opening} developed the Reflexxes type IV motion librar\modify{y}, solving the third order problem without given terminal acceleration and position constraints.
Berscheid \modify{and Kr{\"o}ger} \cite{berscheid2021jerk} fully considered third order chain-of-integrators with arbitrary terminal states and no position constraints, resulting in a widely-used jerk-limited trajectory planning package, i.e., Ruckig.
Some approaches on continuous path following time-optimal trajectory planning are proposed in \cite{he2022real}, with an order lower than 3. To the best of our knowledge, there has been no mature method available for planning 4th order optimal trajectories with full state constraints so far.

For higher-order problems, the control invariant set \cite{gao2020computing} and the switching surface \cite{walther2001computation} are significant tools, which provide the nature of the time-optimal problem \cite{khmelnitsky1998one}.
Mitchell et al. \cite{mitchell2005time} calculated the invariant set based on level sets of solutions of a partial differential equation.
Tahir et al. \cite{tahir2014low} employed a polyhedral approximation to characterize the invariant set.
Doeser et al. \cite{doeser2020invariant} constructed the third order invariant sets for integrators.
Yury \cite{yury2016quasi} obtained the switching surface for the 3rd order problem in part.
He et al. \cite{he2020time} provided the explicitly analytic expressions of the complete switching surfaces for the 3rd order problem with zero terminal states.
However, the investigation of switching surfaces for 3rd order problems remains incomplete, while limited studies on switching manifold have been conducted for 4th order or higher-order problems.
\modify{Specifically, the existence of chattering phenomena \cite{robbins1980junction}, i.e., the control switching for infinite times in a finite time period, still remains unknown in the investigated problem, let alone optimal control.}

Since higher-order problems are lack of theoretical results, scholars either discretize the continuous-time problem and solve the discrete problem \modify{using} numerical optimization, or plan feasible-but-not-optimal trajectories in engineering. Direct methods have been widely applied in numerical optimal control problems, and some optimization solvers have been built, such as CasADi \cite{andersson2019casadi}, ICLOCS2 \cite{nie2018iclocs2}, and Yop \cite{leek2016optimal}. However, the time-optimal control problem in the 4th or higher-order is non-convex in discrete time, leading to high computational time and failure in obtaining optimal solutions. Leomanni et al. \cite{leomanni2022time} transformed the time-optimal problem into sequential convex optimization problems, but the resulting trajectory exhibits serious oscillations. Solving the time-optimal problem for high-order chain-of-integrators systems efficiently and accurately is still challenging.

Quasi-optimal trajectories, especially \modify{those} in the S-curve form, are widely used in industry.
Erkorkmaz et al. \cite{erkorkmaz2001high} applied jerk-limited trajectories in the S-curve form in high-speed CNC machining.
Dai et al. \cite{dai2018generation} planned snap-limited trajectories by solving motion time in S-curves.
Ezair et al. \cite{ezair2014planning} proposed a greedy recursive trajectory planning method for the high-order problem with full state constraints, but the method would fail for cases where the terminal position is close to the initial position, and the terminal velocity is far from the initial velocity.
Oland et al. \cite{oland2019controlling} planned suboptimal 4th order trajectories without any constraints to obtain a zero terminal based on exponential activation functions.
However, the current methods either underutillize the inertia of system states or fall short of attaining full state constraints, thereby leading to unsatisfactory in motion time.

This paper theoretically studies the time-optimal control problem for high-order chain-of-integrators systems with full state constraints and arbitrary initial as well as terminal states.
A novel notation system and theoretical framework is established in Section \ref{sec:CostateSystemBehaviorAnalysis}, providing key concepts, i.e., the switching law and the optimal-trajectory manifold. A trajectory planning method named the manifold-intercept method (MIM) is proposed in Section \ref{sec:ManifoldInterceptMethod} \modify{and is open-source in \cite{MIMGithub2024}.} 
The contributions of this paper include the following aspects:
\begin{enumerate}
    \item This paper establishes a novel notation system and theoretical framework for the classical and longstanding problem in optimal control, i.e., the time-optimal control problem for high-order chain-of-integrators systems. Through comprehensive analysis of states and costates, the framework can provide the switching manifold for high-order problems in the form of switching laws. Notably, limited studies on the switching hypersurface have been conducted even for 4th order problems. This paper derives properties of switching laws \modify{regarding} signs as well as dimension and proposes a definite condition of augmented switching laws, imposing a necessary condition on optimal control \modify{for directly solving the nonlinear equation system}.
    \item Guided by the developed framework, this paper proposes an efficient trajectory planning algorithm called manifold-intercept method (MIM). The proposed MIM can plan near-time-optimal trajectories for 4th or higher-order problems with only negligible extra motion time compared to time-optimal trajectories\modify{, avoiding chattering in high-order problems with near-optimal terminal time compared to optimal solutions. The proposed MIM} outperforms all baselines on computational time, computational accuracy, and trajectory quality by a large gap. To the best of our knowledge, there has been no mature method available for planning 4th order optimal trajectories with full state constraints so far. 
    \item For 3rd order problems, the proposed MIM can plan strictly time-optimal trajectories with full state constraints. According to the available literature, the proposed MIM is the first 3rd order trajectory planning method with full state constraints and arbitrary terminal states. While it is claimed that Ruckig in pro version \cite{Ruckig2023Online} can achieve 3rd order trajectories with full state constraints, it is neither open-source nor explicit publishes its underlying principles. Since 3rd order trajectories are widely applied in the industry \cite{wang2021local,doeser2020invariant}, the proposed algorithm has significant application value.
\end{enumerate}

\section{Problem Formulation}\label{sec:Problem_Formulation}
    As mentioned in Section \ref{sec:Introduction}, time-optimal control for chain-of-integrator systems is a classical problem in the optimal control theory domain, remaining challenging to plan trajectories with input saturation, full state constraints \cite{he2020time}, and arbitrary terminal states \cite{kroger2011opening,berscheid2021jerk}. Generally, the problem can be summarized as
    \begin{IEEEeqnarray}{rl}\label{eq:optimalproblem}
        \min\quad& J=\int_{0}^{\tf}\mathrm{d}t=\tf,\IEEEyesnumber\IEEEyessubnumber*\\
        \st\quad&\modify{\dot{x}_k\left(t\right)=x_{k-1}\left(t\right),\,\forall t\in\left[0,t_\f\right],k=2,3,\dots,n,\label{eq:optimalproblem_dynamic}}\\
        &\modify{\dot{x}_1\left(t\right)=u\left(t\right),\,\forall t\in\left[0,t_\f\right],}\label{eq:optimalproblem_dynamic2}\\
        &\vx\left(0\right)=\vx_0,\,\vx\left(\tf\right)=\vx_\f,\\
        &\abs{x_k\left(t\right)}\leq M_k,\,\forall t\in\left[0,t_\f\right],\,k=1,2,\dots,n,\label{eq:optimalproblem_x_constraint}\\
        &\abs{u\left(t\right)}\leq M_0,\,\forall t\in\left[0,t_\f\right],
    \end{IEEEeqnarray}
    where $\vx=\left(x_1,x_2,\dots,x_n\right)\in\R^n$ is the state vector of the system, $u\in\R$ is the control input. $\vx_0=\left(x_{01},x_{02},\dots,x_{0n}\right)$ and $\vx_\f=\left(x_{\f1},x_{\f2},\dots,x_{\f n}\right)$ are the assigned initial state vector and terminal state vector. $n$ is called the order of problem \eqref{eq:optimalproblem}. The notation $\left(\bullet\right)$ means $\left[\bullet\right]^\top$.
    
    \modify{Problem \eqref{eq:optimalproblem} possesses a clear physical significance. For example, i}f $n=4$, $x_4$, $x_3$, $x_2$, $x_1$, and $u$ represent the position, velocity, acceleration, jerk, and snap of a 1-axis motion system, respectively. \modify{Then, \eqref{eq:optimalproblem} requires planning a trajectory with minimum motion time from a given initial state vector to a terminal state vector under box constraints.}
    

    For the convenience, denote $\vM=\left(M_0,M_1,\dots,M_n\right)\in\R_{++}\times\overline{\R}_{++}^{n}$, where $\overline{\R}_{++}=\R_{++}\cup\left\{\infty\right\}$ is the strictly positive part of the extended real number line. $M_k=\infty$ if $x_k$ is unconstrained. Problem \eqref{eq:optimalproblem} is denoted as $\problem\left(\vx_0,\vx_\f;\vM\right)$.

    In order to solve the optimal control problem \eqref{eq:optimalproblem}, the Hamiltonian is constructed as
    \begin{equation}\label{eq:hamilton}
        \begin{aligned}
            &\hamilton\left(\vx\left(t\right),u\left(t\right),\modify{\lambda_0,}\vlambda\left(t\right),\veta\left(t\right),t\right)\\
            =&\modify{\lambda_0}+\lambda_1u+\sum_{k=2}^{n}\lambda_k x_{k-1}+\sum_{k=1}^{n}\eta_k\left(\abs{x_k}-M_k\right),
        \end{aligned}
    \end{equation}
    where \modify{$\lambda_0\geq0$.} $\vlambda\left(t\right)=\left(\lambda_1\left(t\right),\lambda_2\left(t\right),\dots,\lambda_n\left(t\right)\right)\in\R^n$ is the costate vector, \modify{and $\left(\lambda_0,\vlambda\left(t\right)\right)\not=0$. $\vlambda$ satisfies} the Euler-Lagrange equations \cite{fox1987introduction}:
    \begin{equation}
        \dot\lambda_k=-\frac{\partial\hamilton}{\partial x_k},\,k=1,2,\cdots,n,
    \end{equation}
    i.e.,
    \begin{equation}\label{eq:derivative_costate}
        \begin{dcases}
            \dot\lambda_k=-\lambda_{k+1}-\eta_k\,\sgn\left(x_k\right),\,k<n,\\
            \dot\lambda_n=-\eta_n\,\sgn\left(x_n\right).
        \end{dcases}
    \end{equation}
    The initial costates $\vlambda\left(0\right)$ and the terminal costates $\vlambda\left(\tf\right)$ are not assigned since $\vx\left(0\right)$ and $\vx\left(\tf\right)$ are given in problem \eqref{eq:optimalproblem}.

    In \eqref{eq:hamilton}, $\veta$ is the multiplier vector induced by inequality constraints \eqref{eq:optimalproblem_x_constraint}, satisfying
    \begin{equation}\label{eq:eta_constraint_zero}
        \eta_k\geq0,\,\eta_k\left(\abs{x_k}-M_k\right)=0,\,k=1,2,\cdots,n.
    \end{equation}
    Equivalently, $\forall t\in\left[0,t_\f\right]$, $\eta_k\left(t\right)\not=0$ only if $\abs{x_k\left(t\right)}=M_k$.

    PMP gives the results that the input control $u\left(t\right)$ minimizes the Hamiltonian $\hamilton$ \modify{\cite{hartl1995survey}}, i.e.,
    \begin{equation}
        u\left(t\right)\in\mathop{\arg\min}\limits_{\abs{U}\leq M_0}\hamilton\left(\vx\left(t\right),U,\modify{\lambda_0,}\vlambda\left(t\right),\veta\left(t\right),t\right).
    \end{equation}
    Therefore,
    \begin{equation}\label{eq:bang_singular_bang_law}
        u\left(t\right)=
        \begin{dcases}
            M_0,&\lambda_1\left(t\right)<0\\
            *,&\lambda_1\left(t\right)=0\\
            -M_0,&\lambda_1\left(t\right)>0
        \end{dcases},
    \end{equation}
    where $u\left(t\right)$ is undetermined during $\lambda_1\left(t\right)=0$. \modify{Evidently}, the value of $u\left(t\right)$ in a zero-measure set \modify{have no influence on the integral result}. However, a singular condition occurs when $\lambda_1\left(t\right)=0$ holds for a period of time. The control law for the 3rd order version of \eqref{eq:bang_singular_bang_law} was previously reasoned and named the Bang-Singular-Bang time-optimal control law in \cite{he2020time}.

    The continuity of the system is guaranteed in two folds. On one hand,
    \begin{equation}\label{eq:hamilton_equiv_0}
        \forall t\in\left[0,t_\f\right],\,\hamilton\left(\vx\left(t\right),u\left(t\right),\modify{\lambda_0,}\vlambda\left(t\right),\veta\left(t\right),t\right)\equiv0
    \end{equation}
    holds along the time-optimal trajectory. On the other hand, the junction condition \cite{hartl1995survey} is \modify{introduced} as a guarantee of \eqref{eq:hamilton_equiv_0}. More specifically, $\vlambda$ might jump when the state vector $\vx$ enters or exits the boundaries of the inequality constraints \eqref{eq:optimalproblem_x_constraint}. In other words, $\vlambda$ might not be continuous when $\vx$ switches between $\abs{x_k}< M_k$ and $\abs{x_k}=M_k$ for some $k$.

\section{System Behavior Analysis and Switching Laws for the Time-Optimal Control Problem}\label{sec:CostateSystemBehaviorAnalysis}
    Section \ref{sec:Problem_Formulation} has provided the problem form and some properties of system states and costates. The Bang-Singular-Bang control law \eqref{eq:bang_singular_bang_law} indicates the importance of the costate analysis. Section \ref{subsec:Jump_Condition_of_Costates} and Section \ref{subsec:Roots_and_Orders_of_Costates} analyze costates and system behaviors, respectively. Then, system behaviors of problem \eqref{eq:optimalproblem} are classified into finite ones and the theory of the switching law is developed in Section \ref{subsec:SwitchingLawandOptimalStateManifold}. Finally, the definite condition of problem \eqref{eq:optimalproblem} is induced by the proposed augmented switching law in Section \ref{sub2sec:TangentMarker}.

    \subsection{Jump Condition of Costates $\vlambda$}\label{subsec:Jump_Condition_of_Costates}
        As mentioned in Section \ref{sec:Problem_Formulation}, the costate vector $\vlambda$ might not be continuous \modify{at the connection of unconstrained arcs and constrained boundary}, i.e., $\abs{x_k}$ increases onto $M_k$ or decreases from $M_k$ for some $k$. The above behavior is called the junction condition (or jump condition) of costates \modify{\cite{jacobson1971new}}. The junction condition has a significant influence on the switching law of the optimal control $u$. 



        Assume the junction condition for $\abs{x_k}\leq M_k$ occurs at time $t_1$. \modify{According to \cite{jacobson1971new}}, $\exists\mu\modify{\leq0}$, s.t. 
        \begin{equation}\label{eq:junction_condition}
            \vlambda\left(t_1^+\right)-\vlambda\left(t_1^-\right)=\mu\frac{\partial\left(\abs{x_k}-M_k\right)}{\partial\vx}=\mu\,\sgn\left(x_k\right)\ve_k,
        \end{equation}
        where the $k$-th component of $\ve_k\in\R^n$ is 1, and other components of $\ve_k$ are 0. In other words, $\forall i\not=k$, $\lambda_i$ is continuous at $t_1$, while $\lambda_k$ might jump at $t_1$.

        Two cases for the junction condition exist\modify{:} (a) $\abs{x_k}\equiv M_k$ for a period of time, \modify{i.e., the connection of an unconstrained arc and a constrained arc;} (b) $\abs{x_k}$ touches $M_k$ at a single time point \modify{i.e., the connection of two unconstrained arcs at the constrained boundary}. The above two cases will be discussed in Section \ref{sub2sec:junction_period} and Section \ref{sub2sec:junction_TouchMaxTimePoint}. Among them, Case 1 induces the system behavior defined in Definition \ref{def:SystemBehavior}, while Case 2 induces the tangent marker in Definition \ref{def:TangentMarker}. \modify{The limit point of chattering phenomena in problem \eqref{eq:optimalproblem}, in a one-sided neighborhood of which $\abs{x_k}=M_k$ and $\abs{x_k}<M_k$ occur for infinite times, i.e., unconstrained arcs are connected at the constrained boundary with lengths converging to 0, will be investigated in our related work \cite{wang2024part2}.}

        \subsubsection{Case 1. $\abs{x_k}\equiv M_k$ for A Period of Time}\label{sub2sec:junction_period}
            \quad
            
            Without loss of generality, assume that $\forall t\in\left[t_1,t_2\right]$, $x_k\left(t\right)\equiv M_k$, and $\exists \delta>0$, $\forall t\in\left(t_1-\delta,t_1\right)$, $x_k\left(t\right)<M_k$. The case where $x_k$ leaves $M_k$ or $x_k\equiv-M_k$ \modify{is} similar.

            \modify{For} $t\in\left[t_1,t_2\right]$, $x_k\equiv M_k$, $\dot{x}_1=u$, and $\forall i<k$, $\dot{x}_{i+1}=x_i$. Hence, $u\equiv0$ and $\forall i<k$, $x_i\equiv0$. \modify{Note that $\forall i>k$, $x_i\left(t\right)$ is a polynomial of degree $\left(i-k\right)$, and $\abs{\frac{\mathrm{d}^{\left(i-k\right)}x_i}{\mathrm{d}t^{\left(i-k\right)}}}\equiv M_k\not=0$. So $\abs{x_i}=M_i$ holds at most \modify{for} $2\left(i-k\right)$ number of points. In other words, $\forall i>k$, $\abs{x_i}<M_i$ holds except at finite time points.} From \eqref{eq:eta_constraint_zero} and \eqref{eq:bang_singular_bang_law}, $\lambda_1\equiv0$, and $\forall i\not=k$, $\eta_i\equiv0$ almost everywhere. It can be reasoned from \eqref{eq:derivative_costate} that $\forall i\leq k$, $\lambda_i\equiv0$. The term ``almost everywhere'' means a proposition holds for all points except for a zero-measure set \cite{stein2009real}. 


            Furthermore, in the case where $k=1$, $\lambda_1$ keeps continuous despite the junction condition \eqref{eq:junction_condition}. According to \eqref{eq:hamilton_equiv_0},
            \begin{equation}
                \hamilton^\pm=\modify{\lambda_0}+\lambda_1^{\modify{\pm}}u^\pm +\sum_{k=2}^{n}\lambda_k x_{k-1}=0.
            \end{equation} 
            \modify{From $\hamilton^+=\hamilton^-$}, $\lambda_1^+u^+=\lambda_1^-u^-$. The notation $\bullet^\pm$ means the left and right-hand limits of variable $\bullet$ at the junction time. Note that $u^+=0$, $u^-=M_0$, $\lambda_1^+=0$, so $\lambda_1^-=0$. Therefore, $\lambda_1^+=\lambda_1^-=0$, i.e., $\lambda_1$ keeps continuous at the junction time. 

            \modify{N}oted that the case where $\abs{x_n}\equiv M_n$ for a period of time does not exist. By contradiction, if $\exists t_1<t_2$, s.t. $\forall t\in\left[t_1,t_2\right]$, $\abs{x_n\left(t\right)}\equiv M_n$, then $x_1\left(t\right)=x_2\left(t\right)=\dots=x_{n-1}\left(t\right)=0$ \modify{for} $t\in\left[t_1,t_2\right]$. In other words, the system state vector $\vx$ parks at $\pm M_n\ve_n$ for a period of time, which contradicts the time-optimality. As a corollary, \eqref{eq:derivative_costate} can be written as
            \begin{equation}\label{eq:derivative_costate_lambdan_zero}
                \begin{dcases}
                    \dot\lambda_k=-\lambda_{k+1}-\eta_k\,\sgn\left(x_k\right),\,k<n,\\
                    \dot\lambda_n=0,
                \end{dcases}
            \end{equation}
            since $\eta_n=0$ almost everywhere. It should be pointed out that $\lambda_n$ might not be constant though $\dot\lambda_n=0$, because $\lambda_n$ might jump when $\abs{x_n}$ touches $M_n$ in the following Case 2.

            \modify{An example of Case 1 is shown in Fig. \ref{fig:costate_show1}. $\abs{x_1}\equiv M_1$ \modify{for} $t\in\left[t_1,t_2\right]\cup\left[t_5,t_6\right]$, but $\lambda_1$ keeps continuous despite to the junction condition \eqref{eq:junction_condition}. $x_2\equiv M_2$ \modify{for} $t\in\left[t_3,t_4\right]$, and $\lambda_2$ jumps decreasingly at $t_3$ and $t_4$.} 

            \begin{figure}[!t]
                \centering
                \includegraphics[width=\columnwidth]{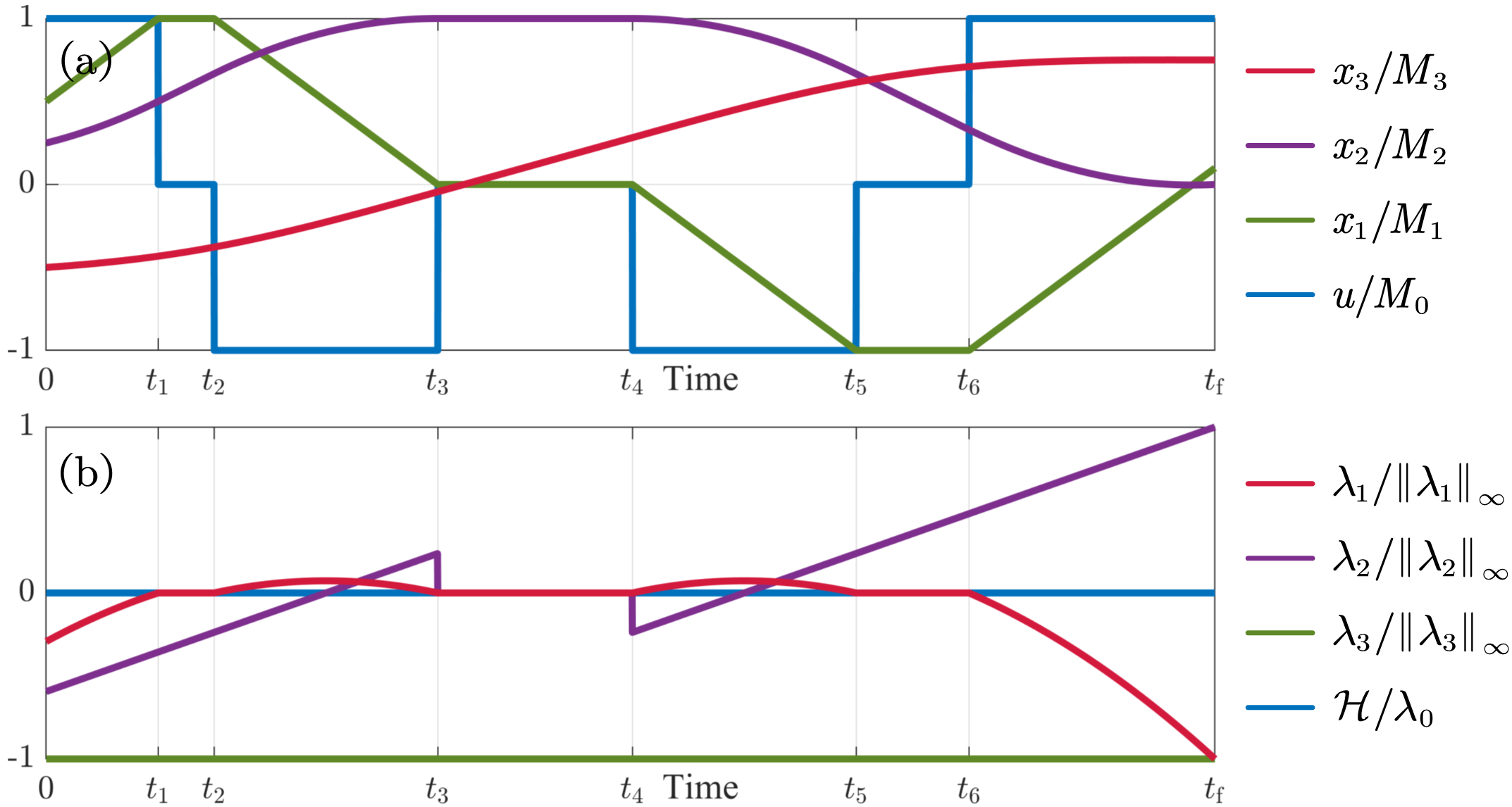}
                \caption{\modify{A 3rd order optimal trajectory represented by $\overline{01}\underline{0}\overline{2}\underline{01}\overline{0}$ in this paper, where $\lambda_0>0$. (a) The state vector. (b) The costate vector.}}
                \label{fig:costate_show1}
            \end{figure} 

        \subsubsection{Case 2. $\abs{x_k}$ Touches $M_k$ at A Single Time Point}\label{sub2sec:junction_TouchMaxTimePoint}
            \quad
            
            Without loss of generality, assume that $x_k$ touches $M_k$ at $t_1$, i.e., $x_k\left(t_1\right)=M_k$, and $\exists \delta>0$, $\forall t\in\left(t_1-\delta,t_1\right)\cup\left(t_1,t_1+\delta\right)$, $x_k\left(t\right)<M_k$. The case where $x_k$ touches $-M_k$ can be reasoned similarly.
            


            The case where $k=1$ is evident to Case 1 in Section \ref{sub2sec:junction_period}. \modify{Without loss of generality, assume $x_1$ touches $M_1$ at $t_1$. $\exists\varepsilon>0$, $\forall t\in\left(0,\varepsilon\right)$, $u\left(t_1+t\right)<0$, while $u\left(t_1-t\right)>0$. By \eqref{eq:bang_singular_bang_law}, $\forall t\in\left(0,\varepsilon\right)$, $\lambda_1\left(t_1-t\right)\leq0\leq\lambda_1\left(t_1+t\right)$; hence, $\lambda_1\left(t_1^+\right)\geq\lambda_1\left(t_1^-\right)$. \eqref{eq:junction_condition} indicates that $\lambda_1\left(t_1^+\right)-\lambda_1\left(t_1^-\right)\leq0$. So $\lambda_1\left(t_1^+\right)=\lambda_1\left(t_1^-\right)=0$, i.e., $\lambda_1$ keeps continuous at $t_1$.} 
            
            In the case where $k\geq 2$, $\exists \varepsilon\in\left(0,\delta\right)$, s.t. $\forall t\in\left(0,\varepsilon\right)$, $u\left(t_1+t\right)\equiv u^+$, and $u\left(t_1-t\right)\equiv u^-$. \modify{(a) If $x_{k-1}=x_{k-2}=\dots=x_{1}=0$ at $t_1$, then $\forall t\in\left(0,\varepsilon\right)$, $x_k\left(t_1+t\right)=M_k+\frac{u^+}{k!}t^k< M_k$, and $x_k\left(t_1-t\right)=M_k+\frac{u^-}{k!}\left(-t\right)^k< M_k$. Then, $u^+=-M_0$, $u^-=\left(-1\right)^{k-1}M_0$. (b) If $\exists i=1,2,\dots,k-1$, $x_i\left(t_1\right)\not=0$, then let $h^*=\mathrm{argmin}\left\{h:x_{k-h}\left(t_1\right)\not=0\right\}>1$, noting that $x_{k-1}\left(t_1\right)=\where{\frac{\mathrm{d}x_k}{\mathrm{d}t}}{t=t_1}=0$. Since $x_{k-h^*}$ is continuous, $\exists \tilde{\varepsilon}\in\left(0,\varepsilon\right)$, $\forall t\in\left(-\tilde{\varepsilon},\tilde{\varepsilon}\right)$, $\abs{x_{k-h^*}\left(t_1+t\right)-x_{k-h^*}\left(t_1+t\right)}<\frac12\abs{x_{k-h^*}\left(t_1\right)}$; hence, $\sgn\left(x_{k-h^*}\left(t_1+t\right)\right)\equiv\sgn\left(x_{k-h^*}\left(t_1\right)\right)$. By Taylor expansions of $x_k$ at $t_1$, $\forall t\in\left[-\tilde{\varepsilon},\tilde{\varepsilon}\right]$, $\exists\theta_t\in\left(0,1\right)$, $x_k\left(t_1+t\right)-x_k\left(t_1\right)=\frac{t^{h^*}}{h^*!}x_{k-h^*}\left(t_1+\theta_t t\right)<0$. Therefore, $x_{k-h^*}\left(t_1\right)<0$ and $h^*$ is even.} 

            For cases where (a) $n\leq2$, and where (b) $n=3$ with $M_3=\infty$\modify{, there exist costate vectors for optimal trajectories which do not jump when Case 2 occurs, noting that the costate vector can be non-unique for one optimal solution.} In the case where $n=3$ with $M_3<\infty$, the jerk-limited time-optimal problem under position constraints is still not well solved yet. He et al. \cite{he2020time} developed the analytic expression of switching surfaces in 3rd order with zero terminal states, where Case 2, fortunately, does not exist with zero terminal states. Reflexxes \cite{kroger2011opening} and Ruckig in community version \cite{berscheid2021jerk}, the two most famous open-source online trajectory planning packages, are not able to plan jerk-limited trajectories under position constraints. 

            \modify{An example of Case 2 is shown in Fig. \ref{fig:costate_show2}. Though $\dot\lambda_3\equiv0$ almost everywhere, $\lambda_3$ jumps at $t_3$ since $x_3$ touches $M_3$ at $t_3$. Two critical corollaries based on analysis above are as follows:} 

            

            \begin{figure}[!t]
                \centering
                \includegraphics[width=\columnwidth]{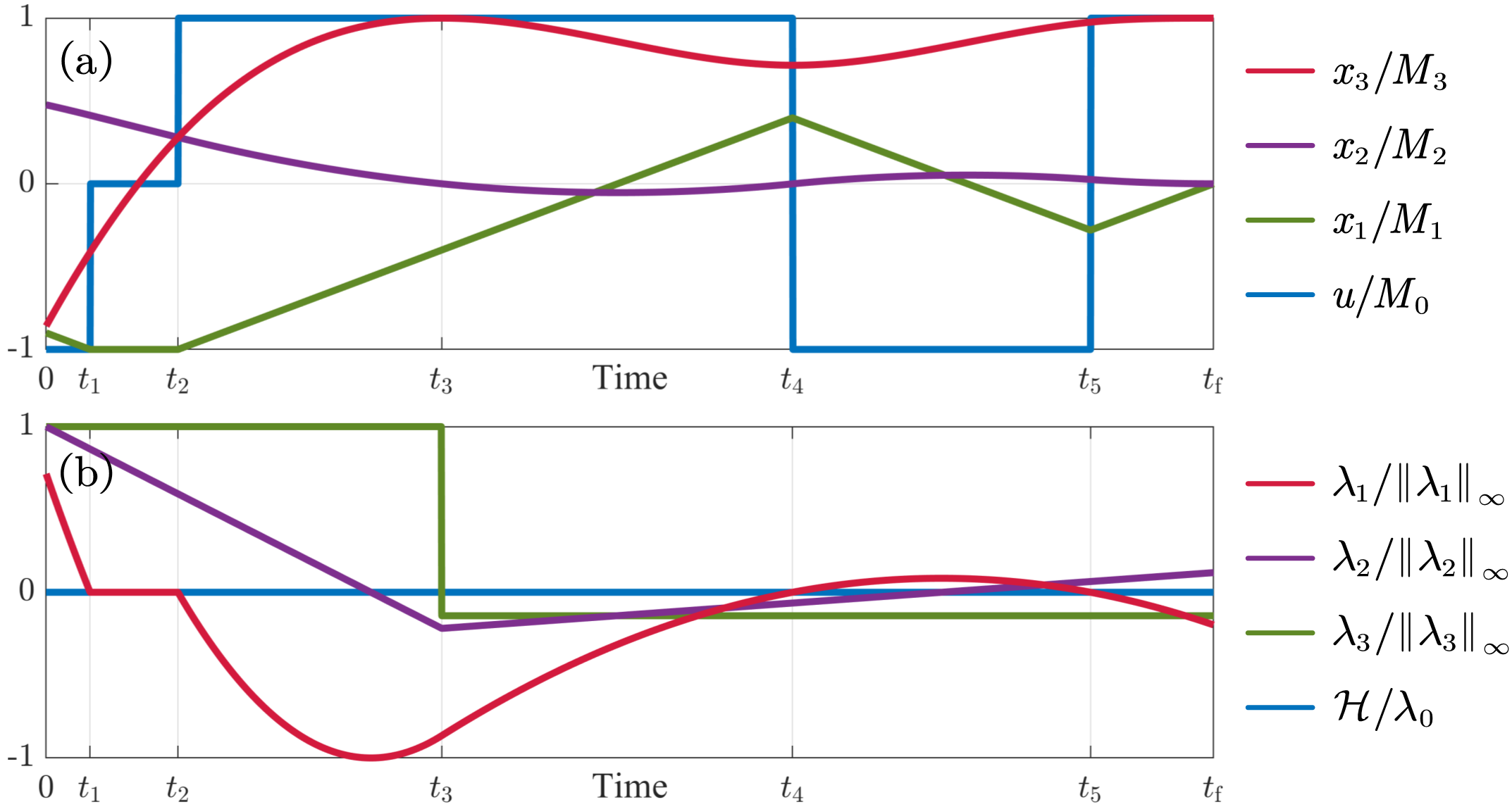}
                \caption{\modify{A 3rd order optimal trajectory represented by $\underline{01}\overline{0}\left(\overline{3},2\right)\overline{0}\underline{0}\overline{0}$ in this paper, where $\lambda_0>0$. (a) The state vector. (b) The costate vector.}}
                \label{fig:costate_show2}
            \end{figure} 

            \begin{proposition}[Bang-Singular-Bang Control Law]\label{prop:bang_zero_bang_law}
                The optimal control $u\left(t\right)$ of \eqref{eq:optimalproblem} satisfies
                \begin{equation}\label{eq:bang_zero_bang_law}
                    u\left(t\right)=
                    \begin{dcases}
                        M_0,&\lambda_1\left(t\right)<0\\
                        0,&\lambda_1\left(t\right)=0\\
                        -M_0,&\lambda_1\left(t\right)>0
                    \end{dcases}\text{ almost everywhere.}
                \end{equation}
            \end{proposition}
            \begin{proof}
                $\lambda_1\left(t\right)\equiv0$ occurs only if $\exists k$, s.t. $\abs{x_k\left(t\right)}\equiv M_k$, or else $\vlambda\left(t\right)\equiv\vzero$. \modify{If $\vlambda\left(t\right)\equiv\vzero$, then \eqref{eq:hamilton_equiv_0} implies that $\lambda_0=0$, which contradicts $\left(\lambda_0,\vlambda\left(t\right)\right)\not=0$.} Therefore, $u\left(t\right)\equiv0$ if $\lambda_1\left(t\right)\equiv0$. Hence, \eqref{eq:bang_zero_bang_law} is equivalent to \eqref{eq:bang_singular_bang_law} almost everywhere.
            \end{proof}


            \begin{proposition}\label{prop:unique_optimal_control}
                \modify{If \eqref{eq:optimalproblem} has an optimal control, then the optimal control is unique in an almost everywhere sense. In other words, if $u=u_1^*\left(t\right)$ and $u=u_2^*\left(t\right)$, $t\in\left[0,t_\f^*\right]$, are both the optimal control of \eqref{eq:optimalproblem}, then $u_1^*\left(t\right)=u_2^*\left(t\right)$ almost everywhere.}
            \end{proposition}
            
            \begin{proof}
                \modify{Denote $\mu$ as the Lebesgue measure \cite{stein2009real} on $\R$, and $Q_1\triangleq\left\{t:u_1^*\left(t\right),u_2^*\left(t\right)\in\left\{0,\pm M_0\right\}\right\}$. From Proposition \ref{prop:bang_zero_bang_law}, $\mu\left(Q_1\right)=t_\f^*$. Let $Q_2\triangleq\left\{t:u_1^*\left(t\right)\not=u_2^*\left(t\right)\right\}$. If $\mu\left(Q_2\right)>0$, then denote $u_3^*\left(t\right)=\frac14u_1^*\left(t\right)+\frac34u_2^*\left(t\right)$, $t\in\left[0,t_\f^*\right]$. It can be verified that $u_3^*\left(t\right)$ is also an optimal control of problem \eqref{eq:optimalproblem}. However, $\forall t\in Q_1\cap Q_2$, $u_3^*\left(t\right)\not\in\left\{0,\pm M_0\right\}$; hence, $0=\mu\left\{u_3^*\left(t\right)\not\in\left\{0,\pm M_0\right\}\right\}>\mu\left(Q_1\cap Q_2\right)=\mu\left(Q_2\right)>0$, which contradicts Proposition \ref{prop:bang_zero_bang_law}. Therefore, $\mu\left(Q_2\right)=0$.} 
            \end{proof}

    \subsection{System Behavior}\label{subsec:Roots_and_Orders_of_Costates}




        The case where the inequality constraints \eqref{eq:optimalproblem_x_constraint} hold strictly, i.e., $\forall k$, $\abs{x_k}<M_k$, is first considered. \modify{In this} case, it is evident that $\vlambda$ does not jump and $\veta\equiv\vzero$ in this period. Therefore,
        \begin{equation}\label{eq:dotlambda_without_constraint}
            \begin{dcases}
                \dot\lambda_k=-\lambda_{k+1},\,k<n\\
                \dot\lambda_n=0
            \end{dcases},\text{ if }\forall t\in\left[t_1,t_2\right],\,\abs{\vx\left(t\right)}<\vM.
        \end{equation}
        Since $\vlambda$ is continuous, \eqref{eq:dotlambda_without_constraint} indicates that $\lambda_k\left(t\right)$ is an $\left(n-k\right)$-th degree polynomial for $t\in\left[t_1,t_2\right]$. According to the fundamental theorem of algebra \cite{stein2009real}, $\lambda_k$ has no more than $\left(n-k\right)$ roots \modify{for} $t\in\left[t_1,t_2\right]$. A corollary, also well-known in previous works \cite{lee1967foundations}, is reasoned as follows:
        \begin{proposition}
            Assume $\forall t\in\left[t_1,t_2\right]$, $\forall k$, $\abs{x_k\modify{\left(t\right)}}<M_k$ in problem \eqref{eq:optimalproblem}. Then, the optimal control $u$ switches between $M_0$ and $-M_0$ with no more than $\left(n-1\right)$ times \modify{in $\left[t_1,t_2\right]$}.
        \end{proposition}

        \begin{proof}
            If $\forall t\in\left[t_1,t_2\right]$, $\forall k$, $\abs{x_k}<M_k$, then $\vlambda$ is continuous. Specifically, $\lambda_n\equiv\mathrm{const}$. By \eqref{eq:dotlambda_without_constraint}, $\lambda_1$ is an $\left(n-1\right)$-th degree polynomial for $t\in\left[t_1,t_2\right]$. Note that $\lambda_1\equiv0$ contradicts \eqref{eq:hamilton_equiv_0} when $\veta\equiv0$\modify{.} Therefore, $\lambda_1$ has no more than $\left(n-1\right)$ roots. According to Proposition \ref{prop:bang_zero_bang_law}, $u$ switches between $M_0$ and $-M_0$ with no more than $\left(n-1\right)$ times.
        \end{proof}

        The properties of states and costates reasoned in Section \ref{subsec:Jump_Condition_of_Costates} and Section \ref{subsec:Roots_and_Orders_of_Costates} \modify{are} summarized as follows:
        \begin{theorem}\label{thm:costate}
            The following propositions hold for the optimal control of problem \eqref{eq:optimalproblem}.
            \begin{enumerate}
                \item $u\left(t\right)=-\sgn\left(\lambda_1\left(t\right)\right)M_0$ almost everywhere. Specifically, $u\left(t\right)\modify{\equiv}0$ if $\lambda_1\left(t\right)\modify{\equiv}0$.
                \item $\forall t\in\left[0,\tf\right]$, $\vlambda\left(t\right)\not=\vzero$.
                \item\label{subthm:lambda1continue} $\lambda_1$ and $\vx$ are continuous, while $\lambda_k\left(k>1\right)$ might jump at junction time.
                \item\label{subthm:lambdak_polynomial} $\lambda_k$ consists of $\left(n-k\right)$-th degree polynomials and zero. Specifically, $\lambda_1\equiv0$ for a period of time if and only if $\exists k,\abs{x_k}\equiv M_k$ during the period of time.
                \item For $k<n$, cases might exist where $\exists t_1<t_2,\delta>0$, s.t. $\forall t\in\left[t_1,t_2\right]$, $\abs{x_k\left(t\right)}\equiv M_k$, while $\forall t\in\left(t_1-\delta,t_1\right)\cup\left(t_2,t_2+\delta\right)$, $\abs{x_k\left(t\right)}<M_k$. Then, $\forall t\in\left[t_1,t_2\right]$,
                \begin{enumerate}
                    \item $\forall i\leq k$, $\lambda_i\left(t\right)\equiv0$, while $\forall i>k$, $\lambda_i\left(t\right)$ is an $\left(n-i\right)$-th degree polynomial. Furthermore, $\lambda_{k+1}\left(t\right)$ is not always zero.
                    \item $u\left(t\right)\equiv0$. $\forall i<k$, $x_i\left(t\right)\equiv0$.
                    \item \modify{$\forall i\not=k$, $\lambda_i$ is continuous.} \modify{Only if} $1<k<n$, $\lambda_k$ might jump at $t_1$ and $t_2$. 
                \end{enumerate}
                \item\label{subthm:costate_junction_secondbracket} For $k>2$, case might exist where $\exists t_1\in\left(0,\tf\right)$, $\delta>0$, s.t. $\abs{x_k\left(t_1\right)}=M_k$, while $\forall t\in\left(t_1-\delta,t_1\right)\cup\left(t_1,t_1+\delta\right)$, $\abs{x_k\left(t\right)}<M_k$. Then, \modify{$\lambda_k$ might jump at $t_1$, and one and only one of the following cases holds:
                \begin{enumerate}
                    \item $\exists l<\frac{k}{2}$, s.t. $x_{k-1}=x_{k-2}=\dots=x_{k-2l+1}=0$ at $t_1$, while $x_{k-2l}\left(t_1\right)\not=0$, and $\sgn\left(x_{k-2l}\left(t_1\right)\right)=-\sgn\left(x_{k}\left(t_1\right)\right)$. Denote $h=2l$ as the degree of $\abs{x_k\left(t_1\right)}=M_k$.
                    \item $x_{k-1}=x_{k-2}=\dots=x_1=0$ at $t_1$. $u\left(t_1^+\right)=-\frac{M_0}{M_k}x_k\left(t_1\right)$, and $u\left(t_1^-\right)=\left(-1\right)^{k-1}\frac{M_0}{M_k}x_k\left(t_1\right)$. Denote $h=k$ as the degree of $\abs{x_k\left(t_1\right)}=M_k$.
                \end{enumerate}} 
            \end{enumerate}
        \end{theorem}

        \begin{proof}
            Based on discussion in Section \ref{subsec:Jump_Condition_of_Costates} and Section \ref{subsec:Roots_and_Orders_of_Costates}, Theorem \ref{thm:costate} is evident.
        \end{proof}

        Theorem \ref{thm:costate} \modify{fully lists} system behavior\modify{s} in a single stage \modify{of a sub-arc without limit points of chattering phenomena in an optimal trajectory, since connections between unconstrained arcs or constrained arcs at constrained boundaries with positive length have been fully discussed}. The system behavior is classified into finite ones and is proposed formally as follows:

        \begin{definition}\label{def:SystemBehavior}
            \textbf{System behavior} of problem \eqref{eq:optimalproblem} at a single stage is denoted as follows:
            \begin{enumerate}
                \item The stage where $u\equiv M_0\,\left(-M_0\right)$ is denoted as $\overline{0}\,\left(\underline{0}\right)$.
                \item The stage where $x_k\equiv M_k\,\left(-M_k\right), u\equiv0$, and $\forall i<k,x_i\equiv0$, is denoted as $\overline{k}\,\left(\underline{k}\right)$.
                \item $\forall 0\leq k\leq n$, the signs of $\overline{k}$ and $\underline{k}$ are denoted \modify{by} $\sgn\left(\overline{k}\right)=1$, $\sgn\left(\underline{k}\right)=-1$. The value of $\overline{k}$ and $\underline{k}$ is denoted as $\abs{\overline{k}}=\abs{\underline{k}}=k$.
            \end{enumerate}
        \end{definition}

        Based on Theorem \ref{thm:costate} and Definition \ref{def:SystemBehavior}, system behaviors during the whole moving process can be studied. In the following, the sign of a system behavior can be left out if no ambiguity exists.

    \subsection{Switching Law and Optimal-Trajectory Manifold}\label{subsec:SwitchingLawandOptimalStateManifold}

        Building on the analysis of the system behavior in Section \ref{subsec:Jump_Condition_of_Costates} and Section \ref{subsec:Roots_and_Orders_of_Costates}, this section focuses on how the system behavior switches along the time-optimal trajectory. The core idea in this section is the switching law and the optimal-trajectory manifold for time-optimal control, which are defined in Section \ref{sub2sec:SwitchingLawandOptimalStateManifold_Definition}. The properties of the switching law on dimension and sign are reasoned in Section \ref{sub2sec:AnalysisOfDimension} and Section \ref{sub2sec:AnalysisOfSign}, respectively.

        \subsubsection{Definitions}\label{sub2sec:SwitchingLawandOptimalStateManifold_Definition}
            \quad
            
            

            \begin{definition}\label{def:switchinglaw}
                Given a problem $\problem$, assume the time-optimal trajectory passes through system behaviors $s_1,s_2,\dots,s_N$ successively. Then, the series of system behaviors $S=s_1s_2\cdots s_N$ is called the \textbf{switching law} w.r.t. $\problem$, denoted as $S=\swithchlaw\left(\problem\right)$, where $N\in\N\modify{{}^*}$ is called the length of $S$.
            \end{definition}

            \modify{The switching law of a problem is unique according to Proposition \ref{prop:unique_optimal_control}.} As an example, two second order optimal problems with the same terminal state vector are shown in Fig. \ref{fig:2order_demo_switchinglaw}. The switching law of $\problem_1=\problem\left(\vx_0^{\left(1\right)},\vx_\f;\vM\right)$ is $\swithchlaw\left(\problem_1\right)=\underline{0}\underline{1}\overline{0}$. In other word\modify{s}, along the time-optimal trajectory from $\vx_0^{\left(1\right)}$ to $\vx_\f$, $\exists 0<t_1<t_2<t_3<\tf$, s.t. 
            \begin{enumerate}
                \item $\underline{0}$: $\forall t\in\left(0,t_1\right),u\left(t\right)\equiv-M_0$. $\vx$ starts from $\vx\left(0\right)=\vx_0^{\left(1\right)}$, and $x_1$ enters $-M_1$ at $t_1$.
                \item $\underline{1}$: $\forall t\in\left(t_1,t_2\right),x_1\left(t\right)\equiv -M_1$, while $u\left(t\right)\equiv\modify{0}$. 
                \item $\overline{0}$: $\forall t\in\left(t_2,t_\f\right),u\left(t\right)\equiv M_0$. $x_1$ exits $-M_1$ at $t_2$, and the system state vector $\vx$ reaches the terminal state vector $\vx_\f$ at $\tf$.
            \end{enumerate}
            A similar analysis applies to $\problem_2=\problem\left(\vx_0^{\left(2\right)},\vx_\f;\vM\right)$. From examples shown in Fig. \ref{fig:2order_demo_switchinglaw}, it can be observed that the switching law depends on the initial states.

            Notably, a switching law focuses on how the system behavior switches along the trajectory, without information on motion time. The optimal control of problem \eqref{eq:optimalproblem} consists of a switching law and motion time for every system behavior, \modify{whose} definite condition\modify{s} will be discussed in Section \ref{sub2sec:TangentMarker}.

            For an $n$-th order problem \eqref{eq:optimalproblem}, a natural approach is to employ mathematical induction, utilizing the solutions of lower-order problems to solve the $n$-th order problem. The above idea induces the definitions of \modify{the} sub-problem and optimal-trajectory manifold.

            For the convenience of discussion, the optimal trajectory, the optimal control, and the terminal time of a problem $\problem$ are denoted \modify{by} $\vx^*\left(t;\problem\right),u^*\left(t;\problem\right)$, and $t_\f^*\left(\problem\right)$, respectively.

            \begin{figure}[!t]
                \centering
                \includegraphics[width=0.65\columnwidth]{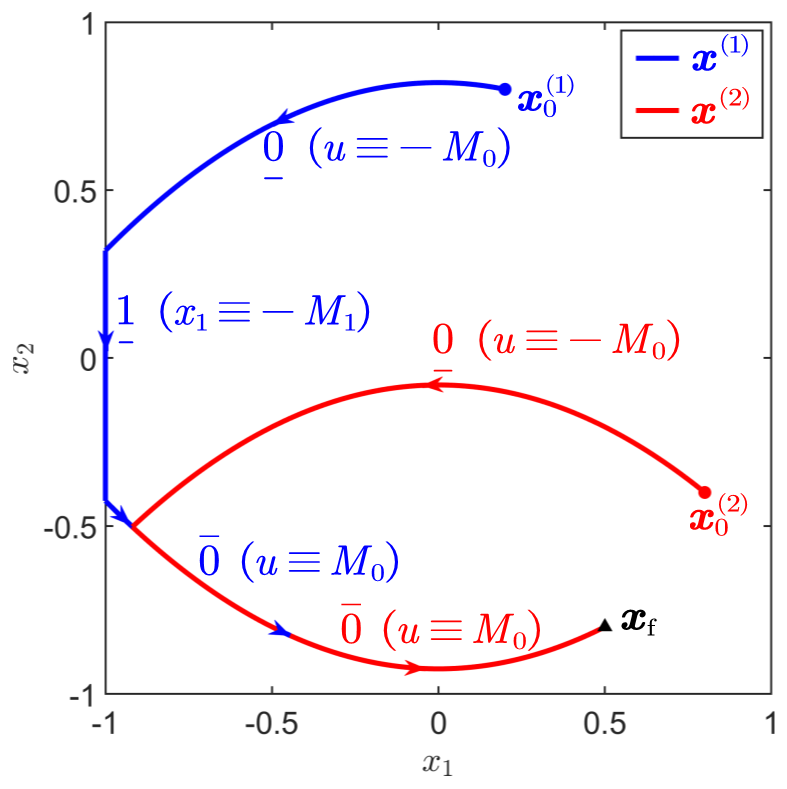}
                \caption{Switching laws for 2 optimal problems with the same terminal state vector $\vx_\f$, i.e., $\problem_1=\problem\left(\vx_0^{\left(1\right)},\vx_\f;\vM\right)$ and $\problem_2=\problem\left(\vx_0^{\left(2\right)},\vx_\f;\vM\right)$. Among them, $n=2$, $\vM=\left(1,1,1\right)$, $\mathcal{S}\left(\problem_1\right)=\underline{0}\underline{1}\overline{0}$, and $\mathcal{S}\left(\problem_2\right)=\underline{0}\overline{0}$.
                }
                \label{fig:2order_demo_switchinglaw}
            \end{figure}



            \begin{definition}
                For an $n$-th order problem $\problem=\problem\left(\vx_0,\vxf;\vM\right)$ and $0\leq a_i\leq b_i\leq n$, $a_1,a_2\geq1$, a \textbf{sub-problem} of $\problem$ is $\hat{\problem}=\problem\left(\vx_0^{a_1:b_1},\vx_\f^{a_2:b_2};\vM^{a_3:b_3}\right)$, i.e., the time-optimal problem for chain-of-integrators with the following constraints:
                \begin{equation}
                    \begin{dcases}
                        x_k\left(0\right)=x_{0k},&a_1\leq k\le b_1\\
                        x_k\left(\tf\right)=x_{\f k},&a_2\leq k\le b_2,\\
                        \abs{x_k}\leq M_k,&a_3\leq k\le b_3.
                    \end{dcases}
                \end{equation}
            \end{definition}

            \begin{definition}
                Two time-optimal problems $\problem_1,\problem_2$ are \textbf{equivalent}, denoted as $\problem_1\Leftrightarrow\problem_2$, if $\problem_1$ and $\problem_2$ have the same solution of optimal control, i.e.,
                \begin{equation}
                    \begin{dcases}
                        t_\f^*\left(\problem_1\right)=t_\f^*\left(\problem_2\right),\IEEEyesnumber\IEEEyessubnumber*\\
                        u^*\left(t,\problem_1\right)=u^*\left(t,\problem_2\right)\,\text{\modify{almost everywhere}}.
                    \end{dcases}
                \end{equation} 
            \end{definition}
            \modify{The equivalence between two problems is well-defined since the optimal for each problem is unique, according to Proposition \ref{prop:unique_optimal_control}.} Two problems are equivalent when their boundary conditions can replace each other. Specifically, if a problem $\problem$ is equivalent to its sub-problem $\hat{\problem}$, then some boundary conditions and constraints of $\problem$ is unnecessary. The above observation contributes to the following definition:

            \begin{definition}
                Given $\vxf\in\R^n$ and $\vM\in\R_{++}\times\overline{\R}_{++}^{n}$, define
                \begin{equation}\label{eq:definition_state_optimal_manifold}
                    \begin{aligned}
                        \manifold_{k}\left(\vxf,\vM\right)\triangleq\bigcup\left\{\vx_0\in\R^n:\problem\left(\vx_0,\vx_\f;\vM\right)\Leftrightarrow\right. \\ 
                        \left.\problem\left(\vx_0^{1:k},\vx_\f^{1:k};\vM\right)\text{ are feasible}\right\}
                    \end{aligned}
                \end{equation}
                as the $k$-th order \textbf{optimal-trajectory manifold} of $\vxf$ under $\vM$. \modify{Proposition \ref{prop:unique_optimal_control} provides the well-posedness of $\manifold_{k}\left(\vxf,\vM\right)$.} 
            \end{definition}

            ``$\Leftrightarrow$" is evidently an equivalence relation among the set of feasible time-optimal control problems, which ensures the well-definedness of \eqref{eq:definition_state_optimal_manifold}. Intuitively, $\manifold_{k}\left(\vxf,\vM\right)$ consists of time-optimal trajectories for $k$-th order sub-problem, with high-dimensional components of states additionally. \modify{Noting that differential properties of $\manifold_{k}\left(\vxf,\vM\right)$ are not utilized in this work, this paper names it as a ``manifold'' to establish its geometric intuition as a hypersurface, as shown in Fig. \ref{fig:demos_manifold}.} The following proposition \modify{indicates that elements in} $\manifold_{k}\left(\vxf,\vM\right)$ \modify{are uniquely determined by their first $k$ components}.

            \begin{proposition}\label{thm:dim_manifold}
                \modify{If $\vx_0,\vy_0\in\manifold_{k}\left(\vxf,\vM\right)$, $\forall i\leq k$, $x_{0i}=y_{0i}$, then $\vx_0=\vy_0$.}
            \end{proposition}
            \begin{proof}
                \modify{By} $\problem\left(\vx_0,\vx_\f;\vM\right)\Leftrightarrow\problem\left(\vy_0,\vx_\f;\vM\right)$, \modify{denote the} optimal control as \modify{$u^*$}. Note that \modify{$u^*$ drives $\vx$ from} $\vx_0\modify{,}\vy_0$ \modify{to $\vxf$}; hence, \modify{\eqref{eq:optimalproblem_dynamic} and \eqref{eq:optimalproblem_dynamic2} imply} $\vx_0=\vy_0$.
            \end{proof}

            \begin{definition}
                Assume $\manifold_{k}\left(\vxf,\vM\right)\not=\varnothing$ in \eqref{eq:definition_state_optimal_manifold}. The \textbf{switching-law representation} of $\manifold_{k}\left(\vxf,\vM\right)$ is defined as
                \begin{equation}
                    \mathcal{SF}_k\left(\vxf,\vM\right)=\bigcup_{\vx_0\in\manifold_{k}\left(\vxf,\vM\right)}\left\{\swithchlaw\left(\problem\left(\vx_0,\vxf;\vM\right)\right)\right\}.
                \end{equation}
            \end{definition}

            An example of the second order optimal-trajectory manifold is shown in Fig. \ref{fig:demos_manifold}(a). It is indicated in Fig. \ref{fig:demos_manifold}(a) that $\mathcal{SF}_2\left(\vzero,\vM\right)=$\{$\overline{0}$, $\underline{0}\overline{0}$, $\underline{1}\overline{0}$, $\underline{0}\underline{1}\overline{0}$, $\underline{0}$, $\overline{0}\underline{0}$, $\overline{1}\underline{0}$, $\overline{0}\overline{1}\underline{0}$\}. Fig. \ref{fig:demos_manifold}(a) is corresponding to $\mathcal{F}_2\left(\vzero,\vM\right)$ in Fig. \ref{fig:demos_manifold}(b). It can be observed that $\mathcal{F}_2\left(\vzero,\vM\right)\subset\mathcal{F}_3\left(\vzero,\vM\right)$ is a sub-manifold of 2 dimension\modify{s}, which confirms the validity of Proposition \ref{thm:dim_manifold}. Furthermore, every switching law induces a sub-manifold, i.e., any given sub-manifold of $\manifold_{k}\left(\vxf,\vM\right)$ can be represented by a subset of $\mathcal{SF}_k\left(\vxf,\vM\right)$.

        \subsubsection{Dimension Property of the Switching Law}\label{sub2sec:AnalysisOfDimension}
            \quad
            
            The dimension analysis of the optimal-trajectory manifold is indispensable to solve the time-optimal problem \eqref{eq:optimalproblem}. For a given switching law $S$, the motion time of every stage can be \modify{determined by solving equations as follows} only if $S$ is of $n$ dimension\modify{, since the number of variables should equal the number of independent equations \cite{verbeke1995newton}}. 

            \begin{figure}[!t]
                \centering
                \includegraphics[width=0.7\columnwidth]{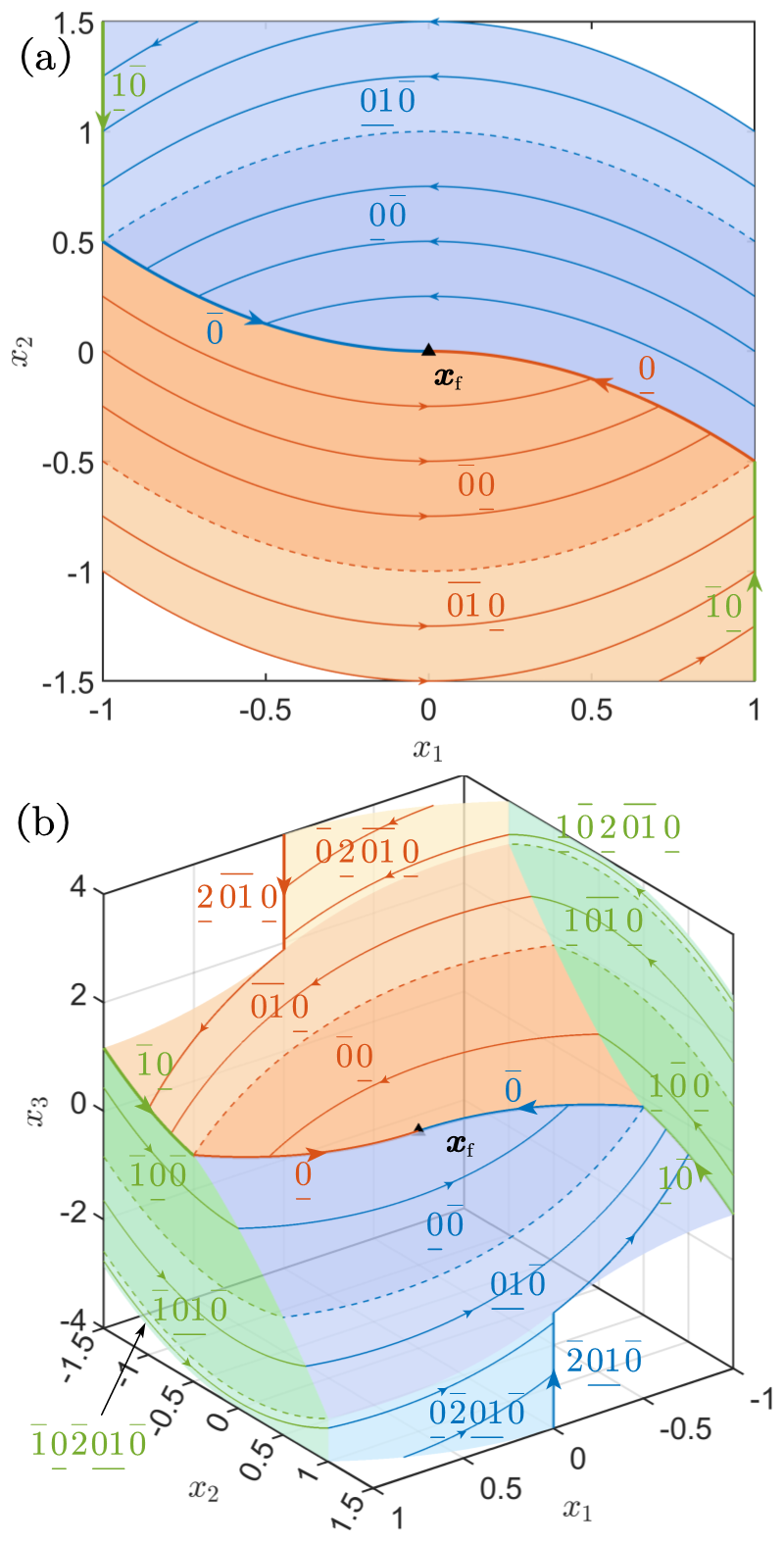}
                \caption{Examples of the optimal-trajectory manifolds. (a) Second order. $\vM=\left(1,1,1.5\right)$, $\vxf=\vzero$. (b) Third order. $\vM=\left(1,1,1.5,4\right)$, $\vxf=\vzero$. \modify{The parameters and the switching surfaces are the same as \cite{he2020time}.}} 
                \label{fig:demos_manifold}
            \end{figure}

            \begin{theorem}\label{thm:DimensionSwitchLaw}
                Given a switching law $S=s_1s_2\dots s_N\in\mathcal{SF}_k\left(\vxf,\vM\right)$, \modify{assume that $\forall 1\leq i<j\leq N+1$, if $\abs{s_i},\abs{s_j}>0$, then $\sum_{k=i+1}^{j-1}\abs{s_k}+\min\left\{\abs{s_i},\abs{s_j}\right\}<j-i$, where $\abs{s_{N+1}}\triangleq n$. T}he dimension of \modify{$S$, defined as the number of variables minus the number of equations except $\vx_0$}, is
                \begin{equation}\label{eq:DimensionSwitchLaw}
                    \dim S=N-\sum_{i=1}^{N}\modify{\abs{s_i}}.
                \end{equation} 
                Furthermore, if $\abs{s_i}\not=0$, then $\abs{s_{i-1}}=\abs{s_{i+1}}=0$.
            \end{theorem}
            \begin{proof}
                Without loss of generality, assume $k=n$, or else it can be considered as $S\in\mathcal{SF}_k\left(\vxf^{1:k},\vM\right)$ equivalently.

                Assume $\vx_i\in\R^n,i=0,1,\dots,N$ and $t_i\geq 0,i=1,2,\dots,N$. $\vx_{i-1}$ moves to $\vx_i$ after time $t_i$ under the system behavior $s_i$, $\forall i=1,2,\dots,N$. The terminal state vector $\vx_N=\vxf$ is given. \modify{Denote $\vx_{i}=\left(x_{i,k}\right)_{k=1}^n$. Note that $\abs{s_N}=0$.} 

                Under the system behavior $s_i$, the control $u\equiv u_i$, where
                \begin{equation}\label{eq:DimensionSwitchLaw_u}
                    u_i=\begin{dcases}
                        \sgn\left(s_i\right)M_0,&\abs{s_i}=0\\
                        0,&\text{otherwise}\\
                    \end{dcases}.
                \end{equation}
                
                Under the system behavior $\abs{s_i}=0$, $Nn$ constraints on state transformation exists, i.e., $\forall i=1,2,\dots,N,k=1,2,\dots,n$, 
                \begin{equation}\label{eq:constraints_state_transformation}
                    x_{i,k}=\sum_{j=1}^{k}\dfrac{1}{\left(k-j\right)!}x_{i-1,j}t_i^{k-j}+\dfrac{1}{k!}u_it_i^k.
                \end{equation}

                Under the system behavior $\abs{s_i}\not=0$, the states satisfy
                \begin{equation}\label{eq:constraints_state_system_behavior}
                    \begin{dcases}
                        x_{i-1,j}=0,j<\abs{s_i}\\
                        x_{i-1,\modify{\abs{s_i}}}=\sgn\left(s_i\right)M_{\abs{s_i}}
                    \end{dcases},\forall i=1,2,\dots,N\modify{-1}.
                \end{equation} 
                \eqref{eq:constraints_state_system_behavior} provides $\sum_{i=1}^{N}\abs{s_i}$ constraints.

                In summary, there exists $Nn+\sum_{i=1}^{N\modify{-1}}\abs{s_i}$ constraints and $\left(n+1\right)N$ variables in total. \modify{The independence of the above equations is proved in Appendix \ref{app:prove_independence}.} Therefore, the dimension of $S$ is
                \begin{equation}
                    \dim S=\left(n+1\right)N-\left(Nn+\sum_{i=1}^{N}\abs{s_i}\right)=N-\sum_{i=1}^{N}\abs{s_i}.
                \end{equation}

                Assume $\abs{s_i},\abs{s_{i-1}}\not=0$. Note that $s_{i-1}\not=s_i$, and assume $\abs{s_i}<\abs{s_{i-1}}$. According to \eqref{eq:constraints_state_system_behavior}, $\abs{x_{\abs{s_i}}\left(T_i^+\right)}=M_{\abs{s_i}}$, where $T_i$ is the switching time between $s_{i-1}$ and $s_i$. However, $x_{\abs{s_i}}\left(T_i^-\right)=0$; hence, $x_{\abs{s_i}}$ is not continuous at $T_i$, which leads to a contradiction. Therefore, if $\abs{s_i}\not=0$, then $\abs{s_{i-1}}=\abs{s_{i+1}}=0$.
            \end{proof}

            \modify{Intuitively, the dimension of a switching law characterizes the sub-manifold represented by the switching law.} Take the optimal-trajectory manifold in Fig. \ref{fig:demos_manifold}(b) as an example. $0$, $10$, and $2010$ are of 1 dimension. $00$, $1010$, and $02010$ are of 2 dimension\modify{s}. $000$, $0010$, and $0102010$ are of 3 dimension\modify{s}.
            


        \subsubsection{Sign Property of the Switching Law}\label{sub2sec:AnalysisOfSign}
            \quad
            
            From Fig. \ref{fig:demos_manifold}(b), it is observed that $\underline{2}\overline{01}\underline{0}$ exists while $\overline{201}\underline{0}$ and $\underline{20}\overline{1}\underline{0}$ do not exist. The analysis of the sign is given as the following theorem.
            \begin{theorem}\label{thm:AnalysisOfSign}
                $\forall S=s_1s_2\dots s_N\in\mathcal{SF}_k\left(\vxf,\vM\right)$, $\forall i=2,3,\dots,N$,
                \begin{equation}\label{eq:AnalysisOfSign}
                    \sgn\left(s_{i-1}\right)=\begin{dcases}
                        \sgn\left(s_{i}\right),&\text{if }\abs{s_i}\text{ is odd}\\
                        -\sgn\left(s_{i}\right),&\text{if }\abs{s_i}\text{ is even}
                    \end{dcases}.
                \end{equation}
            \end{theorem}
            \begin{proof}
                In the case where $\abs{s_{i-1}}=0,\,\abs{s_i}=0$, note that $s_i\not=s_{i-1}$. Hence, $\sgn\left(s_{i-1}\right)=-\sgn\left(s_{i}\right)$.

                In the case where $\abs{s_{i-1}}\not=0$, according to Theorem \ref{thm:DimensionSwitchLaw}, $\abs{s_i}=0$. Denote the switching time between $s_{i-1}$ and $s_i$ as $T_i$. If $\sgn\left(s_{i-1}\right)=+1$, then $\forall k<\abs{s_{i-1}}$, $x_{k}\left(T_i\right)=0$ and $x_{\abs{s_{i-1}}}\left(T_i\right)=M_{\abs{s_{i-1}}}$. Assume that $\sgn\left(s_i\right)=+1$, i.e., $u\equiv u_i=M_0$ during $s_i$. Then, $\exists \delta>0$, $\forall t\in\left(0,\delta\right)$, $x_{\abs{s_{i-1}}}\left(T_i+t\right)=M_{\abs{s_{i-1}}}+\frac{1}{\abs{s_{i-1}}!}M_0t^{\abs{s_{i-1}}}>M_{\abs{s_{i-1}}}$ since $u\equiv M_0$, which leads to a contradiction. Therefore, $\sgn\left(s_{i-1}\right)=+1$ if $\sgn\left(s_i\right)=-1$. For the similar analysis, $\sgn\left(s_{i-1}\right)=-1$ if $\sgn\left(s_i\right)=+1$.

                In the case where $\abs{s_{i}}\not=0$, according to Theorem \ref{thm:DimensionSwitchLaw}, $\abs{s_{i-1}}=0$. Denote the switching time between $s_{i-1}$ and $s_i$ as $T_i$. $x_{k}\left(T_i\right)=0$ for $k<\abs{s_i}$ and $x_{\abs{s_i}}\left(T_i\right)=M_{\abs{s_i}}$. Then, $\exists \delta>0$, $\forall t\in\left(0,\delta\right)$, $u\left(T_i-t\right)\equiv\sgn\left(s_{i}\right)M_0$; hence, $x_{\abs{s_{i}}}\left(T_i-t\right)=\sgn\left(s_i\right)M_{\abs{s_{i}}}+\frac{1}{\abs{s_{i}}!}M_0t^{\abs{s_{i}}}\cdot\left(-1\right)^{\abs{s_i}}\sgn\left(s_{i-1}\right)$. The constraint where $\abs{x_{\abs{s_i}}\left(T_i-t\right)}\leq M_{\abs{s_i}}$ resulting that $\left(-1\right)^{\abs{s_i}}\sgn\left(s_{i-1}\right)=-\sgn\left(s_{i}\right)$.
            \end{proof}

            Theorem \ref{thm:AnalysisOfSign} indicates that signs of all system behaviors in a switching law \modify{are} uniquely determined by the sign of the last system behavior. For a switching law $S$ \modify{of} length $N$, Theorem \ref{thm:AnalysisOfSign} reduces the possible signs of all system behaviors from $2^N$ to $2$.

    \subsection{Augmented Switching Law and Tangent Marker}\label{sub2sec:TangentMarker}
        For a given initial state vector $\vx_0\in\R^n$ and the switching law $S=\mathcal{S}\left(\problem\left(\vx_0,\vx_\f;\vM\right)\right)$, Section \ref{sub2sec:AnalysisOfDimension} has pointed out that the optimal control can be solved only if $\dim S=n$. For example, in Fig. \ref{fig:demos_manifold}(a), $\vx_0\in\R^2$ is given. If $S=\overline{0}\underline{0}$, then the full definite conditions are provided by \eqref{eq:DimensionSwitchLaw_u}, \eqref{eq:constraints_state_transformation}, and \eqref{eq:constraints_state_system_behavior}. However, $S=\overline{0}\underline{0}\overline{0}$ is underdefined, while $S=\overline{0}$ is overdefined\modify{; hence, the optimal control can be solved directly from neither $S=\overline{0}\underline{0}\overline{0}$ nor $S=\overline{0}$ by \eqref{eq:DimensionSwitchLaw_u}, \eqref{eq:constraints_state_transformation}, and \eqref{eq:constraints_state_system_behavior}}. \modify{The augmented switching law should be proposed to solve the problem.} 

        \begin{definition}\label{def:AugmentedSwitchingLaw}
            Given a problem $\problem$, \modify{an} \textbf{augmented switching law} w.r.t. $\problem$, denoted as $S\modify{\in}\mathcal{AS}\left(\problem\right)$, is the switching law of $\problem$ attached full definite conditions.
        \end{definition}

        \modify{Definite conditions attached to an augmented switching law $S$ should guarantee that $\dim S=n$ for an $n$-th order problem.} In the above example \modify{where $n=2$}, \modify{$\dim\overline{0}\underline{0}=2$, $\dim\overline{0}\underline{0}\overline{0}=3>2$, and $\dim\overline{0}<2$. Therefore, $\overline{0}\underline{0}$ provides full definite conditions and} is an augmented switching law, while $\overline{0}\underline{0}\overline{0}$ and $\overline{0}$ are not. As a trick, the switching law $\overline{0}$ can be represented by an augmented switching law $\underline{0}\overline{0}$ \modify{or $\overline{0}\underline{0}$}, where the motion time of $\underline{0}$ is 0. \modify{This example also shows that the augmented switching law of a problem can be not unique.} 


        \begin{figure}[!t]
            \centering
            \includegraphics[width=\columnwidth]{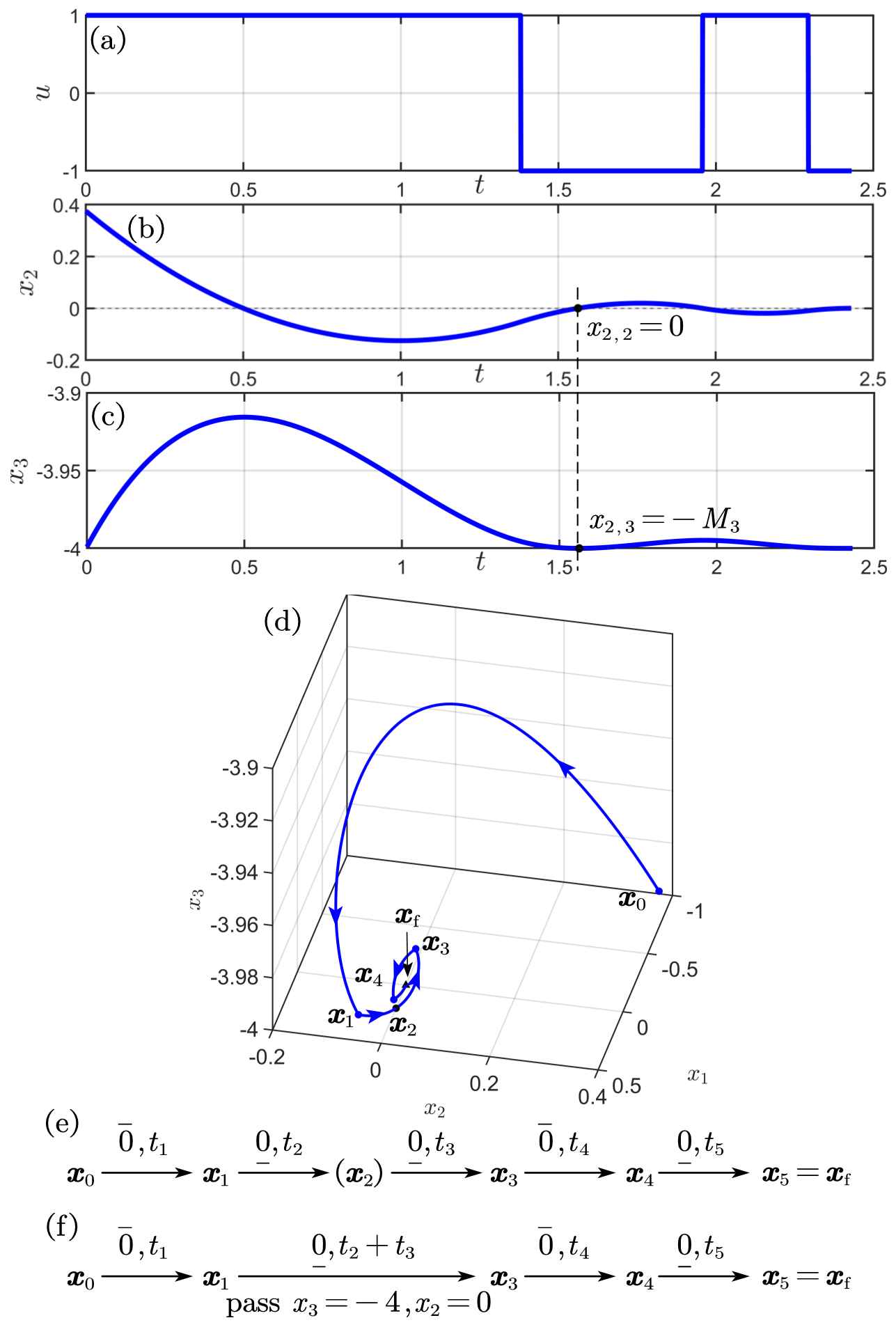}
            \caption{Time-optimal trajectory with a tangent marker. $\problem=\problem\left(\vx_0,\vx_\f;\vM\right)$, where $\vM=\left(1,1,1.5,4\right)$. \modify{An} augmented switching law \modify{for $\problem$} is $\modify{S'=}\overline{0}\underline{0}\left(\underline{3},2\right)\underline{0}\overline{0}\underline{0}\modify{\in\mathcal{AS}\left(\problem\right)}$, where the tangent marker $\left(\modify{\underline{3}},2\right)$ means $x_{2,\modify{3}}=-4,x_{2,\modify{2}}=0$, and $x_{2,\modify{1}}>0$. (a), (b) and (c) are the jerk, velocity, and position plots, respectively. (d) is the trajectories of system states $\vx$. (e) and (f) are the flow charts for $\modify{S'\in}\mathcal{AS}\left(\problem\right)$ and $\modify{S=}\mathcal{S}\left(\problem\right)\modify{\overline{0}\underline{0}\overline{0}\underline{0}}$.}  
            \label{fig:demo_second_bracket}
        \end{figure}


        However, the above trick of adding zero\modify{s} is not enough to provide all augmented switching laws. According to \modify{the} discussion on Case 2 in Section \ref{sub2sec:junction_TouchMaxTimePoint}, the tangent marker is necessary to \modify{be} define\modify{d as a kind of definite condition} and precisely exists in the time-optimal problem \eqref{eq:optimalproblem}. 

        \begin{definition}\label{def:TangentMarker}
            For a time-optimal trajectory, assume \modify{$\abs{x_k}$} touches \modify{$M_k$} at $t_1$ with \modify{a degree $h=k$} or an even degree \modify{$h<k$,} as described in Theorem \modify{\ref{thm:costate}-}\ref{subthm:costate_junction_secondbracket}. Then, the \textbf{tangent marker} is denoted as $\left(s,\modify{h}\right)$, where $\abs{s}=k$ and $\sgn\left(s\right)=\sgn\left(x_k\left(t_1\right)\right)$.
        \end{definition} 

        
        \begin{theorem}\label{thm:TangentMarkerBracket}
             $S=S_1\left(s,\modify{h}\right)S_2$ where $S_i=s_1^{\left(i\right)}s_2^{\left(i\right)}\dots s_{N_i}^{\left(i\right)}$ are augmented switching laws, $i=1,2$, then,
            \begin{enumerate}
                \item $\abs{s_{N_1}^{\left(1\right)}}=\abs{s_{1}^{\left(2\right)}}=0$.
                \item $\left(s,\modify{h}\right)$ contributes $-\modify{h}<0$ dimension.
            \end{enumerate}
        \end{theorem}
        \begin{proof}
            According to Section \ref{sub2sec:junction_TouchMaxTimePoint}, $\abs{s_{N_1}^{\left(1\right)}}=\abs{s_{1}^{\left(2\right)}}=0$, and $\left(s,\modify{h}\right)$ does not provide extra motion time. Assume the state vector at $\left(s,\modify{h}\right)$ is $\vx_1$. Then, $\left(s,\modify{h}\right)$ contributes $\modify{h}$ extra constraints, i.e., $x_{1,\abs{s}}=\sgn\left(s\right)M_{\abs{s}}$, and $x_{1,\abs{s}-k}=0$, $k=1,2,\dots,\modify{h}-1$. Hence, $\left(s,\modify{h}\right)$ contributes $-\modify{h}$ dimension.
        \end{proof}
        
        Fig. \ref{fig:demo_second_bracket} shows an example of the tangent marker. From Fig. \ref{fig:demo_second_bracket}(a), it is observed that $\mathcal{S}\left(\problem\right)=\overline{0}\underline{0}\overline{0}\underline{0}$. Therefore, $\dim\mathcal{S}\left(\problem\right)=4>3$. Hence, the problem $\problem$ is not determined by $\mathcal{S}\left(\problem\right)$. Opportunely,  $\modify{S'=}\overline{0}\underline{0}\left(\modify{\underline{3}},2\right)\underline{0}\overline{0}\underline{0}\modify{\in\mathcal{AS}\left(\problem\right)}$, where the tangent marker $\left(\modify{\underline{3}},2\right)$ induced 2 constraints, i.e., $x_{2,\modify{3}}=-M_2$ and $x_{2,\modify{2}}=0$. Therefore, $\dim\modify{S'}=5-2=3$, and $\problem$ is determined by $\modify{S'}$. As shown in Fig. \ref{fig:demo_second_bracket}(f), $\underline{0}\left(\modify{\underline{3}},2\right)\underline{0}$ with time $t_2,t_3$ performs the same as $\underline{0}$ with time $t_2+t_3$ on the surface, and $\left(\modify{\underline{3}},2\right)$ means the system state vector passes $x_{3}=-4$, $x_{2}=0$ during the stage $\underline{0}$. 

        \modify{The physical meaning of the tangent marker $\left(\underline{3},2\right)$ is clear. Without consideration of the constraint $x_3\geq-M_3$, the switching law should be $\underline{0}\overline{0}\underline{0}$, where $x_3<-M_3$ occurs. To guarantee $x_3\geq -M_3$, an accelerating stage $\overline{0}$ is introduced first, then $\underline{0}\overline{0}\underline{0}$ is applied for minimum motion time. The tangent marker $\left(\underline{3},2\right)$ shows that the first stage $\overline{0}$ lasts as short as possible. In other words, once the constraint $x_3\geq-M_3$ holds in the future, an optimal switching law $\underline{0}\overline{0}\underline{0}$ is applied. Therefore, the constraint $x_3\geq-M_3$ is active, and $\left(\underline{3},2\right)$ reflects some conditions for extremum.} 

        Considering the full discussion in Section \ref{subsec:Jump_Condition_of_Costates}, a conjecture is proposed that system behaviors and tangent markers can provide full definite condition\modify{s} for time-optimal trajectories.

        \begin{conjecture}\label{conjecture:ASL_systembehavior_tangentmarker}
            $\forall S$ is an augmented switching law of the time-optimal problem \eqref{eq:optimalproblem} \modify{where chattering phenomena do not occur}, $S$ consists of system behaviors in Definition \ref{def:SystemBehavior} and tangent markers in Definition \ref{def:TangentMarker}.
        \end{conjecture}

        Obviously, if Conjecture \ref{conjecture:ASL_systembehavior_tangentmarker} can be proven constructively, then \modify{there is a prospect of} the time-optimal problem for high-order chain-of-integrator systems \modify{which is} an open problem in the optimal control theory.


\section{Manifold-Intercept Method}\label{sec:ManifoldInterceptMethod}
    Section \ref{sec:CostateSystemBehaviorAnalysis} studies the properties of time-optimal control and builds a novel notation system of the time-optimal problem for chain-of-integrators. However, it is a daunting task to \modify{solve the optimal control} for arbitrary given problems $\problem\left(\vx_0,\vx_\f;\vM\right)$\modify{, especially when a chattering phenomenon occurs}. Based on conclusions of optimal control reasoned in Section \ref{sec:CostateSystemBehaviorAnalysis}, this section propose\modify{s} a trajectory planning method for high-order chain-of-integrators systems, named the manifold-intercept method (MIM).

    It should be pointed out that MIM is an efficient and quasi-optimal method. \modify{As is pointed out in our related work \cite{wang2024part2}, chattering phenomena exist in 4th order or higher-order problems, which impedes the computation of optimal control.} Section \ref{sec:results} will indicate that MIM is near-optimal for 4th or higher-order problems\modify{, and is able to avoid chattering}. Furthermore, for 3rd or lower-order problems, MIM achieves time-optimality. In this section, the switching law and the optimal-trajectory manifold are corresponding to those induced by MIM. Propositions in Section \ref{sec:CostateSystemBehaviorAnalysis} except \modify{for} Conjecture \ref{conjecture:ASL_systembehavior_tangentmarker} still hold true, while theorems in this section might not be true for time-optimal control unless emphasized.

    \subsection{Manifold-Intercept Method}\label{subsec:MIM}
        The key idea of MIM is \modify{to follow a greedy and conservative principle. W}hen the current state vector is ``higher'' (or ``lower'') than the lower-order optimal-trajectory manifold of the terminal state vector, the control \modify{greedily} choose\modify{s} the minimum (or maximum) value to \modify{drive} the system \modify{to} enter the constant velocity phase of minimum (or maximum) velocity, i.e., $\underline{n-1}$ (or $\overline{n-1}$) as quick as possible\modify{, conservatively considering state constraints}. \modify{Once the states enter the lower order optimal-trajectory manifold, the states move along lower order trajectories to reach the terminal states, where the Bellman's principle of optimality \cite{bellman1952theory} is applied.} 

        \begin{definition}\label{def:proper_position_function}
            The \textbf{proper-position function} is defined as $p^*:\dom p^*\to\R$, $\left(\vx_0,\vxf;\vM\right)\mapsto \hat{x}_n$, where
            \begin{align}
                \dom p^*=\left\{\left(\vx_0;\vxf,\vM\right):\problem\left(\vx_0,\vxf;\vM\right)\text{ is feasible},\right.\nonumber\\
                \left.\vx_0,\vxf\in\R^n,\vM\in\R_{++}\times\overline{\R}_{++}^n\right\},
            \end{align}
            \modify{s.t.} $\hat{\vx}_0\modify{\triangleq}\left(x_{0,1},x_{0,2},\dots,x_{0,n-1},\hat{x}_n\right)\in\manifold_{n-1}\left(\vxf,\vM\right)$.
        \end{definition} 

        \begin{figure}[!t]
            \centering
            \includegraphics[width=0.7\columnwidth]{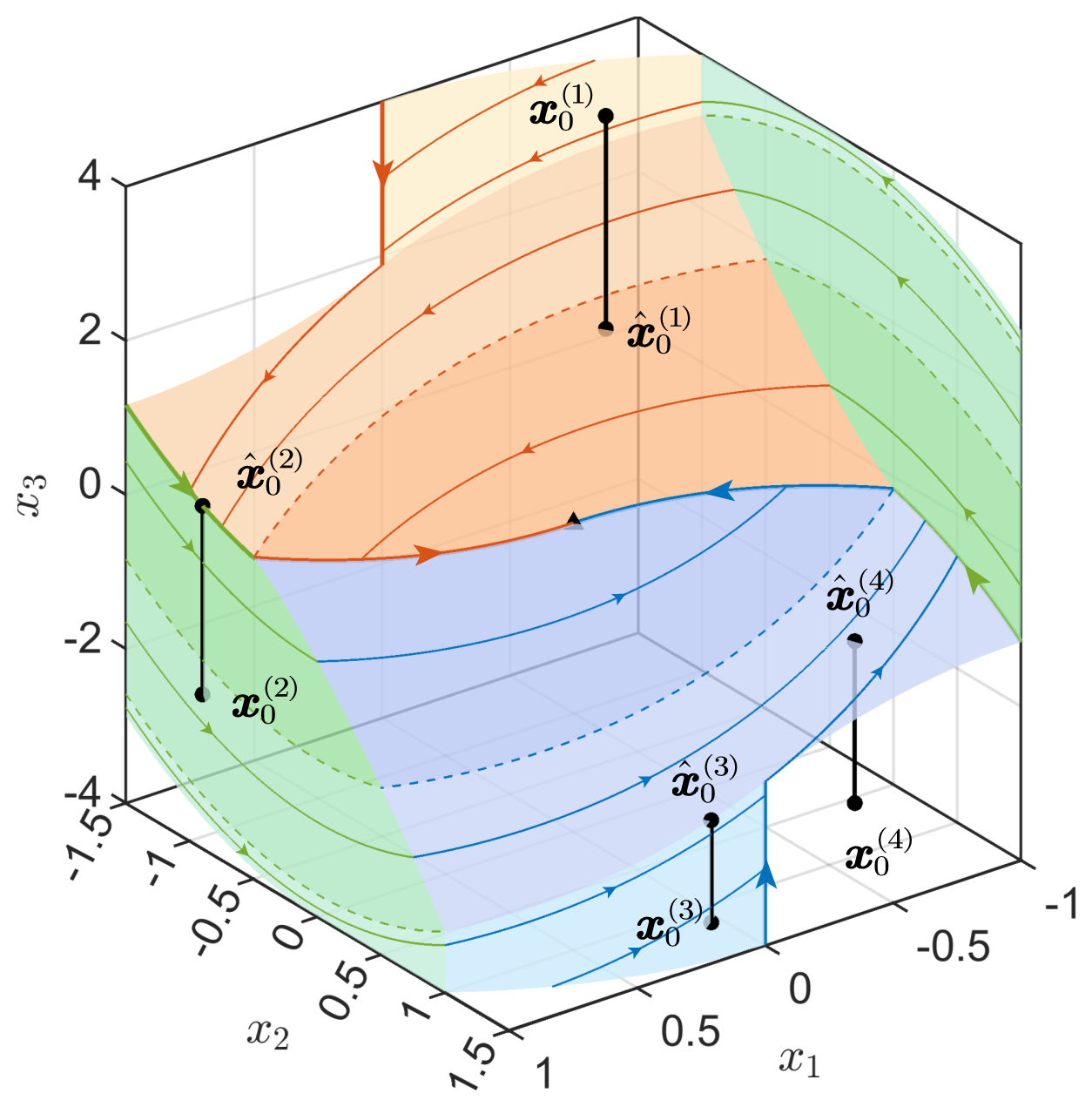}
            \caption{Some examples of proper positions. The parameters and optimal-trajectory manifold is the same as Fig. \ref{fig:demos_manifold}(b). $\forall i=1,2,3,4,\hat{\vx}_0^{\left(i\right)}$ has a proper position and the same $x_1,x_2,x_3$ as th\modify{ose} of $\vx_0^{\left(i\right)}$.}
            \label{fig:properposition_demo}
        \end{figure}

        \begin{definition}\label{def:proper_position}
            For a feasible problem $\problem\left(\vx_0,\vxf;\vM\right)$ of $n$-th order, $\vx_0$ is called \textbf{higher} (or \textbf{lower}) than $\manifold_{n-1}\left(\vxf,\vM\right)$, if $x_{0,n}>$ (or $<$) $p^*\left(\vx_0;\vxf,\vM\right)$. If $x_{0,n}=p^*\left(\vx_0;\vxf,\vM\right)$, then $\vx_0$ is called to \textbf{have a proper position}.
        \end{definition}

        The proper-position function is well-defined according to Proposition \ref{thm:dim_manifold}. In Definition \ref{def:proper_position_function}, $\hat{\vx}_0$ can be regarded as the projection of $\vx_0$ along the $x_n$-axis onto $\manifold_{n-1}\left(\vxf,\vM\right)$. In Definition \ref{def:proper_position}, the relationship between $\vx_0$ and $\manifold_{n-1}\left(\vxf,\vM\right)$ is represented by the positional relationship between $\vx_0$ and its projection, i.e., $\hat{\vx}_0$.

        A 3rd order example is shown in Fig. \ref{fig:properposition_demo}. $\vx_0^{\left(1\right)}$ is higher than $\manifold_{2}\left(\vxf,\vM\right)$, while $\vx_0^{\left(2\right)}$, $\vx_0^{\left(3\right)}$, and $\vx_0^{\left(4\right)}$ are lower than $\manifold_{2}\left(\vxf,\vM\right)$. $\forall i$, $\hat{\vx}_0^{\left(i\right)}\in\manifold_{2}\left(\vxf,\vM\right)$ has a proper position.

        Now MIM is introduced by mathematical induction.

        For the base case, i.e., $n=1$, the optimal control has a trivial analytic expression, i.e.,
        \begin{equation}\label{eq:1order_optimal_control}
            u^*\left(t\right)=M_0\,\sgn\left(x_\f-x_0\right),\,\forall 0\leq t\leq \tf=\frac{\abs{x_\f-x_0}}{M_0}.
        \end{equation}

        Assume that $\forall 1\leq k<n$, the $k$-th order trajectories with arbitrary $\vx_0,\vxf,\vM$ can be planned by MIM. In the case where $n\geq2$ and $M_n=\infty$, $\hat{x}_n=p^*\left(\vx_0;\vxf,\vM\right)$ can be calculated by solving $\problem\left(\vx_0^{1:\left(n-1\right)},\vxf^{1:\left(n-1\right)};\vM^{0:\left(n-1\right)}\right)$. By Definition \ref{def:proper_position_function}, \modify{one} can judge whether $\vx_0$ is higher or lower than $\mathcal{F}_{n-1}\left(\vxf,\vM\right)$ by comparing $x_{0,n}$ and $\hat{x}_n$.
        \begin{enumerate}
            \item If $\vx_0$ has a proper position, according to \eqref{eq:definition_state_optimal_manifold}, $\problem\triangleq\problem\left(\vx_0,\vxf;\vM\right)$ has the same solution as the $\left(n-1\right)$-th order problem $\problem\left(\vx_0^{1:\left(n-1\right)},\vxf^{1:\left(n-1\right)};\vM^{0:\left(n-1\right)}\right)$.
            \item\label{case:MIM_infty_higher} If $x_0$ is higher than $\mathcal{F}_{n-1}\left(\vxf,\vM\right)$, the system tends to obtain a state vector in minimum uniform speed $-M_{n-1}$, and keeps the velocity as $-M_{n-1}$ until entering $\mathcal{F}_{n-1}\left(\vxf,\vM\right)$.
            \begin{enumerate}
                \item Firstly, the system state vector moves along the MIM-trajectory of the $\left(n-1\right)$-th order problem $\problem_1=\problem\left(\vx_0^{1:\left(n-1\right)},-M_{n-1}\ve_{n-1};\vM^{0:\left(n-1\right)}\right)$, i.e., $\vx\left(t\right)=\vx^*\left(t;\problem_1\right)$.
                \item\label{case:MIM_infty_higher_reach_maxv} If $\vx\left(t_\f^*\left(\problem_1\right)\right)$ is still higher than $\mathcal{F}_{n-1}\left(\vxf,\vM\right)$, denote $\hat{x}_{n}=p^*\left(-M_{n-1}\ve_{n-1};\vxf,\vM\right)$. Then, $u\left(t\right)\equiv0$ for $t\in\left[t_\f^*\left(\problem_1\right),t_1+t_\f^*\left(\problem_1\right)\right]$, where $t_1=\dfrac{x_n\left(t_\f^*\left(\problem_1\right)\right)-\hat{x}_{n}}{M_{n-1}}$. Therefore, $\vx\left(t_1\right)=-M_{n-1}\ve_{n-1}+\hat{x}_n\ve_n$ has a proper position. Denote $\problem_2=\problem\left(-M_{n-1}\ve_{n-1},\vxf^{1:\left(n-1\right)};\vM^{0:\left(n-1\right)}\right)$. Let $\vx$ moves along the MIM-trajectory of $\problem_2$, i.e., $\vx\left(t_\f^*\left(\problem_1\right)+t_1+t\right)=\vx^*\left(t;\problem_2\right),0\leq t\leq t_\f^*\left(\problem_2\right)$. Finally, $\vx\left(t_\f^*\left(\problem\right)\right)$ reach\modify{es} $\vxf$, i.e.,
                \begin{equation}
                    t_\f^*\left(\problem\right)=t_\f^*\left(\problem_1\right)+t_1+t_\f^*\left(\problem_2\right).
                \end{equation}
                \item\label{case:MIM_infty_higher_intercept} If $\vx\left(t_\f^*\left(\problem_1\right)\right)$ is lower than $\mathcal{F}_{n-1}\left(\vxf,\vM\right)$, i.e., $\vx\left(t\right)$ enters $\mathcal{F}_{n-1}\left(\vxf,\vM\right)$ at some time $t_2\in\left(0,t_\f^*\left(\problem_1\right)\right)$, then $\vx$ can move along the MIM-trajectory of the $\left(n-1\right)$-th order problem $\problem_3=\problem\left(\vx^{1:\left(n-1\right)}\left(t_2\right),\vxf^{1:\left(n-1\right)};\vM^{0:\left(n-1\right)}\right)$. Finally, $\vx\left(t_\f^*\left(\problem\right)\right)$ reach\modify{es} $\vxf$, i.e.,
                \begin{equation}\label{eq:firstbracket_occur}
                    t_\f^*\left(\problem\right)=t_2+t_\f^*\left(\problem_3\right).
                \end{equation}
            \end{enumerate}
            \item If $x_0$ is lower than $\mathcal{F}_{n-1}\left(\vxf,\vM\right)$, a similar analysis as above applies to this case.
        \end{enumerate}

        If $M_{n-1}=\infty$ above, then the first terminal state vector in Case \ref{case:MIM_infty_higher} is modified as $-M_k\ve_k$ where $k<n-1$ is \modify{the} maximum index satisfying $M_k\not=\infty$. Specifically, if $k=0$, the problem degenerates into a time-optimal problem without state constraints, which can be easily solved by an $n$-th order polynomial system \cite{bartolini2002time}.

        An example of Case \ref{case:MIM_infty_higher_reach_maxv} is $\vx^{\left(1\right)}$ in Fig. \ref{fig:2order_demo_switchinglaw}, where $\vx^{\left(1\right)}$ slides uniformly under a minimum velocity $-M_1$ until $\vx$ enters $\mathcal{F}_{1}\left(\vxf,\vM\right)$. An example of Case \ref{case:MIM_infty_higher_intercept} is $\vx^{\left(2\right)}$ in Fig. \ref{fig:2order_demo_switchinglaw}, where $\vx^{\left(2\right)}$ is intercepted by $\mathcal{F}_{1}\left(\vxf,\vM\right)$ before $\vx^{\left(2\right)}$ reaches $x_1=-M_1$.
        


        The above $n$-th order process meets the Bang-Singular-Bang control law and can be solved by sequential $\left(n-1\right)$-th order problems. Furthermore, the trajectory $\vx\left(t\right)=\vx^*\left(t,\problem\right)$ meets the constraints $\forall k=1,2,\dots,n$, $\abs{x_k}\leq M_k$. Specifically, $\abs{x_n}\leq M_n=\infty$.

        In the case where $n\geq2$ and $M_n<\infty$, the optimal trajectory $\vx\left(t\right)=\vx^*\left(t;\problem_\infty\right)$ is generated based on the above process, where $\problem_\infty=\problem\left(\vx_0,\vxf;\vM^{0:\left(n-1\right)}\right)$.
        \begin{enumerate}
            \item\label{case:MIM_finite_deactivated} If the trajectory $\vx\left(t\right)=\vx^*\left(t;\problem_\infty\right)$ meets the constraint $\abs{x_n}\leq M_n$, then the constraint $\abs{x_n}\leq M_n$ is deactivated. Hence, $\vx\left(t\right)=\vx^*\left(t;\problem_\infty\right)$ is also the optimal trajectory of $\problem=\problem\left(\vx_0,\vxf;\vM\right)$.
            \item\label{case:MIM_finite_activated} If the trajectory $\vx\left(t\right)=\vx^*\left(t;\problem_\infty\right)$ does not meet the constraint, i.e., $\abs{x_n}>M_n$ at some time, then a tangent marker w.r.t. $x_n$ occurs in the optimal problem $\problem$. Considering the dimension of the augmented switching law and according to Theorem \ref{thm:costate}, $\exists \modify{h}\leq \modify{n}$, s.t. $\vx$ reaches a tangent marker $\left(n,\modify{h}\right)$ through an augmented switching law of $\modify{h}$ dimension in optimal time, i.e., by solving \modify{an} $\modify{h}$-th order problem. Assume the system state vector is $\hat{\vx}$ when reaching the tangent marker. Then, the system state vector can move from $\hat{\vx}$ to $\vxf$ in optimal time by solving $\problem\left(\hat{\vx},\vxf;\vM^{0:\left(n-1\right)}\right)$.
        \end{enumerate}


        By mathematical induction, $\forall n\in\N\modify{{}^*}$, trajectories for high-order chain-of-integrators systems with full state constraints $\vM\in\R_{++}\times\overline{\R}_{++}^{n}$ and arbitrary terminal state vector $\vxf\in\R^n$ can be planned by MIM.


    \subsection{Virtual System Behavior}
        The augmented switching law in this section refers to the switching law in MIM with full definite conditions. The virtual system behavior is defined in the case where \eqref{eq:firstbracket_occur} occurs.

        \begin{definition}\label{def:TheFirstBracket}
            In MIM, assume \modify{that} an augmented switching law is $S=\modify{S_1s\left(S_2\right)S_3\in}\mathcal{AS}\left(\problem\left(\vx_0,\vx_\f;\vM\right)\right)$, where $s$ is a system behavior, and $S_i=s_1^{\left(i\right)}s_2^{\left(i\right)}\dots s_{N_i}^{\left(i\right)},i=1,2,3$ are augmented switching laws. $\left(S_2\right)$ is a \textbf{virtual system behavior}, if $\exists \vt^{\left(i\right)}=\left(t_j^{\left(i\right)}\right)_{j=1}^{N_i}\geq\vzero,i=1,2,3$ and $t_1,t_2\geq0$, s.t. 
            \begin{enumerate}
                \item $\vx_0$ moves to $\vxf$ successively passing through $S_1$ by time $\vt^{\left(1\right)}$, $s$ by time $t_1$, and $S_3$ by time $\vt^{\left(3\right)}$. In other words, $S_1sS_3$ with time $\left(\vt^{\left(1\right)},t_1,\vt^{\left(3\right)}\right)$ is the solution of $\problem$.
                \item $\vx_0$ can move successively passing through $S_1$ by time $\vt^{\left(1\right)}$, $s$ by time $t_2$, and $S_2$ by time $\vt^{\left(2\right)}$. In other words, $S_1sS_2$ with time $\vt^{\left(1\right)},t_2,\vt^{\left(2\right)}$ is a feasible solution under constraints \eqref{eq:constraints_state_transformation} and \eqref{eq:constraints_state_system_behavior}.
            \end{enumerate}
        \end{definition}        


        An example is shown in Fig. \ref{fig:demo_first_bracket}, where $\mathcal{S}\left(\problem\right)=\underline{01}\overline{0}\underline{2}\overline{01}\underline{0}\overline{01}\underline{0}\overline{2}\underline{01}\overline{0}$ \modify{in MIM}; hence, $\dim\mathcal{S}\left(\modify{\problem}\right)=6>4$. $\modify{S}=\underline{01}\overline{0}\underline{2}\overline{01}\underline{0}\left(\underline{3}\right)\overline{01}\underline{0}\overline{2}\underline{01}\overline{0}\modify{\in}\mathcal{AS}\left(\problem\right)$; hence, $\dim\modify{S}=4$. As shown in Fig. \ref{fig:demo_first_bracket}(d), $\vx_{i-1}$ moves to $\vx_{i}$ under the corresponding system state vector for time $t_i$, $i\not=8,9$, while $\vx_{7}$ moves to $\vx_{9}$ under $\overline{0}$ for time $t_9$. Furthermore, the ability of $\vx_7$ \modify{to} move to $\vx_{8}$ under $\overline{0}$ for a time period $t_9$ provides definite condition for $\problem_2$, where $\vx_8$ is determined by the virtual system behavior $\left(\underline{3}\right)$, i.e., $x_{8,1}=x_{8,2}=0$, \modify{and} $x_{8,3}=-M_3=-4$. $\vx$ is intercepted by $\mathcal{F}_3\left(\vx_\f,\vM\right)$ at $\vx_7$. 

        \begin{figure}[!t]
            \centering
            \includegraphics[width=\columnwidth]{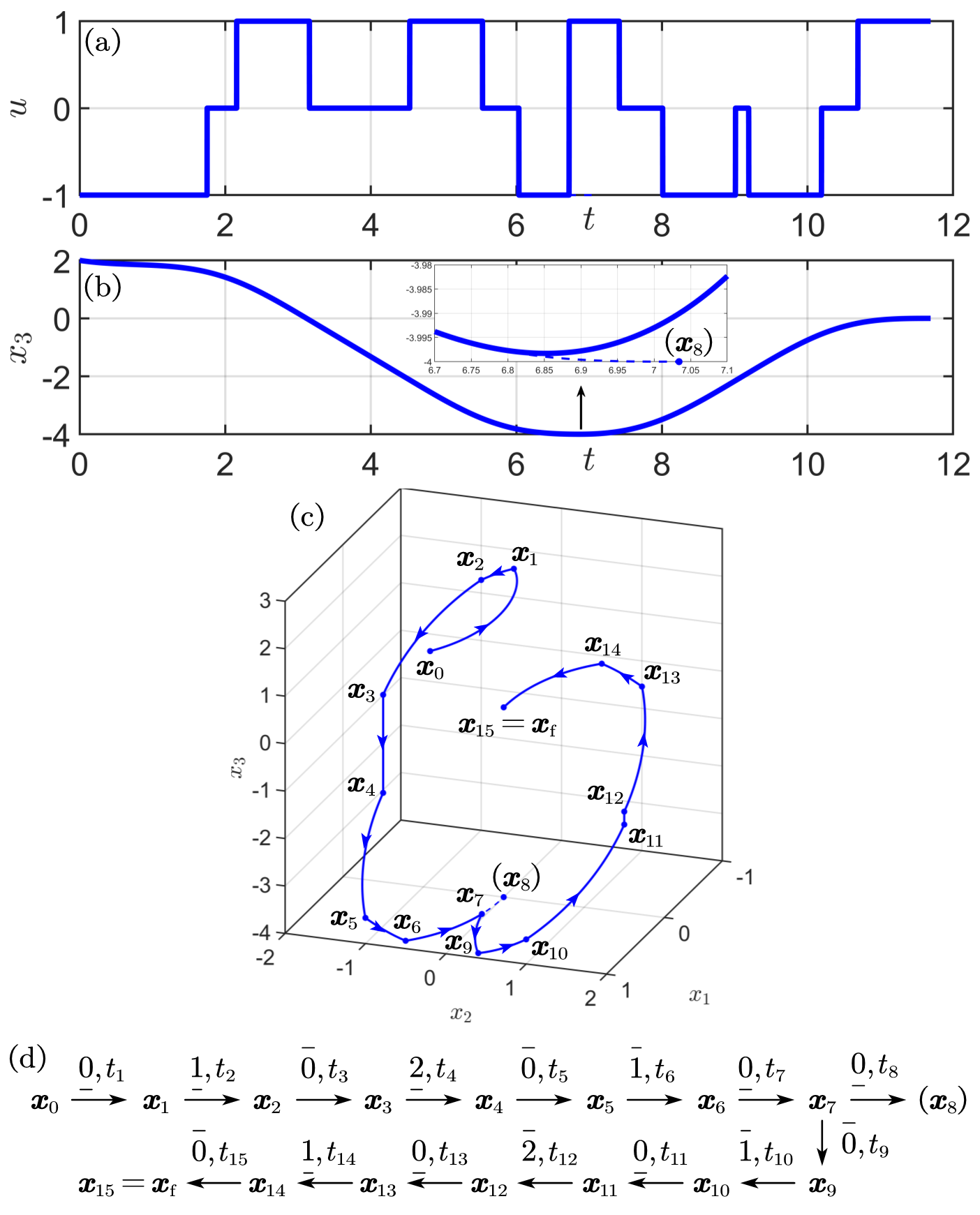}
            \caption{An example of the virtual system behavior. $\vM=\left(1,1,1.5,4,20\right)$\modify{.} $\modify{S}=\underline{01}\overline{0}\underline{2}\overline{01}\underline{0}\left(\underline{3}\right)\overline{01}\underline{0}\overline{2}\underline{01}\overline{0}\modify{\in}\mathcal{AS}\left(\problem\right)$, where $\left(\underline{3}\right)$ is a virtual system behavior. (a) and (b) are the snap and velocity plots\modify{.} (c) is the trajectory of the first three components of system states $\vx^{1:3}$. (d) is the flow charts for $\modify{S}$.}
            \label{fig:demo_first_bracket}
        \end{figure}

        \begin{theorem}\label{thm:TheFirstBracket}
            In MIM, apply \modify{the} notations \modify{from} Definition \ref{def:TheFirstBracket}.
            \begin{enumerate}
                \item\label{thm:TheFirstBracket_dimension} The dimension of $S$ can be calculated like \eqref{eq:DimensionSwitchLaw}, i.e., $\left(S_2\right)$ contributes $\sum_{i=1}^{N_2}\left(1-\abs{s_i^{\left(2\right)}}\right)<0$ dimension to $S$.
                \item\label{thm:TheFirstBracket_s13_sign} $\abs{s_1^{\left(3\right)}}=0$, and $\sgn\left(s\right)=-\sgn\left(s_1^{\left(3\right)}\right)$.
                \item\label{thm:TheFirstBracket_s2end_sign} $\abs{s_{N_2}^{\left(2\right)}}\not=0$, and $\forall k<N_2$, $\abs{s_{N_2}^{\left(2\right)}}>\abs{s_{k}^{\left(2\right)}}$. Furthermore, $\sgn\left(s_{N_2}^{\left(2\right)}\right)=-\sgn\left(s_1^{\left(3\right)}\right)$.
                \item\label{thm:TheFirstBracket_evennumber} Except the tangent marker, $S_2$ has an even number of even numbers, and \eqref{eq:AnalysisOfSign} holds.
                \item\label{thm:TheFirstBracket_simplify} If $\exists N_2'<N_2$, s.t. $\dim S_{2,\left(N_2'+1\right):N_2}=0$, then $\left(S_2\right)$ is equivalent to $\left(S_{2,1:N_2'}\right)$, i.e., $S$ is equivalent to $S_1s\left(S_{2,1:N_2'}\right)S_3$.
            \end{enumerate}
        \end{theorem}
        
        \begin{proof}
            Theorem \ref{thm:TheFirstBracket}-\ref{thm:TheFirstBracket_dimension} holds by the same analysis as \modify{that of} Theorem \ref{thm:DimensionSwitchLaw}.
            Theorem \ref{thm:TheFirstBracket}-\ref{thm:TheFirstBracket_simplify} holds evidently.

            Denote $\problem=\problem\left(\vx_0,\vxf;\vM\right)$. According to Section \ref{subsec:MIM}, the virtual system behavior $\left(S_2\right)$ occurs when the optimal trajectory represented by $S_1sS_2$ is intercepted by a low-dimensional optimal-trajectory manifold $\mathcal{F}$, and $\vx$ moves in $\mathcal{F}$ along an optimal trajectory represented by $S_3$ after that. By induction, we only need to prove Theorem \ref{thm:TheFirstBracket} in the case where the terminal state vector of $S_1sS_2$ has a maximum velocity, i.e., $\abs{s_{N_2}^{\left(2\right)}}=n-1$ for an $n$-th order problem.

            According to \modify{the} discussion above, $S_3\in\mathcal{AF}_{n-1}$; hence, $\abs{s_1^{\left(3\right)}}=0$. Note that $S_1sS_3$ induces $\mathcal{S}\left(\problem\right)$ by removing all brackets, and $s$ as well as $s_1^{\left(3\right)}$ are system behaviors. According to Theorem \ref{thm:AnalysisOfSign}, $\sgn\left(s\right)=-\sgn\left(s_1^{\left(3\right)}\right)$. Theorem \ref{thm:TheFirstBracket}-\ref{thm:TheFirstBracket_s13_sign} holds.

            Note that $S_1sS_2^{1:\left(N_2-1\right)}$ is an augmented switching law with the terminal states $\pm M_{n-1}\ve_{n-1}$. Therefore, $\abs{s_{N_2}^{\left(2\right)}}=n-1\not=0$, and $\abs{s_{N_2}^{\left(2\right)}}>\abs{s_{k}^{\left(2\right)}},\forall k<N_2$. Furthermore, $\abs{s_{1}^{\left(1\right)}}=\abs{s_{N_3}^{\left(3\right)}}=0$ and $\sgn\left(s_{N_3}^{\left(3\right)}\right)=\left(-1\right)^{n-1}\sgn\left(s_{1}^{\left(1\right)}\right)$. Note that $S_3\in\mathcal{AF}_{n-1}$; hence, $\sgn\left(s_{N_3}^{\left(3\right)}\right)=\left(-1\right)^{n-2}\sgn\left(s_{1}^{\left(3\right)}\right)$. Assume $\vx_0$ is higher than $\mathcal{F}_{n-1}\left(\vxf,\vM\right)$. Then, $\sgn\left(s_{N_2}^{\left(2\right)}\right)=\sgn\left(s_{1}^{\left(1\right)}\right)=-1$; hence, $\sgn\left(s_{N_2}^{\left(2\right)}\right)=-\sgn\left(s_{N_3}^{\left(3\right)}\right)$. If $\vx_0$ is lower than $\mathcal{F}_{n-1}\left(\vxf,\vM\right)$, then $\sgn\left(s_{N_2}^{\left(2\right)}\right)=-\sgn\left(s_{N_3}^{\left(3\right)}\right)=1$. Therefore, Theorem \ref{thm:TheFirstBracket}-\ref{thm:TheFirstBracket_s2end_sign} holds.

            According to \modify{the} discussion above, \eqref{eq:AnalysisOfSign} holds for $S_1sS_2$, $S_1sS_3$, and $s_{N_2}^{\left(2\right)}S_3$. In other words, the sign of $S_1s\left(S_2\right)S_3$ is the same as $S_1sS_2S_3$ and satisfies \eqref{eq:AnalysisOfSign}. Note that $\sgn\left(s\right)=-\sgn\left(s_1^{\left(3\right)}\right)=\sgn\left(s_{N_2}^{\left(2\right)}\right)$; hence, $S_2$ has an even number of even numbers. Therefore, Theorem \ref{thm:TheFirstBracket}-\ref{thm:TheFirstBracket_evennumber} holds.
        \end{proof}

        As an example for Theorem \ref{thm:TheFirstBracket}-\ref{thm:TheFirstBracket_simplify}, $\underline{01}\overline{0}\underline{2}\left(\overline{01}\underline{03}\right)\overline{01}\underline{0}\overline{2}\underline{01}\overline{0}$ is equivalent to $\underline{01}\overline{0}\underline{2}\overline{01}\underline{0}\overline{2}\underline{01}\overline{0}$. The tangent marker also \modify{exhibits certain} properties in MIM.

        
        \begin{theorem}\label{thm:TangentMarkerBracket_MIM}
            Apply \modify{the} notations \modify{from} Theorem \ref{thm:TangentMarkerBracket}. In MIM,
            \begin{enumerate}
                \item\label{thm:TangentMarker_MIM_sign} $\sgn\left(s_{N_1}^{\left(1\right)}\right)=\sgn\left(s_{1}^{\left(2\right)}\right)=\sgn\left(s\right)$.
                \item\label{thm:TangentMarkerBracket_less} $\modify{h}<\abs{s}$. \modify{Moreover, $h$ is even.}
            \end{enumerate}
        \end{theorem}

        \begin{proof}
            Similar to the proof of Theorem \ref{thm:TheFirstBracket}, we only need to prove the case where $\abs{s}=n$ for an $n$-th problem. Assume $\vx_0$ is higher than $\mathcal{F}_{n-1}\left(\vxf,\vM\right)$ and the tangent marker occurs. According to Theorem \ref{thm:TangentMarkerBracket_MIM}, $\abs{s_{N_1}^{\left(1\right)}}=\abs{s_{1}^{\left(2\right)}}=0$. According to Section \ref{subsec:MIM}, $\sgn\left(s_1\right)=\left(-1\right)^{d-2}\sgn\left(s_1^{\left(1\right)}\right)=-1$, and $\sgn\left(s\right)=-1$. When $\vx$ reaches $\left(s,d\right)$ as $\tilde{\vx}$, $\tilde{\vx}$ is still higher than $\mathcal{F}_{n-1}\left(\vxf,\vM\right)$, while $s_2S_2$ is the augmented switching law from $\tilde{\vx}$ to $\vxf$. Therefore, $\abs{s_2}=0$, $\sgn\left(s_1\right)=-1$. Applying a similar analysis to the case where $\vx_0$ is lower than $\mathcal{F}_{n-1}\left(\vxf,\vM\right)$, Theorem \ref{thm:TangentMarkerBracket_MIM}-\ref{thm:TangentMarker_MIM_sign} holds.

            Assume $\modify{h}=\abs{s}$ and the system state vector reaches $\left(s,\modify{h}\right)$ at $t$. Then, $u$ is continuous at $t$ and $u\left(t\right)=\sgn\left(s\right)M_0$ by Theorem \ref{thm:TangentMarkerBracket_MIM}-\ref{thm:TangentMarker_MIM_sign}. However, according to Section \ref{sub2sec:junction_TouchMaxTimePoint}, $\sgn\left(u\left(t\right)\right)=-\sgn\left(s\right)$ \modify{since $\abs{x_{\abs{s}}}\leq M_{\abs{s}}$}, which causes contradiction. Therefore, $\modify{h}<\abs{s}$\modify{}{; hence, $h$ is even}.
        \end{proof}

        Theorems \ref{thm:TheFirstBracket} and \ref{thm:TangentMarkerBracket_MIM} indicate that signs of elements of an augmented switching law \modify{of} length $N$ \modify{are} determined by the sign of the last element, which also reduces the possible signs of all elements from $2^N$ to $2$. Then, augmented switching laws can be enumerated fully, as reasoned in Section \ref{subsec:FullEnumerationASL}.

        

    \subsection{Full Enumeration of Augmented Switching Laws in MIM}\label{subsec:FullEnumerationASL}
        
        \begin{definition}\label{def:ASL_representation}
            In MIM, assume $\manifold_{k}\left(\vxf,\vM\right)\not=\varnothing$ in \eqref{eq:definition_state_optimal_manifold}. The \textbf{augmented-switching-law representation} of $\manifold_{k}\left(\vxf,\vM\right)$ is
            \begin{align}
                \mathcal{AF}_k\left(\vxf,\vM\right)=\bigcup\left\{S\modify{\in}\mathcal{AS}\left(\problem\left(\vx_0,\vxf;\vM\right)\right):\right.\nonumber\\
                \left.\vx_0\in\manifold_{k}\left(\vxf,\vM\right),\dim S=k\right\}.
            \end{align}
        \end{definition}

        
        $\forall \vx_0\in\manifold_{k}\left(\vxf,\vM\right),\problem=\problem\left(\vx_0,\vxf;\vM\right),S\modify{\in}\mathcal{AS}\left(\problem\right)$, if $\dim S=k'<k$, $\left(k-k'\right)$ zero can be added before $S$ to increase the dimension to $k$. For example, if $S=\overline{0}\underline{0}$ in Fig. \ref{fig:demos_manifold}(b), the switching law can be seen as $\underline{0}\overline{0}\underline{0}$, and the motion time solved \modify{for} the first $\underline{0}$ is 0. If some augmented switching laws do not end with 0, then some \modify{zeros} can also be added after the augmented switching law with motion time 0.

        
        
        \begin{definition}\label{def:AF}
            In MIM,
            \begin{equation}
                \mathcal{AF}_n\triangleq\bigcup_{\substack{\vxf\in\R^n\\\vM\in\R_{++}\times\overline{\R}_{++}^{n}}}\mathcal{AF}_n\left(\vxf,\vM\right)
            \end{equation}
            is called the $n$\textbf{-th order set of augmented switching laws}.
        \end{definition}
        
        \modify{By} Section \ref{subsec:MIM}, $\mathcal{AF}_n$ can be constructed as follows:
        \begin{theorem}\label{thm:AF_constructed}
            In MIM, $\mathcal{AF}_1=\{0\}$. $\forall n\geq 1,\mathcal{AF}_{n+1}\supset F_{n,1}\cup F_{n,2}\cup F_{n,3}$, where
            \begin{align}
                F_{n,1}=\{&S_1sS_2:S_1,S_2\in\mathcal{AF}_{n},\abs{s}=n\},\\
                F_{n,2}=\{&S_1\left(S_2s\right)S_3:S_1S_2,S_3\in\mathcal{AF}_{n},\abs{s}=n,\nonumber\\
                &S_2s\text{ satisfies Theorem \ref{thm:TheFirstBracket}-\ref{thm:TheFirstBracket_evennumber}}\},\\
                F_{n,3}=\{&S_1\left(s,d\right)S_2:S_1\in\mathcal{AF}_{d},S_2\in F_{n,1}\cup F_{n,2},\nonumber\\
                &\abs{s}=n+1,d<\abs{s}\text{ is even}\}.
            \end{align}
            The signs of system behaviors, tangent markers, and virtual system behaviors satisfy Theorems \ref{thm:AnalysisOfSign}, \ref{thm:TangentMarkerBracket}, \ref{thm:TheFirstBracket}, and \ref{thm:TangentMarkerBracket_MIM}.
        \end{theorem}

        \begin{proof}
            According to Theorems \ref{thm:TheFirstBracket}, \ref{thm:TangentMarkerBracket_MIM}, and discussion in Section \ref{subsec:MIM}, Theorem \ref{thm:AF_constructed} holds evidently.
        \end{proof}

        For example, $\mathcal{AF}_2$=\{00, 010\}. $\mathcal{AF}_3$=\{000, 0010, 0100, 01010, 00200, 002010, 010200, 0102010, 00(3,2)000, 00(3,2)0010, 00(3,2)0100, 00(3,2)01010, 00(3,2)00200, 00(3,2)002010, 00(3,2)010200, 00(3,2)0102010, 010(3,2)000, 010(3,2)0010, 010(3,2)0100, 010(3,2)01010, 010(3,2)00200, 010(3,2)002010, 010(3,2)010200, 010(3,2)0102010\}. $\mathcal{AF}_4$= \{0000, 010(4,2)0102010(3)0102010, $\dots$\}. $\mathcal{AF}_n$ for $n\geq 4$ can be obtained by Theorem \ref{thm:AF_constructed}. \modify{Moreover, previous works on time-optimal jerk-limited trajectories like \cite{he2020time,kroger2011opening,berscheid2021jerk} are all devoted to studying the first 8 switching law through model-based analysis, failing to completely enumerate $\mathcal{AF}_3$.} 

        All augmented switching laws for any trajectory planning problem for high-order chain-of-integrators systems have been fully enumerated so far. In other words, the suboptimal solution of problem \eqref{eq:optimalproblem} with arbitrary initial states, terminal states, and full state constraints can be solved even by enumeration.
        
        
        

            
    \subsection{Calculation of Motion Time and Feasibility Verification}\label{subsec:calculation}
        According to Definition \ref{def:ASL_representation} and \ref{def:AF}, $\forall S\in\mathcal{AF}_n$, $\dim S=n$. Therefore, the motion time can be determined by a given initial state vector $\vx_0$ and an augmented switching law $S$. Denote $\vf:\R^n\times\R\times\R\to\R^n$ whose $k$-th component is
        \begin{equation}\label{eq:polynomial}
            f_k\left(\vx,u,t\right)=\frac{1}{k!}ut^k+\sum_{i=1}^{k}\frac{1}{i!}x_{k-i}t^i.
        \end{equation}
        $\vy=\vf\left(\vx,u,t\right)$ means system state vector moves from $\vx$ to $\vy$ under a constant control $u$ \modify{over} a period of time $t$. Next, describe equations induced by a given $S=s_1s_2\dots s_N\in\mathcal{AF}_n$. Assume the system state vector reaches $\vx_{m-1}$ before $s_m$. The control and time in stages of $s_m$ are $u_m$ and $t_m$, respectively. $u_m$ and $\sgn\left(s_m\right)$ can be calculated by Theorem \ref{thm:AnalysisOfSign} and Theorem \ref{thm:TheFirstBracket}.

        If $s_{m}$ and $s_{m+1}$ are system behaviors, or \modify{if} $s_{m}$ and $s_{m+1}$ are the virtual system behaviors, or \modify{if} $m=N$, then
        \begin{equation}\label{eq:equations_systembehavior}
            \begin{dcases}
                \vx_m=\vf\left(\vx_{m-1},u_m,t_m\right),\\
                x_{m,\abs{s_m}}=\sgn\left(s_m\right)M_{\abs{s_m}},\text{ if }\abs{s_m}\not=0,\\
                x_{m,k}=0,\text{ if }\abs{s_m}\not=0,k<\abs{s_m}.
            \end{dcases}
        \end{equation}

        For $s_l\left(s_{l+1}\dots s_{r-1}\right)s_r$, $s_l$ and $s_r$ are system behaviors, while $\left(s_{l+1}\dots s_{r-1}\right)$ is a virtual system behavior. According to Theorem \ref{thm:TheFirstBracket}-\ref{thm:TheFirstBracket_s13_sign}, $s_r=0$. Then,
        \begin{equation}\label{eq:equations_firstbrackt}
            \begin{dcases}
                \vx_{l}=\vf\left(\vx_{l-1},u_{l},t_{l}\right),\\
                \vx_{r-1}=\vf\left(\vx_{l-1},u_{l},t_{r-1}\right),\\
                \vx_{r}=\vf\left(\vx_{r-1},u_{r},t_{r}\right),\\
                x_{l,\abs{s_l}}=\sgn\left(s_l\right)M_{\abs{s_l}},\text{ if }\abs{s_l}\not=0,\\
                x_{l,k}=0,\text{ if }\abs{s_l}\not=0,k<\abs{s_l}.
            \end{dcases}
        \end{equation}

        For $s_{m-1}\left(s_m,d_m\right)s_{m+1}$, $u_{m-1}=u_{m+1}$, and $u_m$ is not defined. Then,
        \begin{equation}\label{eq:equations_secondbrackt}
            \begin{dcases}
                \vx_{m}=\vf\left(\vx_{m-2},u_{m-1},t_{m-1}\right),\\
                x_{m,\abs{s_m}}=\sgn\left(s_m\right)M_{\abs{s_m}},\\
                x_{m,\abs{s_m}-k}=0,\,k<d_m.
            \end{dcases}
        \end{equation}

        Finally, $\vx_N$ is substituted \modify{with} $\vxf$.

        \begin{algorithm}[!t]
            \caption{Trajectory planning for high-order chain-of-integrators systems by MIM.}
            \label{alg:CalMIM}
            \begin{algorithmic}
                \REQUIRE $n\in\N\modify{{}^*}$, $\vx_0$, $\vxf\in\R^n$, $\vM\in\R_{++}\times\overline{\R}_{++}^{n}$.
                \ENSURE Optimal control $u=u^*\left(t\right)$ of problem \eqref{eq:optimalproblem}.
                \IF{$n=1$}
                    \STATE Solve the problem by \eqref{eq:1order_optimal_control} and \textbf{return}.
                \ENDIF
                \IF{$\vx_0$ is higher than $\mathcal{F}_{n-1}\left(\vx_\f;\vM^{0:\left(n-1\right)}\right)$}
                    \STATE Obtain $u=\hat{u}^*\left(t\right)$ by $\problem\left(-\vx_0,-\vxf,\vM\right)$.
                    \STATE \textbf{return} $u^*\left(t\right)=-\hat{u}^*\left(t\right)$.
                \ENDIF
                \IF{$M_1=M_2=\dots=M_{n-1}=\infty$}
                    \STATE Obtain $\vx^*\left(t\right)$, $u^*\left(t\right)$, $t_\f$ by solving $n$ tandem \eqref{eq:polynomial}.
                \ELSE
                    \STATE $m\leftarrow\arg\max\{k<n:M_k<\infty\}$.
                    \STATE Obtain $\vx^{\left(1\right)}\left(t\right),t_{\f1}$ by $\problem\left(\vx_0^{1:\left(n-1\right)},M_{m}\ve_{m},\vM^{0:m}\right)$.
                    \IF{$x_n^{\left(1\right)}\left(t_{\f1}\right)$ is lower than $\mathcal{F}_{n-1}\left(\vx_\f;\vM^{0:\left(n-1\right)}\right)$}
                        \STATE Obtain $\vx^{\left(2\right)}\left(t\right)$ by $\problem\left(M_{m}\ve_{m},\vx_\f^{1:\left(n-1\right)},\vM^{0:m}\right)$.
                        \STATE Obtain $\vx^*\left(t\right)$ by connecting $\vx^{\left(1\right)}$, $x_m\equiv M_m$, $\vx^{\left(2\right)}$.
                    \ELSE
                        \STATE Solve $t_1$ when $\vx$ enters $\mathcal{F}_{n-1}\left(\vx_\f;\vM^{0:\left(n-1\right)}\right)$.
                        \STATE Obtain $\vx^{\left(3\right)}\left(t\right)$ by \begin{footnotesize}
                            $\problem\left(\vx^{1:\left(n-1\right)}\left(t_1\right),\vx_\f^{1:\left(n-1\right)},\vM^{0:m}\right)$
                        \end{footnotesize}.
                        \STATE Obtain $\vx^*\left(t\right)$ by connecting $\vx^{\left(1\right)}$, $\vx^{\left(3\right)}$.
                    \ENDIF
                \ENDIF
                \IF{$M_n<\infty$ \textbf{and} $\exists t\in\left(0,t_{\f}\right)$, $\abs{x_n^*\left(t\right)}>M_n$}
                    \FOR{$d\leftarrow 2,4,\dots,2\left\lfloor\frac{n-1}{2}\right\rfloor$}
                        \FOR{$S\in \mathcal{AF}_d$}
                            \STATE Obtain $\vx=\vx^{\left(0\right)}\left(t\right),t_{\f0}$ by $S$, where $\vx$ moves from $\vx_0$ to $\left(n,d\right)$;
                            \STATE Obtain $\vx=\vx^{\left(4\right)}\left(t\right)$ by $\problem\left(\vx\left(t_{\f0}\right),\vx_\f,\vM\right)$.
                            \STATE Update the best feasible trajectory $\vx^*\left(t\right)$ by connecting $\vx^{\left(0\right)},\vx^{\left(4\right)}$.
                        \ENDFOR
                    \ENDFOR
                \ENDIF
            \end{algorithmic}
        \end{algorithm}

        Assume $S\in\mathcal{AF}_n$ has $N_0$ system behaviors, $N_1$ virtual system behaviors, and $N_2$ tangent markers. Then, the number of equations equals the number of variables, i.e., $\left(N_0+N_1\right)\left(n+1\right)-n$. $N_2$ satisfies $\dim S=n$.

        Feasibility verification is trivial. On one hand, $\forall m=1,2,\dots,N$, $t_m\geq0$. On the other hand, $\forall m=1,2,\dots,N$, $\forall t\in\left[0,t_m\right]$, $\forall k=1,2,\dots,n$, $\abs{x_{m,k}\left(t\right)}\leq M_k$, where $x_{m,k}\left(t\right)$ is the system state after entering $s_m$ for a period of time $t$. The latter feasibility condition can be verified by $x_k\left(t\right)=f_k\left(\vx_{p},u_m,t\right)$, where $\vx_p$ is the previous system state vector of $s_m$ according to \modify{the} discussion above. Furthermore, checking $\abs{x_{m,k}\left(t\right)}\leq M_k$ at stationary points is enough \cite{doeser2020invariant}.

        Based on MIM in Section \ref{subsec:MIM} and enumeration of $\mathcal{AF}_n$ in Section \ref{subsec:FullEnumerationASL}, a trajectory planning algorithm for high-order chain-of-integrators systems with arbitrary initial states, terminal states, and full state constraints \modify{is} developed in Algorithm \ref{alg:CalMIM}.

\section{Numerical Results and Discussion}\label{sec:results}

    \subsection{Simulation Setup}
        \textbf{Baselines.} To verify the performance of the proposed MIM method, i.e., Algorithm \ref{alg:CalMIM}, simulation experiments for trajectories are conducted. The baselines are as follows.
        \begin{itemize}
            \item \textbf{Ruckig} \cite{berscheid2021jerk} in community version: a jerk-limited time-optimal trajectory solver without position constraints.
            \item \textbf{SOCPs} \cite{leomanni2022time}: a time-optimal control method based on solving sequential convex second-order cone problems. The SOCPs method is achieved by Gurobi \cite{gurobi}.
            \item \textbf{Yop} \cite{leek2016optimal}: a MATLAB toolbox for numerical optimal control problems based on CasADi \cite{andersson2019casadi} by direct methods.
        \end{itemize}
        The control period of Ruckig is set as 1 ms. The number of time points is set as 500 and 150 in SOCPs and Yop, respectively. In practice, if the number of time points is set more than 600 and 200 in SOCPs and Yop, respectively, then the computational time will increase \modify{significantly}.

        \textbf{Metrics.} For trajectory planning methods, the computational efficiency, the computational error, and the trajectory quality are significant performance metrics.
        \begin{itemize}
            \item The computational time $T_\mathrm{c}$. All experiments are conducted \modify{in MATLAB 2021b} on a computer with an AMD Ryzen 7 5800H @ 3.20 GHz processor. 
            \item The error of the terminal states $E_\mathrm{s}$. For a solved final state vector $\hat{\vx}_\f$, $E_\mathrm{s}$ is defined in normalization as
            \begin{equation}
                E_\mathrm{s}\triangleq\sqrt{\sum_{k=1}^{n}\left(\frac
                {x_{\f k}-\hat{x}_{\f k}}{M_k}\right)^2}.
            \end{equation}
            \item The success rate to obtain a feasible solution $R_\mathrm{s}$. A result is determined to be successful if the states along the planned trajectories are feasible and $E_\mathrm{s}\leq 0.1$.
            \item The normalized mean-squared error (MSE) $E_\mathrm{m}$ between the solved $u=u\left(t\right)$ and the Bang-Singular-Bang control law \cite{he2020time} is defined to describe the trajectory quality. According to Proposition \ref{prop:bang_zero_bang_law}, the optimal control $u^*\left(t\right)\in\left\{0,M_0,-M_0\right\}$ almost everywhere. Define
            \begin{equation}
                E_\mathrm{m}\triangleq\sqrt{\dfrac{4}{M_0^2t_\f}\int_{0}^{t_\f}u\left(t\right)^2\wedge\left(\abs{u\left(t\right)}-M_0\right)^2\,\mathrm{d}t}.
            \end{equation}
            Then, the smaller $E_\mathrm{m}$ is, the closer $u=u\left(t\right)$ is to a Bang-Singular-Bang control law, i.e., the better trajectory quality is. Specifically, $E_\mathrm{m}\in\left[0,1\right]$. Furthermore, $E_\mathrm{m}=0$ if and only if $\forall t\in\left[0,t_\f\right]$, $u\left(t\right)\in\left\{0,M_0,-M_0\right\}$ satisfies the Bang-Singular-Bang control law, while $E_\mathrm{m}=1$ if and only if $\forall t\in\left[0,t_\f\right]$, $\abs{u\left(t\right)}=\frac{1}{2}M_0$.
            \item The normalized total variation \cite{stein2009real} $T_\mathrm{v}$ of the solved control is defined to describe the stability of the trajectory. For $u\left(\frac{k}{n}t_\f\right)=u_k$ with $\left(n+1\right)$ waypoints, define
            \begin{equation}
                T_\mathrm{v}\triangleq\frac{1}{2nM_0}\sum_{k=1}^{n}\abs{u_k-u_{k-1}}.
            \end{equation}
            Then, as $T_\mathrm{v}$ decreases, the trajectory exhibits increased oscillations, i.e., the stability of the trajectory decreases. Specifically, $T_\mathrm{v}\in\left[0,1\right]$. Furthermore, $T_\mathrm{v}=0$ if and only if $u_k\equiv\mathrm{const}$, while $T_\mathrm{v}=1$ if and only if $u_k=\left(-1\right)^{k}M_0$ or $u_k=\left(-1\right)^{k-1}M_0$.
        \end{itemize}

\begin{figure*}[!t]
    \centering
    \includegraphics[width=\textwidth]{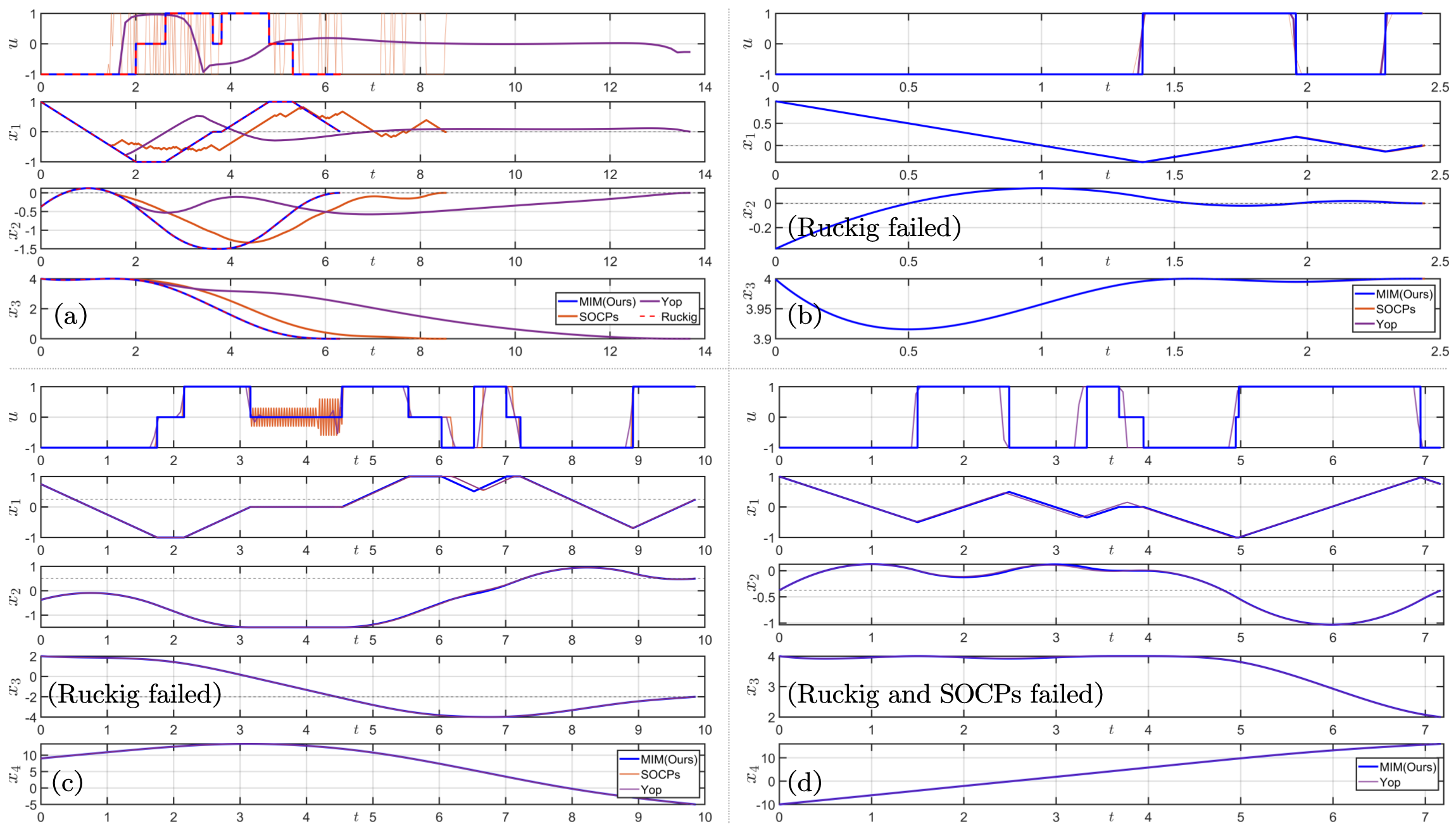}
    \caption{Some trajectories planned by the proposed and baselines. (a) A jerk-limited trajectory represented by $\underline{0}\underline{1}\overline{0}\underline{2}\overline{0}\overline{1}\underline{0}$. $\vx_0=\left(1,-0.375,4\right)$, $\vxf=\left(0,0,0\right)$, $\vM=\left(1,1,1.5,4\right)$. (b) A jerk-limited trajectory represented by $\underline{0}\overline{0}\left(\overline{3},2\right)\overline{0}\underline{0}\overline{0}$. $\vx_0=\left(1,-0.375,3.999\right)$, $\vxf=\left(0,0,4\right)$, $\vM=\left(1,1,1.5,4\right)$. (c) A snap-limited trajectory represented by $\underline{0}\underline{1}\overline{0}\underline{2}\overline{0}\overline{1}\underline{0}(\underline{3})\overline{0}\overline{1}\underline{0}\overline{0}$. $\vx_0=\left(0.75,-0.375,2,9\right)$, $\vxf=\left(0.25,0.5,-2,-5\right)$, $\vM=\left(1,1,1.5,4,20\right)$. (d) A snap-limited trajectory represented by $\underline{0}\overline{0}(\overline{3},2)\overline{0}\underline{0}\overline{0}\overline{3}\underline{0}\underline{1}\overline{0}\underline{0}$. $\vx_0=\left(1,-0.375,4,-10\right)$, $\vxf=\left(0.75,-0.375,2,16\right)$, $\vM=\left(1,1,1.5,4,20\right)$. The trajectories of MIM and Ruckig almost coincide in (a), while $x_3,x_4$ of all methods look to coincide.
    }
    \label{fig:compare_trajectory}
\end{figure*}

\begin{figure*}[!t]
    \centering
    \includegraphics[width=\textwidth]{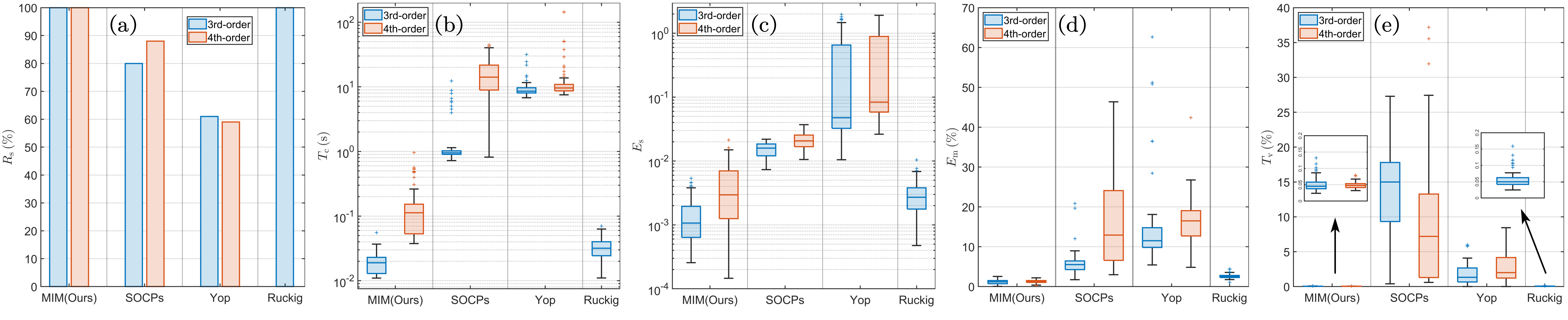}
    \caption{Quantitative results of 100 3rd order trajectories and 100 4th order trajectories with random initial states and terminal states. (a) Success rate $R_\mathrm{s}$. (b) Computational time $T_\mathrm{c}$. (c) Error of terminal states $E_\mathrm{s}$. (d) Normalized MSE $E_\mathrm{m}$ between input and the Bang-Singular-Bang control law. (e) Normalized total variation $T_\mathrm{v}$ of the planned control $u$. Ruckig is only applied in 3rd order trajectories because of disability to plan snap-limited trajectories. $y$-axes in (b) and (c) are in logarithmic scales.}
    \label{fig:Boxplot_34order}
\end{figure*}

    \subsection{Numerical Results}
        Several trajectories planned by the proposed MIM and baselines are shown in Fig. \ref{fig:compare_trajectory}, and the quantitative results of 100 jerk-limited trajectories and 100 snap-limited trajectories are shown in Fig. \ref{fig:Boxplot_34order}. Among them, 3rd order trajectories in Fig. \ref{fig:Boxplot_34order} contain no tangent markers, since Ruckig in community version is not able to plan trajectories with position constraints. Furthermore, tangent markers occur in 3rd order trajectories only if the initial position and the terminal position are both close to boundaries, which occurs with a low probability.
        
        It is noteworthy that the time-optimal problem \eqref{eq:optimalproblem} with an order $n\geq3$ is non-convex if it is solved directly in \modify{discrete time} \cite{debrouwere2013time}. Hence, discrete methods like Yop might fail to obtain an optimal trajectory, as shown in Fig. \ref{fig:compare_trajectory}(a). Though SOCPs successfully transforms problem \eqref{eq:optimalproblem} into sequential convex problems, SOCPs might fail to obtain a feasible solution during iteration in some cases, as shown in Fig. \ref{fig:compare_trajectory}(d). Furthermore, trajectories planned by Yop and SOCPs do not meet constraints \eqref{eq:optimalproblem_x_constraint}, causing failure as well. As shown in Fig. \ref{fig:Boxplot_34order}(a), Yop and SOCPs has a limited success rate, while Ruckig can plan all 3th order trajectories. Supported by \modify{the established} theory, the proposed MIM succeeds in planning all randomly selected trajectories in 3rd and 4th order.

        As shown in Fig. \ref{fig:Boxplot_34order}(c) and (d), the proposed MIM outperforms all baselines on computational time and computational accuracy. It can be observed that the proposed MIM achieves a computational time reduction of approximately 2 orders of magnitude compared to SOCPs and Yop, while improving computational accuracy by 90.8\% (79.0\%), 99.6\% (99.0\%), and 52.1\% compared to SOCPs, Yop, and Ruckig, respectively, for 3rd order (4th order) problem\modify{s} \modify{on} average. Specifically, MIM is not based on discrete time; hence, its computational efficiency keeps constant no matter what the control period is. On the contrary, as a real-time method, Ruckig requires longer computational time with a higher control frequency, but Ruckig is able to achieve a short computational time than the proposed MIM if its control frequency is set lower than 100 Hz, which is not conducive to precise motion control. Furthermore, the upper bound of the computational accuracy of MIM is determined only by the control frequency.

        $E_\mathrm{m}$ and $T_\mathrm{v}$ shown in Fig. \ref{fig:Boxplot_34order}(d) and (e) describe the trajectory quality. $E_\mathrm{m}$ and $T_\mathrm{v}$ of MIM and Ruckig are evidently lower than those of SOCPs and Yop. Therefore, the proposed MIM and Ruckig achieve much higher trajectory quality. Compared with SOCPs and Yop, trajectories generated by MIM are strict Bang-Singular-Bang. In addition to necessary switching, the input $u\left(t\right)$ of the proposed MIM does not exhibit oscillations like those of SOCPs, as shown in Fig. \ref{fig:compare_trajectory}(a) and (c). When the trajectory is short enough like Fig. \ref{fig:compare_trajectory}(b), SOCPs has a short step width and can avoid oscillations. Quantitatively, MIM achieves a normalized MSE of control $E_\mathrm{m}$ that is 80.78\% (92.3\%), 91.9\% (92.2\%), and 57.2\% lower than that of SOCPs, Yop, and Ruckig, respectively, for 3rd order (4th order) problems \modify{on} average. Taking oscillations into consideration, MIM achieves a normalized total variation of control $T_\mathrm{v}$ that is 99.6\% (99.5\%), 97.1\% (99.5\%), and 10.1\% lower than that of SOCPs, Yop, and Ruckig, respectively, for 3rd order (4th order) problems \modify{on} average.



    \subsection{Discussion}
        By comparing Algorithm \ref{alg:CalMIM} of 3rd order and Ruckig in community version \cite{berscheid2021jerk} with 1-DOF, it can be examined that the proposed MIM plans the same trajectory with Ruckig for a given 3rd order problem. \modify{I}n conjunction with discussion in Section \ref{sec:CostateSystemBehaviorAnalysis}, the proposed \modify{MIM} can plan strictly time-optimal trajectories in 3rd order and lower-order with full state constraints. \modify{Furthermore, it can be proved by induction that MIM plans optimal trajectories for $n$-th order problems when only $\abs{u}\leq M_0$ and $\abs{x_1}\leq M_1$ are active.}

        For 4th and higher-order problems, MIM plans quasi-optimal trajectories since the virtual system behavior in Definition \ref{def:TheFirstBracket} does not exist in time-optimal trajectories. \modify{Some more complex behaviors like chattering can occur in 4th and higher-order problems. Although the proposed MIM does not consider the chattering phenomenon in problem \eqref{eq:optimalproblem}, MIM achieves terminal time close to optimal solution, and avoids infinite times of chattering. For example, a set of typical kinematic parameters for the ultra-precision wafer stage is given in \cite{li2015data}, where $M_0=64\,000\,\mathrm{m/s^4}$, $M_1=790\,\mathrm{m/s^3}$, $M_2=10\,\mathrm{m/s^2}$, $M_3=0.25\,\mathrm{m/s}$, and $M_4=0.02\,\mathrm{m}$. Consider a position-to-position trajectory, i.e., $\vx_0=-M_4\ve_4$ and $\vx_\f=M_4\ve_4$. Then, MIM plans a suboptimal trajectory with terminal time $t_{\f,\mathrm{MIM}}\approx0.2100\,\mathrm{s}$. According to \cite{wang2024part2}, the optimal trajectory with chattering achieves terminal time $t_{\f}^*$, where $t_{\f,\mathrm{MIM}}-t_{\f}^*\approx29.2\,\mathrm{ns}$. So MIM achieves a relative error of $0.14\%$ in terminal time compared to the optimal solution and is able to avoid chattering; note that the control period is usually $200\,\mathrm{\mu s}\gg 29.2\,\mathrm{ns}$. Furthermore, if substituting the optimal trajectory in the chattering period by MIM-trajectory, the induced relative error in terminal time is strictly less than $0.12\%$ for any 4th order problems, which is proved in \cite{wang2024part2}. Therefore, the MIM-trajectory is near-optimal and practical for high-order problems, since chattering does not occur in MIM theoretically according to the process of Algorithm \ref{alg:CalMIM}.}

        It is meaningful to strictly solve the time-optimal problem \eqref{eq:optimalproblem} in further study, which would be a landmark achievement in the optimal control theory. To the best of our knowledge, the theoretical framework established in Section \ref{sec:CostateSystemBehaviorAnalysis} provides unprecedented insights into problem \eqref{eq:optimalproblem}, surpassing current literature. Therefore, it is believed that the theoretical framework established in Section \ref{sec:CostateSystemBehaviorAnalysis} would be a noteworthy mathematical tool for the final resolution of the time-optimal problem \eqref{eq:optimalproblem}.



\section{Conclusion}
    This paper has set out to theoretically study a classical and challenging problem in the optimal control theory domain, i.e., the time-optimal control problem for high-order chain-of-integrators systems with full state constraints and arbitrary terminal states. To this end, this paper establishes a novel notation system and theoretical framework, providing the switching manifold for high-order problems in the form of switching laws. The framework derives properties of switching laws \modify{regarding} signs as well as dimension and reasons a definite condition of augmented switching laws. Guided by the developed framework, a trajectory planning method named the manifold-intercept method (MIM) has been proposed, outperforming all baselines by a large gap \modify{i}n computational time, computational accuracy, and trajectory quality. The proposed MIM can achieve time-optimal trajectories for 3rd order or lower-order problems with full state constraints. MIM can also plan near-time-optimal trajectories efficiently and accurately with negligible extra motion time compared to time-optimal trajectories that are \modify{in} lack of mature algorithms currently\modify{, avoiding \modify{the} chattering phenomenon that impedes numerical computation in practice}.

\section*{Acknowledgment}
    \modify{The authors would like to thank the anonymous reviewers for their valuable comments and suggestions, to thank Yujie Lin and Hui Ma for their expertise in differential geometry, and to thank Zongying Shi for her experitse in optimal control. This work was supported by the National Key Research and Development Program of China under Grant 2023YFB4302003.}

\ifCLASSOPTIONcaptionsoff
  \newpage
\fi



\bibliographystyle{IEEEtran}
\bibliography{IEEEabrv,refs/ref}
%



%

\vskip -2\baselineskip plus -1fil

\begin{IEEEbiography}[{\includegraphics[width=1in,keepaspectratio]{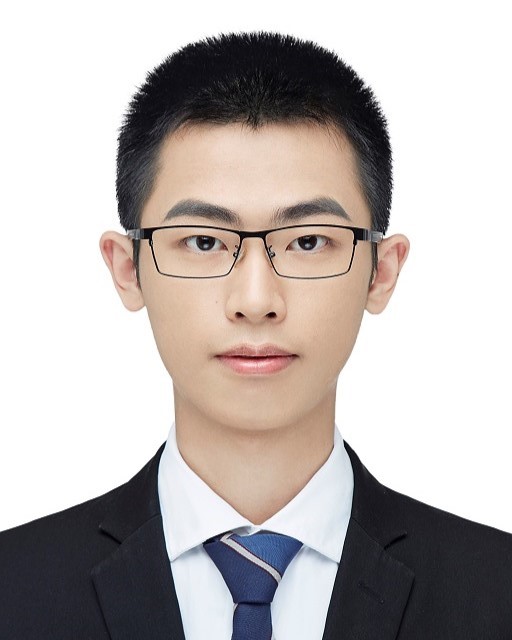}}]{Yunan Wang}
    (S'22) received the B.S. degree in mechanical engineering, in 2022, from the Department of Mechanical Engineering, Tsinghua University, Beijing, China. He is currently working toward the Ph.D. degree in mechanical engineering. 

    His research interests include optimal control, trajectory planning, toolpath planning, and precision motion control. He was the recipient of the Best Conference Paper Finalist at the 2022 International Conference on Advanced Robotics and Mechatronics, and 2021 Top Grade Scholarship for Undergraduate Students of Tsinghua University.
\end{IEEEbiography}

\vskip -2\baselineskip plus -1fil

\begin{IEEEbiography}[{\includegraphics[width=1in,keepaspectratio]{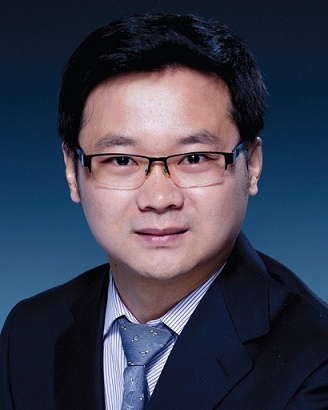}}]{Chuxiong Hu}
    (S'09-M'11-SM'17) received his B.S. and Ph.D. degrees in Mechatronic Control Engineering from Zhejiang University, Hangzhou, China, in 2005 and 2010, respectively.
    
    He is currently an Associate Professor (tenured) at Department of Mechanical Engineering, Tsinghua University, Beijing, China. From 2007 to 2008, he was a Visiting Scholar in mechanical engineering with Purdue University, West Lafayette, USA. In 2018, he was a Visiting Scholar in mechanical engineering with University of California, Berkeley, CA, USA. His research interests include precision motion control, high-performance multiaxis contouring control, precision mechatronic systems, intelligent learning, adaptive robust control, neural networks, iterative learning control, and robot.
    
    Prof. Hu was the recipient of the Best Student Paper Finalist at the 2011 American Control Conference, the 2012 Best Mechatronics Paper Award from the ASME Dynamic Systems and Control Division, the 2013 National 100 Excellent Doctoral Dissertations Nomination Award of China, the 2016 Best Paper in Automation Award, the 2018 Best Paper in AI Award from the IEEE International Conference on Information and Automation, and 2022 Best Paper in Theory from the IEEE/ASME International Conference on Mechatronic, Embedded Systems and Applications. He is now an Associate Editor for the IEEE Transactions on Industrial Informatics and a Technical Editor for the IEEE/ASME Transactions on Mechatronics.
\end{IEEEbiography}

\vskip -2\baselineskip plus -1fil

\begin{IEEEbiography}[{\includegraphics[width=1in,keepaspectratio]{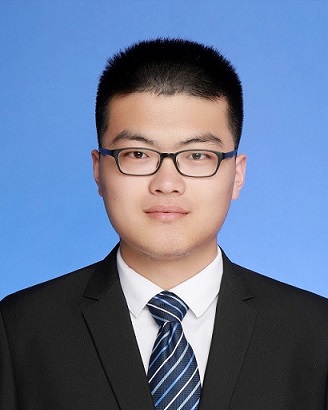}}]{Zeyang Li}
    received the B.S. degree in mechanical engineering in 2021, from the School of Mechanical Engineering, Shanghai Jiao Tong University, Shanghai, China. He is currently working toward the M.S. degree in mechanical engineering with the Department of Mechanical Engineering, Tsinghua University, Beijing, China.

    His current research interests include the areas of optimal control and reinforcement learning.
\end{IEEEbiography}

\vskip -2\baselineskip plus -1fil

\begin{IEEEbiography}[{\includegraphics[width=1in,keepaspectratio]{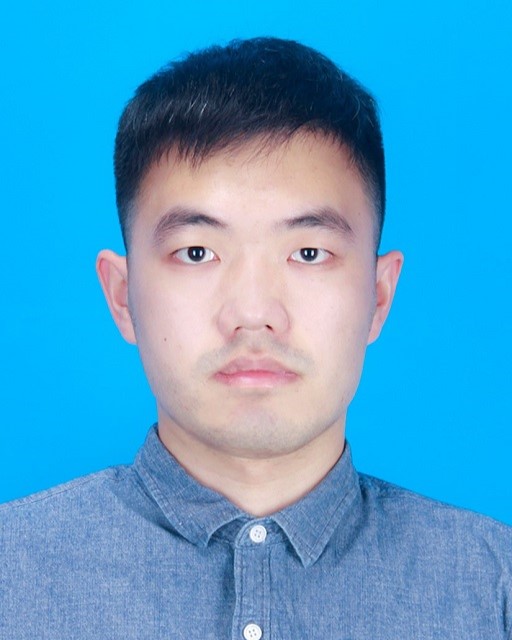}}]{Shize Lin}
    received the B.S. degree in mechanical engineering, in 2020, from the Department of Mechanical Engineering, Tsinghua University, Beijing, China, where he is currently working toward the Ph.D. degree in mechanical engineering. His research interests include robotics, motion planning and precision motion control.
\end{IEEEbiography}

\vskip -2\baselineskip plus -1fil

\begin{IEEEbiography}[{\includegraphics[width=1in,keepaspectratio]{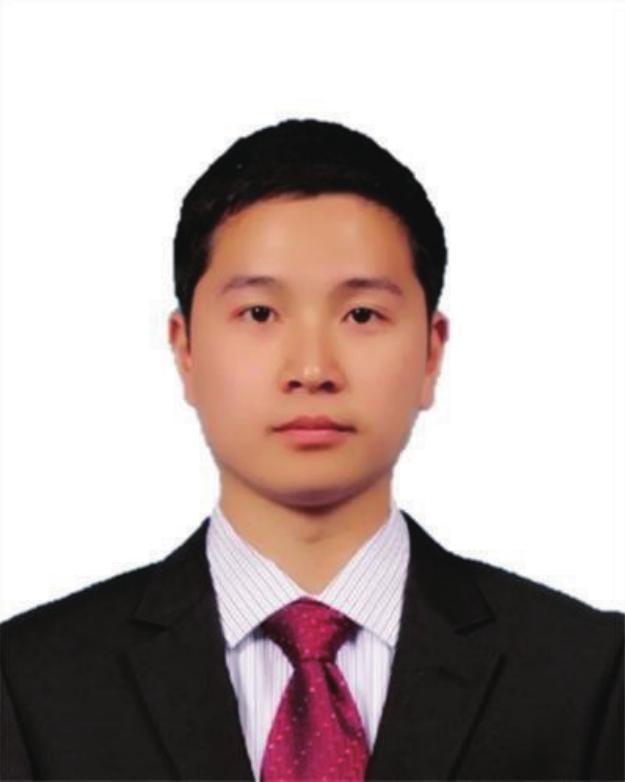}}]{Suqin He}
    received the B.S. degree in mechanical engineering from Department of Mechanical Engineering, Tsinghua University, Beijing, China, in 2016, and the Ph.D. degree in mechanical engineering from the Department of Mechanical Engineering, Tsinghua University, Beijing, China, in 2023.

    His research interests include multi-axis trajectory planning and precision motion control on robotics and CNC machine tools. He was the recipient of the Best Automation Paper from IEEE Internal Conference on Information and Automation in 2016.
\end{IEEEbiography}





\vskip -2\baselineskip plus -1fil

\begin{IEEEbiography}[{\includegraphics[width=1in,keepaspectratio]{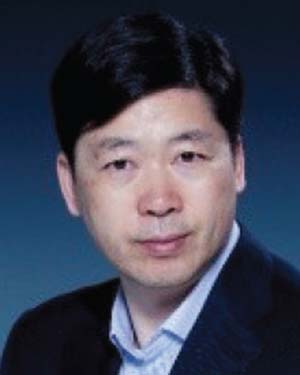}}]{Yu Zhu}
    (Member, IEEE) received the B.S. degree in radio electronics from Beijing Normal University, Beijing, China, in 1983, and the M.S. degree in computer applications and the Ph.D. degree in mechanical design and theory from the China University of Mining and Technology, Beijing, in 1993 and 2001, respectively. He is currently a Professor with the State Key Laboratory of Tribology, Department of Mechanical Engineering, Tsinghua University, China. He has authored or coauthored more than 180 technical papers. He is also the holder of more than 140 warranted invention patents. His research interests include precision measurement and motion control, ultraprecision mechanical design and manufacturing, two-photon microfabrication, and electronics manufacturing technology and equipment.
\end{IEEEbiography}









\newpage

\onecolumn

\appendices
\section{Proof of Independence of Equations in Theorem \ref{thm:DimensionSwitchLaw}}\label{app:prove_independence}

Given the terminal state vector $\vxf\in\R^n$, the constraints $\vM=\left(M_0,M_1,\dots,M_n\right)\in\R_{++}\times\overline{\R}_{++}^{n}$, and a switching law $S=s_1s_2\dots s_N$, we need to prove that if $\exists\vx_0\in\R^n$, s.t. $S\in\problem\left(\vx_0,\vx_\f;\vM\right)$, then the dimension of the manifold represented by $S$ is
\begin{equation}
    \dim S=N-\sum_{i=1}^{N}\abs{s_i}.
\end{equation}

Before proving, we need to prove that $\vJ$ has full row rank. The condition in Theorem \ref{thm:DimensionSwitchLaw} is equivalent to the following condition.

\begin{condition}\label{condition:3_2_switchinglaw}
    $\forall i=1,2,\dots,N$, the following conditions hold:
    \begin{enumerate}
        \item\label{condition:3_2_switchinglaw_1ji} If $\exists 1\leq j<i$, s.t., $\abs{s_j}\geq\abs{s_i}$, then $\sum_{k=j+1}^{i}\abs{s_k}<i-j$.
        \item\label{condition:3_2_switchinglaw_ijN} If $\exists i<j\leq N$, s.t., $\abs{s_j}\geq\abs{s_i}$, then $\sum_{k=i}^{j-1}\abs{s_k}<j-i$.
        \item\label{condition:3_2_switchinglaw_ijN_} If $\abs{s_i}>0$ and $\forall i<j\leq N$, $\abs{s_j}<\abs{s_i}$, then $\sum_{k=i}^{N}\abs{s_k}\leq N-i$.
    \end{enumerate}
\end{condition}

Condition \ref{condition:3_2_switchinglaw} is naturally introduced to avoid some special cases. To understand Condition \ref{condition:3_2_switchinglaw}-\ref{condition:3_2_switchinglaw_1ji}, we consider the trajectory $\vx_{1:\abs{s_i}}\left(t\right)$, $t\in\left[t_j,t_i\right]$. It is evident that $\vx_{1:\abs{s_i}}\left(t_i\right)=M_{\abs{s_i}}\sgn\left(s_i\right)\ve_{\abs{s_i}}$. If $\abs{s_j}=\abs{s_i}$, then $\vx_{1:\abs{s_i}}\left(t_j\right)=M_{\abs{s_i}}\sgn\left(s_j\right)\ve_{\abs{s_i}}$; if $\abs{s_j}>\abs{s_i}$, then $\vx_{1:\abs{s_i}}\left(t_j\right)=\vzero$. Then, the ``degree of freedom'' should not be smaller than $\abs{s_i}$, i.e., the number of variables should not be less than the number of equations. Hence, it is assumed that $\sum_{k=j+1}^{i}\abs{s_k}<i-j$. Similar analysis can be applied to Condition \ref{condition:3_2_switchinglaw}-\ref{condition:3_2_switchinglaw_ijN}. For Condition \ref{condition:3_2_switchinglaw}-\ref{condition:3_2_switchinglaw_ijN_}, it does not hold only when the terminal state vector $\vxf$ is in a zero-measure set. For example, if $\abs{s_N}\not=0$, then $\vxf$ satisfies $\vx_{\f,1:\abs{s_N}}=\sgn\left(s_N\right)M_{\abs{s_N}}$; however, once $\vxf$ is disturbed slightly resulting in $\hat\vx_\f$, then except of a zero-measure set in the neighborhood of $\vxf$, the switching law should with terminal states $\hat\vx_\f$ should end with $s_N$ and $\abs{s_N}$ zeros.

We assume $\vx_i=\left(x_{i,k}\right)_{k=1}^n\in\R^n,i=0,1,\dots,N$ and $t_i\geq 0,i=1,2,\dots,N$. $\vx_{i-1}$ moves to $\vx_i$ after time $t_i$ under the system behavior $s_i$, $\forall i=1,2,\dots,N$. Among them, the terminal state vector $\vx_N=\vxf$ is given, while $\vx_i$, $i=0,1,\dots,N$ and $t_i,i=1,2,\dots,N$ are variables. By Condition \ref{condition:3_2_switchinglaw}-\ref{condition:3_2_switchinglaw_ijN_}, $\abs{s_N}=0$.

Under the system behavior $s_i$, the control $u\equiv u_i$, where
\begin{equation}\label{eq:3.2.2_ui}
    x_{i,0}\triangleq u_i=\begin{dcases}
        \sgn\left(s_i\right)M_0,&\abs{s_i}=0\\
        0,&\abs{s_i}\not=0\\
    \end{dcases}.
\end{equation}

Denote
\begin{equation}
    \vX=\left[\begin{array}[]{c}
        \vx_0\\\vx_1\\\vdots\\\vx_{N-1}\\t_1\\t_2\\\vdots\\t_N
    \end{array}\right]\in\R^{N\left(n+1\right)}.
\end{equation}

Then, we have the following constraints:
\begin{IEEEeqnarray}{l}
    f_{\left(i-1\right)n+k}\left(\vX\right)=\sum_{j=1}^{k}\dfrac{1}{\left(k-j\right)!}x_{i-1,j}t_i^{k-j}+\dfrac{1}{k!}u_it_i^k-x_{i,k}=0,\,i=1,2,\dots,N,\,k=1,2,\dots,n,\IEEEyesnumber\IEEEyessubnumber*\label{eq:R3C2_constraint1}\\
    f_{Nn+\sum_{j=1}^{i-1}\abs{s_j}+k}\left(\vX\right)=x_{i,k}=0,\,i=1,2,\dots,N-1,\,k=1,2,\dots,\abs{s_i}-1,\label{eq:R3C2_constraint2}\\
    f_{Nn+\sum_{j=1}^{i}\abs{s_j}}\left(\vX\right)=x_{i,k}-\sgn\left(s_i\right)M_{\abs{s_i}}=0,\,i=1,2,\dots,N-1.\label{eq:R3C2_constraint3}
\end{IEEEeqnarray}
Among them, $x_{N,k}=0$ for $k<\abs{s_N}$ and $x_{N,\abs{s_N}}=\sgn\left(s_N\right)M_{\abs{s_N}}$ are guaranteed by the given $\vx_\f$ since $\exists\vx_0\in\R^n$, s.t. $S\in\problem\left(\vx_0,\vx_\f;\vM\right)$. Naturally, according to \eqref{eq:3.2.2_ui}, $x_{N-1,k}=0$ for $k<\abs{s_N}$ and $x_{N-1,\abs{s_N}}=\sgn\left(s_N\right)M_{\abs{s_N}}$.

Denote the above constraints as $\vf\left(\vX\right)=\vzero$. Since $\exists\vx_0\in\R^n$, s.t. $S\in\problem\left(\vx_0,\vx_\f;\vM\right)$, the system of equations $\vf\left(\vX\right)=\vzero$ has at least one solution, denoted as
\begin{equation}
    \vX^*=\left[\begin{array}[]{c}
        \vx_0^*\\\vx_1^*\\\vdots\\\vx_{N-1}^*\\t_1^*\\t_2^*\\\vdots\\t_N^*
    \end{array}\right]\in\R^{N\left(n+1\right)}.
\end{equation}

Then, we try to calculate $\where{\frac{\mathrm{d}\vf}{\mathrm{d}\vX}}{\vX^*}$, i.e., the Jacobi matrix of $\vf\left(\vX\right)$ at $\vX^*$. For \eqref{eq:R3C2_constraint1}, $\forall i=1,2,\dots,N$,
\begin{IEEEeqnarray}{l}
    \where{\frac{\partial\vf_{\left(\left(i-1\right)n+1\right):in}}{\partial\vx_{i-1}}}{\vX^*}=\vPhi_n\left(t_i^{*}\right),\IEEEyesnumber\IEEEyessubnumber*\label{eq:3_2_dfdx_11i1}\\
    \where{\frac{\partial\vf_{\left(\left(i-1\right)n+1\right):in}}{\partial\vx_{i}}}{\vX^*}=-\vI_n,\,i<N,\\
    \where{\frac{\partial\vf_{\left(\left(i-1\right)n+1\right):in}}{\partial\vx_{j}}}{\vX^*}=\vzero_{n\times n},\,j\not=i,i-1,\\
    \where{\frac{\partial\vf_{\left(\left(i-1\right)n+1\right):in}}{\partial t_i}}{\vX^*}=\left(\sum_{j=0}^{k-1}\dfrac{1}{\left(k-j\right)!}x_{i-1,j}^*{t_i^*}^{k-j}\right)_{k=1}^n\xlongequal{\eqref{eq:R3C2_constraint1}}\vx^*_{i,0:\left(n-1\right)},\\
    \where{\frac{\partial\vf_{\left(\left(i-1\right)n+1\right):in}}{\partial t_j}}{\vX^*}=\vzero_{n\times1},j\not=i.
\end{IEEEeqnarray}
Among them, for a vector $\vx=\left(x_i\right)_{i=1}^n$, $\vx_{a:b}\triangleq\vx=\left(x_i\right)_{i=a}^b$. For a matrix $\vA=\left(\left(a_{ij}\right)_{i=1}^m\right)_{j=1}^n$, $\vA_{a:b,c:d}\triangleq\left(\left(a_{ij}\right)_{i=a}^b\right)_{j=c}^d$. $\vzero_{m\times n}=\left(\left(0\right)_{i=1}^m\right)_{j=1}^n$. In \eqref{eq:3_2_dfdx_11i1},
\begin{equation}\label{eq:3_2_vPhi}
    \vPhi_a\left(t\right)\triangleq\left[\begin{array}{ccccc}
        1&&&&\\
        t&1&&&\\
        \frac{{t}^2}{2}&t&1&&\\
        \vdots&\vdots&\vdots&\ddots&\\
        \frac{{t}^{a-1}}{\left(a-1\right)!}&\frac{{t}^{a-2}}{\left(a-2\right)!}&\frac{{t}^{a-3}}{\left(a-3\right)!}&\cdots&1
    \end{array}\right],\, a\in\N^*,
\end{equation}
is invertible. $\forall t\in\R$, $a\in\N^*$, $\det\vPhi_a\left(t\right)=1$.

For \eqref{eq:R3C2_constraint2} and \eqref{eq:R3C2_constraint3}, $\forall i=1,2,\dots,N-1$,
\begin{IEEEeqnarray}{l}
    \where{\frac{\partial\vf_{\left(Nn+\sum_{j=1}^{i-1}\abs{s_j}+1\right):\left(Nn+\sum_{j=1}^{i}\abs{s_j}\right)}}{\partial\vx_{i}}}{\vX^*}=\vI_{\abs{s_i}\times n},\IEEEyesnumber\IEEEyessubnumber*\\
    \where{\frac{\partial\vf_{\left(Nn+\sum_{j=1}^{i-1}\abs{s_j}+1\right):\left(Nn+\sum_{j=1}^{i}\abs{s_j}\right)}}{\partial\vx_{j}}}{\vX^*}=\vzero_{\abs{s_i}\times n},\,j\not=i,\\
    \where{\frac{\partial\vf_{\left(Nn+\sum_{j=1}^{i-1}\abs{s_j}+1\right):\left(Nn+\sum_{j=1}^{i}\abs{s_j}\right)}}{\partial t_j}}{\vX^*}=\vzero_{\abs{s_i}\times 1},\,j=1,2,\dots,N.
\end{IEEEeqnarray}
Among them, $\vI_{m\times n}=\left(\left(\delta_{ij}\right)_{i=1}^m\right)_{j=1}^n$. Specifically, if $m\leq n$, $\vI_{m\times n}=\left[\vI_m,\vzero_{m\times\left(n-m\right)}\right]$. If $m>n$, $\vI_{m\times n}=\vI_{n\times m}^\top$.

Therefore, the Jacobi matrix of $\vf\left(\vX\right)$ at $\vX^*$, denoted as $\vJ=\where{\frac{\mathrm{d}\vf}{\mathrm{d}\vX}}{\vX^*}$, is
\begin{equation}\label{eq:3_2_Jacobian}
    \resizebox{0.9\textwidth}{!}{
    $\vJ=
    \left[\begin{array}{cccccccccc}
        \vPhi_n\left(t_1^*\right)&-\vI_n&&&&\vx_{1,0:\left(n-1\right)}^*&&&&\\
        &\vPhi_n\left(t_2^*\right)&-\vI_n&&&&\vx_{2,0:\left(n-1\right)}^*&&&\\
        &&\ddots&\ddots&&&&\ddots&&\\
        &&&\vPhi_n\left(t_{N-1}^*\right)&-\vI_n&&&&\vx_{N-1,0:\left(n-1\right)}^*&\\
        &&&&\vPhi_n\left(t_{N}^*\right)&&&&&\vx_{N,0:\left(n-1\right)}^*\\
        &\vI_{\abs{s_1}\times n}&&&&&&&&\\
        &&\vI_{\abs{s_2}\times n}&&&&&&&\\
        &&&\ddots&&&&&&\\
        &&&&\vI_{\abs{s_{N-1}}\times n}&&&&&
    \end{array}\right].$
    }
\end{equation}

\begin{theorem}\label{thm:3_2_fullrowrank}
    If Condition \ref{condition:3_2_switchinglaw} holds, then $\vJ$ in \eqref{eq:3_2_Jacobian} has full row rank.
\end{theorem}

Due to the complexity of the proof of Theorem 3.2.1, we relegate the demonstration of this theorem to the conclusion of the response to this comment.

\begin{lemma}\label{lemma:3_2_phi}
    $n\in\N^*$. Denote $\phi_k\left(t\right)\triangleq\frac{t^{k-1}}{\left(k-1\right)!}$. Denote $\vphi_n\left(t\right)\triangleq\left(\phi_k\left(t\right)\right)_{k=1}^n\in\R^n$. $\vPhi_n$ is defined in \eqref{eq:3_2_vPhi}. Assume $\vx=\vx\left(t\right)=\left(x_k\left(t\right)\right)_{k=1}^n$ is the solution of the initial value problem
    \[
        \begin{dcases}
            \dot{x}_{1}\left(t\right)=0,\,t\in\R,\\
            \dot{x}_{k+1}\left(t\right)={x}_{k}\left(t\right),\,k=1,2,\dots,n-1,\,t\in\R,\\
            \vx\left(t_0\right)=\vx_0.
        \end{dcases}
    \]
    Among them, $t_0\in\R$ and $\vx_0=\left(x_{0,k}\right)_{k=1}^n\in\R^n$ are given. Then,
    \begin{enumerate}
        \item\label{lemma:3_2_phi_0} $\vphi_n\left(0\right)=\ve_1\in\R^n$, $\vPhi_n\left(0\right)=\vI_n$.
        \item\label{lemma:3_2_phi_t} $\forall t_1,t_2\in\R$, $\vPhi_n\left(t_2\right)\vphi_n\left(t_1\right)=\vphi_n\left(t_1+t_2\right)$, $\vPhi_n\left(t_1\right)\vPhi_n\left(t_2\right)=\vPhi_n\left(t_2\right)\vPhi_n\left(t_1\right)=\vPhi_n\left(t_1+t_2\right)$.
        \item \label{lemma:3_2_phi_-1} $\forall t\in\R$, $\vPhi_n\left(t\right)$ is invertible, and $\vPhi_n^{-1}\left(t\right)=\vPhi_n\left(-t\right)$.
        \item\label{lemma:3_2_phi_x} $\forall t\in\R$, $\vx\left(t\right)=\vPhi_n\left(t-t_0\right)\vx_0$.
    \end{enumerate}
\end{lemma}

\begin{proof}
    Lemma \ref{lemma:3_2_phi}.\ref{lemma:3_2_phi_0} holds evidently.

    Now we prove Lemma \ref{lemma:3_2_phi}.\ref{lemma:3_2_phi_t}. $\forall t_1,t_2\in\R$, $\vPhi_n\left(t_k\right)=\left(\left(\frac{t_k^{i-j}}{\left(i-j\right)!}\delta_{i\geq j}\right)_{i=1}^n\right)_{j=1}^n$, $k=1,2$. Therefore, 
    \[
        \begin{aligned}
            \vPhi_n\left(t_1\right)\vPhi_n\left(t_2\right)=&\left(\left(\frac{t_1^{i-k}}{\left(i-k\right)!}\delta_{i\geq k}\right)_{i=1}^n\right)_{k=1}^n\left(\left(\frac{t_2^{k-j}}{\left(k-j\right)!}\delta_{k\geq j}\right)_{k=1}^n\right)_{j=1}^n\\
            =&\left(\left(\sum_{k=1}^{n}\frac{t_1^{i-k}}{\left(i-k\right)!}\delta_{i\geq k}\frac{t_2^{k-j}}{\left(k-j\right)!}\delta_{k\geq j}\right)_{i=1}^n\right)_{j=1}^n\\
            =&\left(\left(\sum_{k=j}^{i}\frac{t_1^{i-k}t_2^{k-j}}{\left(i-k\right)!\left(k-j\right)!}\delta_{i\geq j}\right)_{i=1}^n\right)_{j=1}^n\\
            =&\left(\left(\frac{\delta_{i\geq j}}{\left(i-j\right)!}\sum_{k=0}^{i-j}\frac{\left(i-j\right)!}{\left(i-j-k\right)!k!}t_1^{i-j-k}t_2^{k}\right)_{i=1}^n\right)_{j=1}^n\\
            =&\left(\left(\frac{\left(t_1+t_2\right)^{i-j}}{\left(i-j\right)!}\delta_{i\geq j}\right)_{i=1}^n\right)_{j=1}^n\\
            =&\vPhi_n\left(t_1+t_2\right).
        \end{aligned}
    \]
    So $\vPhi_n\left(t_1\right)\vPhi_n\left(t_2\right)=\vPhi_n\left(t_2\right)\vPhi_n\left(t_1\right)=\vPhi_n\left(t_1+t_2\right)$ for the same reason. Note that $\vphi_n\left(t_2\right)$ and $\vphi_n\left(t_1+t_2\right)$ are the first row of $\vPhi_n\left(t_2\right)$ and $\vPhi_n\left(t_1+t_2\right)$, respectively. Therefore, $\vPhi_n\left(t_2\right)\vphi_n\left(t_1\right)=\vphi_n\left(t_1+t_2\right)$. Lemma \ref{lemma:3_2_phi}.\ref{lemma:3_2_phi_t} holds.

    By Lemma \ref{lemma:3_2_phi}.\ref{lemma:3_2_phi_t}, $\forall t\in\R$, $\vPhi_n\left(t\right)\vPhi_n^{-1}\left(t\right)=\vPhi_n^{-1}\left(t\right)\vPhi_n\left(t\right)=\vI_n$. Therefore, $\vPhi_n^{-1}\left(t\right)=\vPhi_n\left(-t\right)$. Lemma \ref{lemma:3_2_phi}.\ref{lemma:3_2_phi_-1} holds.

    Next we prove Lemma \ref{lemma:3_2_phi}.\ref{lemma:3_2_phi_x}. By solving the initial value problem, $\forall t\in\R$, $k=1,2,\dots,n$, $x_k\left(t\right)=\sum_{i=0}^{k}\frac{\left(t-t_0\right)^{i}}{i!}x_{k-i}\left(t_0\right)$. Therefore, $\vx\left(t\right)=\vPhi_n\left(t-t_0\right)\vx_0$. Lemma \ref{lemma:3_2_phi}.\ref{lemma:3_2_phi_x} holds.
\end{proof}

\begin{lemma}\label{lemma:3_2_hatphi}
    $n\in\N^*$. Denote $\hat{\phi}_k\left(t\right)\triangleq t^{k-1}$. Denote $\hat{\vphi}_n\left(t\right)\triangleq\left(\hat{\phi}_k\left(t\right)\right)_{k=1}^n$. Denote $\widehat{\vPhi}_n\left(t\right)=\left(\left(\frac{i!}{j!\left(i-j\right)!}t^{i-j}\right)_{i=1}^n\right)_{j=1}^n$, i.e.,
    \[
        \hat{\vphi}_n\left(t\right)\triangleq\left[\begin{array}{c}
            1\\t\\t^2\\\vdots\\t^{n-1}
        \end{array}\right],\,
        \widehat{\vPhi}_n\left(t\right)\triangleq\left[\begin{array}{ccccc}
            1&&&&\\
            2t&1&&&\\
            3t^2&3t&1&&\\
            \vdots&\vdots&\vdots&\ddots&\\
            nt^{n-1}&\frac{n(n-1)}{2}t^{n-2}&\frac{n(n-1)(n-2)}{6}t^{n-3}&\cdots&1
        \end{array}\right].
    \]
    Then,
    \begin{enumerate}
        \item\label{lemma:3_2_hatphi_0} $\hat{\vphi}_n\left(0\right)=\ve_1\in\R^n$, $\widehat{\vPhi}_n\left(0\right)=\vI_n$.
        \item\label{lemma:3_2_hatphi_t} $\forall t_1,t_2\in\R$, $\widehat{\vPhi}_n\left(t_2\right)\hat{\vphi}_n\left(t_1\right)=\hat{\vphi}_n\left(t_1+t_2\right)$, $\widehat{\vPhi}_n\left(t_1\right)\widehat{\vPhi}_n\left(t_2\right)=\widehat{\vPhi}_n\left(t_2\right)\widehat{\vPhi}_n\left(t_1\right)=\widehat{\vPhi}_n\left(t_1+t_2\right)$.
        \item \label{lemma:3_2_hatphi_-1} $\forall t\in\R$, $\widehat{\vPhi}_n\left(t\right)$ is invertible, and $\widehat{\vPhi}_n^{-1}\left(t\right)=\widehat{\vPhi}_n\left(-t\right)$.
    \end{enumerate}
\end{lemma}

\begin{proof}
    Lemma \ref{lemma:3_2_hatphi}.\ref{lemma:3_2_hatphi_0} holds evidently. Note that $\hat{\vphi}_n\left(t\right)=\vC\vphi_n\left(t\right),\,\widehat{\vPhi}_n\left(t\right)=\vC\widehat{\vPhi}_n\left(t\right)\vC^{-1}$, where
    \[
        \vC=\left[\begin{array}{cccc}
            0!&&&\\ &1!&&\\ &&\ddots&\\ &&&\left(n-1\right)!
        \end{array}\right].
    \]
    
    According to Lemma \ref{lemma:3_2_phi}.\ref{lemma:3_2_phi_t}, $\forall t_1,t_2\in\R$, 
    \[
        \widehat{\vPhi}_n\left(t_2\right)\hat{\vphi}_n\left(t_1\right)=\vC\vPhi_n\left(t_2\right)\vC^{-1}\vC\vphi_n\left(t_1\right)=\vC\vphi_n\left(t_1+t_2\right)=\hat{\vphi}_n\left(t_1+t_2\right).
    \]
    Similarly,
    \[
        \widehat{\vPhi}_n\left(t_1\right)\widehat{\vPhi}_n\left(t_2\right)=\vC\vPhi_n\left(t_1\right)\vC^{-1}\vC\vPhi_n\left(t_2\right)\vC^{-1}=\vC\vPhi_n\left(t_1\right)\vPhi_n\left(t_2\right)\vC^{-1}=\vC\vPhi_n\left(t_1+t_2\right)\vC^{-1}=\widehat{\vPhi}_n\left(t_1+t_2\right).
    \]
    So Lemma \ref{lemma:3_2_hatphi}.\ref{lemma:3_2_hatphi_t} holds. Furthermore, $\forall t\in\R$, $\widehat{\vPhi}_n\left(t\right)\widehat{\vPhi}_n\left(-t\right)=\widehat{\vPhi}_n\left(-t\right)\widehat{\vPhi}_n\left(t\right)=\widehat{\vPhi}_n\left(0\right)=\vI_n$; hence, Lemma \ref{lemma:3_2_hatphi}.\ref{lemma:3_2_hatphi_-1} holds.
\end{proof}

\begin{lemma}\label{lemma:3_2_difference}
    $\forall\left\{T_k\right\}_{k=1}^\infty\subset\R$ is monotonically strictly increasing, i.e., $T_1<T_2<T_3<\dots$, define
    \begin{equation}\label{eq:3_2_definedifference}
        D_{n,m,p}=D_{n,m}\left(\left(T_k\right)_{k=p}^{m+p}\right)\triangleq\begin{dcases}
            1,&m=0,n=0,\\
            T_p^n,&m=0,n\geq1,\\
            \frac{D_{n,m-1}\left(\left(T_k\right)_{k=p+1}^{m+p}\right)-D_{n,m-1}\left(\left(T_k\right)_{k=p}^{m+p-1}\right)}{T_{m+p}-T_p},&m\geq1,n\geq0.
        \end{dcases}
    \end{equation}
    Then,
    \begin{equation}\label{eq:3_2_Dnmp_express}
        D_{n,m,p}=\begin{dcases}
            \sum_{p\leq i_1\leq i_2\leq\dots\leq i_{n-m}\leq m+p}\prod_{k=1}^{n-m}T_{i_k},&0\leq m<n,\\
            1,&m=n,\\
            0,&m>n.
        \end{dcases}
    \end{equation}
    Among them, $\forall a\in\R$, $a^0=0$ is conventionally defined. Specifically, for cases where $0<m\leq n$,
    \begin{equation}\label{eq:3_2_Dnmp_0mn}
        D_{n,m,p}=\sum_{k=0}^{n-m}T_p^kD_{n-k-1,m-1,p+1}=\sum_{k=0}^{n-m}T_{m+p}^kD_{n-k-1,m-1,p}.
    \end{equation}
\end{lemma}

\begin{proof}
    $D_{n,m,p}$ is well-defined since $\forall p\in\N^*$, $m\geq 1$, $T_p<T_{m+p}$. For the case where $n=0$ or $m=0$, \eqref{eq:3_2_Dnmp_express} holds evidently. For the case where $0<m=n$, \eqref{eq:3_2_Dnmp_0mn} holds evidently if \eqref{eq:3_2_Dnmp_express} holds.


    Now we prove the case where $0<m\leq n$ by induction. Since $D_{n,0,p}=T_p^n$, we have
    \begin{equation}
        \begin{aligned}
            D_{n,1,p}=&\frac{D_{n,0,p+1}-D_{n,0,p}}{T_{p+1}-T_p}=\frac{T_{p+1}^n-T_p^n}{T_{p+1}-T_p}=\sum_{k=0}^{n-1}T_p^kT_{p+1}^{n-1-k}=\sum_{p\leq i_1\leq i_2\leq\dots\leq i_{n-1}\leq p+1}\prod_{k=1}^{n-1}T_{i_k}\\
            =&\sum_{k=0}^{n-1}T_p^kD_{n-k-1,0,p+1}=\sum_{k=0}^{n-1}T_{p+1}^kD_{n-k-1,0,p}.
        \end{aligned}
    \end{equation}
    Therefore, \eqref{eq:3_2_Dnmp_express} and \eqref{eq:3_2_Dnmp_0mn} hold for $1=m\leq n$.
    
    Assume that \eqref{eq:3_2_Dnmp_express} and \eqref{eq:3_2_Dnmp_0mn} hold for $m\leq M$, $n\geq m$, $M\geq1$. Consider the case where $n=m=M+1$. $D_{M+1,M,p}=\sum_{p\leq i\leq M+p}T_i=\sum_{i=p}^{M+p}T_i$. Hence,
    \begin{equation}
        D_{M+1,M+1,p}=\frac{D_{M+1,M,p+1}-D_{M+1,M,p}}{T_{p+M+1}-T_p}=\frac{\sum_{i=p+1}^{M+p+1}T_i-\sum_{i=p}^{M+p}T_i}{T_{p+M+1}-T_p}=1.
    \end{equation}
    Therefore, \eqref{eq:3_2_Dnmp_express} and \eqref{eq:3_2_Dnmp_0mn} hold for $n=m=M+1$.

    Consider the case where $n>m=M+1$. Note that $n>m-1=M$, so
    \begin{equation}
        \begin{aligned}
            &D_{n,M+1,p}=\frac{D_{n,M,p+1}-D_{n,M,p}}{T_{p+M+1}-T_p}\\
            =&\frac{\sum_{k=0}^{n-M}T_{p+M+1}^kD_{n-k-1,M-1,p+1}-\sum_{k=0}^{n-M}T_{p}^kD_{n-k-1,M-1,p+1}}{T_{p+M+1}-T_p}=\sum_{k=0}^{n-M}\frac{T_{p+M+1}^k-T_p^k}{T_{p+M+1}-T_p}D_{n-k-1,M-1,p+1}\\
            =&\sum_{k=0}^{n-M}\sum_{q=0}^{k-1}T_p^qT_{p+M+1}^{k-q-1}\sum_{p+1\leq i_{q+1}\leq i_{q+2}\leq\dots\leq i_{n-M-k+q}\leq p+M}\prod_{j=q+1}^{n-M-k+q}T_{i_j}\\
            =&\sum_{k=0}^{n-M}\sum_{q=0}^{k-1}\sum_{\substack{i_1=i_2=\dots=i_q=p\\p+1\leq i_{q+1}\leq i_{q+2}\leq\dots\leq i_{n-M-k+q}\leq p+M\\i_{n-M-k+q+1}=i_{n-M-k+q+1}=\dots=i_{n-M-1}=p+M+1}}\prod_{j=1}^{n-M-1}T_{i_j}\\
            =&\sum_{p\leq i_{1}\leq i_{2}\leq\dots\leq i_{n-M-1}\leq p+M+1}\prod_{j=1}^{n-M-1}T_{i_j}.
        \end{aligned}
    \end{equation}
    Therefore, \eqref{eq:3_2_Dnmp_express} holds for $n>m=M+1$, implying \eqref{eq:3_2_Dnmp_0mn}. By induction, Lemma \ref{lemma:3_2_difference} holds.
\end{proof}

\begin{proof}[Proof of Theorem \ref{thm:3_2_fullrowrank}]
    In this proof, we denote $\vA\leftrightarrow\vB$ if $\vA$ has full row rank if and only if $\vB$ has full row rank.

    Denote $\vx=\vx^{\left(i\right)}\left(t\right)$, $i=1,2,\dots,N$, as the solution of the initial value problem
    \[
        \begin{dcases}
            \dot{x}_{1}\left(t\right)=0,\,t\in\R,\\
            \dot{x}_{k+1}\left(t\right)={x}_{k}\left(t\right),\,k=1,2,\dots,n-1,\,t\in\R,\\
            \vx\left(\sum_{j=1}^{i}t_j^*\right)=\vx_{i,0:\left(n-1\right)}^*.
        \end{dcases}
    \]
    Then, according to Lemma \ref{lemma:3_2_phi},
    \[
        \resizebox{\textwidth}{!}{$
        \begin{aligned}
            \vJ\leftrightarrow&\left[\begin{array}{ccccc}
                \vPhi_n^{-1}\left(t_1^*\right)&&&&\\
                &\vPhi_n^{-1}\left(t_2^*\right)&&&\\
                &&\ddots&&\\
                &&&\vPhi_n^{-1}\left(t_N^*\right)&\\
                &&&&\vI_{\sum_{i=1}^{N-1}\abs{s_i}}\\
            \end{array}\right]\vJ\\
            \leftrightarrow&\left[\begin{array}{cccccccccc}
                \vI_n&-\vPhi_n\left(-t_1^*\right)&&&&\vx^{\left(1\right)}\left(0\right)&&&&\\
                &\vI_n&-\vPhi_n\left(-t_2^*\right)&&&&\vx^{\left(2\right)}\left(t_1^*\right)&&&\\
                &&\ddots&\ddots&&&&\ddots&&\\
                &&&\vI_n&-\vPhi_n\left(-t_{N-1}^*\right)&&&&\vx^{\left(N-1\right)}\left(\sum_{i=1}^{N-2}t_i^*\right)&\\
                &&&&\vI_n&&&&&\vx^{\left(N\right)}\left(\sum_{i=1}^{N-1}t_i^*\right)\\
                &\vI_{\abs{s_1}\times n}&&&&&&&&\\
                &&\vI_{\abs{s_2}\times n}&&&&&&&\\
                &&&\ddots&&&&&&\\
                &&&&\vI_{\abs{s_{N-1}}\times n}&&&&&
            \end{array}\right].
        \end{aligned}
        $}
    \]
    Upon left multiplication by
    \[
        \prod_{i=1}^{N-1}\left[\begin{array}{cccc}
            \vI_{\left(i-1\right)n}&&&\\
            &\vI_n&\vPhi_n\left(-t_{i}^*\right)&\\
            &&\vI_n&\\
            &&&\vI_{\left(N-i-1\right)n+\sum_{i=1}^{N-1}\abs{s_i}}
        \end{array}\right],
    \]
    we have
    \[
        \vJ\leftrightarrow\left[\begin{array}{cccccccc}
            \vI_n&&&&\vx^{\left(1\right)}\left(0\right)&\vx^{\left(2\right)}\left(0\right)&\cdots&\vx^{\left(N\right)}\left(0\right)\\
            &\vI_n&&&&\vx^{\left(2\right)}\left(t_1\right)&\cdots&\vx^{\left(N\right)}\left(t_1\right)\\
            &&\ddots&&&&\ddots&\vdots\\
            &&&\vI_n&&&&\vx^{\left(N\right)}\left(\sum_{i=1}^{N-1}t_i\right)\\
            &\vI_{\abs{s_1}\times n}&&&&&&\\
            &&\ddots&&&&&\\
            &&&\vI_{\abs{s_{N-1}}\times n}&&&&
        \end{array}\right].
    \]
    Upon left multiplication by
    \[
        \left[\begin{array}{ccccc}
            \vPhi_n\left(\sum_{i=1}^{N}t_{i}^*\right)&&&&\\
            &\vPhi_n\left(\sum_{i=2}^{N}t_{i}^*\right)&&&\\
            &&\ddots&&\\
            &&&\vPhi_n\left(t_{N}^*\right)&\\
            &&&&\vI_{\sum_{i=1}^{N-1}\abs{s_i}}
        \end{array}\right],
    \]
    since $t_\f^*=\sum_{i=1}^{N}t_{i}^*$, we have
    \[
        \vJ\leftrightarrow\left[\begin{array}{cccccccc}
            \vPhi_n\left(\sum_{i=1}^{N}t_{i}^*\right)&&&&\vx^{\left(1\right)}\left(t_\f^*\right)&\vx^{\left(2\right)}\left(t_\f^*\right)&\cdots&\vx^{\left(N\right)}\left(t_\f^*\right)\\
            &\vPhi_n\left(\sum_{i=2}^{N}t_{i}^*\right)&&&&\vx^{\left(2\right)}\left(t_\f^*\right)&\cdots&\vx^{\left(N\right)}\left(t_\f^*\right)\\
            &&\ddots&&&&\ddots&\vdots\\
            &&&\vPhi_n\left(t_{N}^*\right)&&&&\vx^{\left(N\right)}\left(t_\f^*\right)\\
            &\vI_{\abs{s_1}\times n}&&&&&&\\
            &&\ddots&&&&&\\
            &&&\vI_{\abs{s_{N-1}}\times n}&&&&
        \end{array}\right].
    \]

    According to Lemma \ref{lemma:3_2_phi}.\ref{lemma:3_2_phi_-1}, $\forall a\in\N^*$, $t\in\R$, $\vPhi_{a}\left(t\right)$ is invertible. Eliminate the first $n$ rows and the first $n$ columns of $\vJ$, we have
    \[
        \vJ\leftrightarrow\left[\begin{array}{ccccccccc}
            \vPhi_n\left(\sum_{i=2}^{N}t_{i}^*\right)&&&&&\vx^{\left(2\right)}\left(t_\f^*\right)&\vx^{\left(3\right)}\left(t_\f^*\right)&\cdots&\vx^{\left(N\right)}\left(t_\f^*\right)\\
            &\vPhi_n\left(\sum_{i=3}^{N}t_{i}^*\right)&&&&&\vx^{\left(3\right)}\left(t_\f^*\right)&\cdots&\vx^{\left(N\right)}\left(t_\f^*\right)\\
            &&\ddots&&&&&\ddots&\vdots\\
            &&&\vPhi_n\left(t_{N}^*\right)&&&&&\vx^{\left(N\right)}\left(t_\f^*\right)\\
            \vI_{\abs{s_1}\times n}&&&&&&&&\\
            &\vI_{\abs{s_2}\times n}&&&&&&&\\
            &&\ddots&&&&&&\\
            &&&\vI_{\abs{s_{N-1}}\times n}&&&&&
        \end{array}\right].
    \]

    For each $\vI_{\abs{s_i}\times n}=\left[\begin{array}{cc}\vI_{\abs{s_i}}&\vzero_{\abs{s_i}\times\left(n-\abs{s_i}\right)}\end{array}\right]$, $i=1,2,\dots,N-1$, eliminate the rows and columns occupied by $\vI_{\abs{s_i}}$. Then,
    \[
        \resizebox{\textwidth}{!}{$
        \vJ\leftrightarrow\left[\begin{array}{ccccccccc}
            \vzero_{\abs{s_1}\times\left(n-\abs{s_1}\right)}&&&&&\vx_{1:\abs{s_1}}^{\left(2\right)}\left(t_\f^*\right)&\vx_{1:\abs{s_1}}^{\left(3\right)}\left(t_\f^*\right)&\cdots&\vx_{1:\abs{s_1}}^{\left(N\right)}\left(t_\f^*\right)\\
            \vPhi_{n-\abs{s_1}}\left(\sum_{i=2}^{N}t_{i}^*\right)&&&&&\vx_{\left(\abs{s_1}+1\right):n}^{\left(2\right)}\left(t_\f^*\right)&\vx_{\left(\abs{s_1}+1\right):n}^{\left(3\right)}\left(t_\f^*\right)&\cdots&\vx_{\left(\abs{s_1}+1\right):n}^{\left(N\right)}\left(t_\f^*\right)\\
            &\vzero_{\abs{s_2}\times\left(n-\abs{s_2}\right)}&&&&&\vx_{1:\abs{s_2}}^{\left(3\right)}\left(t_\f^*\right)&\cdots&\vx_{1:\abs{s_2}}^{\left(N\right)}\left(t_\f^*\right)\\
            &\vPhi_{n-\abs{s_2}}\left(\sum_{i=3}^{N}t_{i}^*\right)&&&&&\vx_{\left(\abs{s_2}+1\right):n}^{\left(3\right)}\left(t_\f^*\right)&\cdots&\vx_{\left(\abs{s_2}+1\right):n}^{\left(N\right)}\left(t_\f^*\right)\\
            &&\ddots&&&&&\ddots&\vdots\\
            &&&\vzero_{\abs{s_{N-1}}\times\left(n-\abs{s_{N-1}}\right)}&&&&&\vx_{1:\abs{s_{N-1}}}^{\left(N\right)}\left(t_\f^*\right)\\
            &&&\vPhi_{n-\abs{s_{N-1}}}\left(t_N^*\right)&&&&&\vx_{\left(\abs{s_{N-1}}+1\right):n}^{\left(N\right)}\left(t_\f^*\right)\\
        \end{array}\right].
        $}
    \]
    Eliminate the rows and columns occupied by $\vPhi_{n-\abs{s_i}}\left(\sum_{j=i+1}^{N}t_{j}^*\right)$, $\forall i=1,2,\dots,N-1$. Then, $\vJ\leftrightarrow\vJ_1$, where
    \begin{equation}\label{eq:3_2_J1}
        \vJ_1=\left[\begin{array}{cccc}
            \vx_{1:\abs{s_1}}^{\left(2\right)}\left(t_\f^*\right)&\vx_{1:\abs{s_1}}^{\left(3\right)}\left(t_\f^*\right)&\cdots&\vx_{1:\abs{s_1}}^{\left(N\right)}\left(t_\f^*\right)\\
            &\vx_{1:\abs{s_2}}^{\left(3\right)}\left(t_\f^*\right)&\cdots&\vx_{1:\abs{s_2}}^{\left(N\right)}\left(t_\f^*\right)\\
            &&\ddots&\vdots\\
            &&&\vx_{1:\abs{s_{N-1}}}^{\left(N\right)}\left(t_\f^*\right)
        \end{array}\right].
    \end{equation}

    Next we prove $\vJ$ in \eqref{eq:3_2_J1} has full row rank. If $\forall i=1,2,\dots,N-1$, $\abs{s_i}=0$, then $\rank\vJ=Nn=Nn+\sum_{i=1}^{N-1}\abs{s_i}$. In this case, $\vJ$ has full row rank. We consider the case where $\exists i=1,2,\dots,N-1$, $\abs{s_i}\not=0$ in the following. Assume $1\leq i_1<i_2<\dots<i_{N'}\leq N-1$, s.t.
    \[
        \begin{dcases}
            \abs{s_i}\not=0,\text{ if }i\in\left\{i_j\right\}_{j=1}^{N'},\\
            \abs{s_i}=0,\text{ otherwise.}
        \end{dcases}
    \]
    Then, \eqref{eq:3_2_J1} can be rewritten as
    \begin{equation}\label{eq:3_2_J1_}
        \resizebox{0.9\textwidth}{!}{$
        \vJ_2=\left[\begin{array}{cccccccccccc}
            \vx_{1:\abs{s_{i_1}}}^{\left(i_1+1\right)}\left(t_\f^*\right)&\vx_{1:\abs{s_{i_1}}}^{\left(i_1+2\right)}\left(t_\f^*\right)&\cdots&\vx_{1:\abs{s_{i_1}}}^{\left(i_2\right)}\left(t_\f^*\right)&\vx_{1:\abs{s_{i_1}}}^{\left(i_2+1\right)}\left(t_\f^*\right)&\vx_{1:\abs{s_{i_1}}}^{\left(i_2+2\right)}\left(t_\f^*\right)&\cdots&\vx_{1:\abs{s_{i_1}}}^{\left(i_{N'}\right)}\left(t_\f^*\right)&\vx_{1:\abs{s_{i_1}}}^{\left(i_{N'}+1\right)}\left(t_\f^*\right)&\vx_{1:\abs{s_{i_1}}}^{\left(i_{N'}+2\right)}\left(t_\f^*\right)&\cdots&\vx_{1:\abs{s_{i_1}}}^{\left(N\right)}\left(t_\f^*\right)\\
            &&&&\vx_{1:\abs{s_{i_2}}}^{\left(i_2+1\right)}\left(t_\f^*\right)&\vx_{1:\abs{s_{i_2}}}^{\left(i_2+2\right)}\left(t_\f^*\right)&\cdots&\vx_{1:\abs{s_{i_2}}}^{\left(i_{N'}\right)}\left(t_\f^*\right)&\vx_{1:\abs{s_{i_2}}}^{\left(i_{N'}+1\right)}\left(t_\f^*\right)&\vx_{1:\abs{s_{i_2}}}^{\left(i_{N'}+2\right)}\left(t_\f^*\right)&\cdots&\vx_{1:\abs{s_{i_2}}}^{\left(N\right)}\left(t_\f^*\right)\\
            &&&&&&\ddots&\ddots&\vdots&\vdots&\ddots&\vdots\\
            &&&&&&&&\vx_{1:\abs{s_{i_{N'}}}}^{\left(i_{N'}+1\right)}\left(t_\f^*\right)&\vx_{1:\abs{s_{i_{N'}}}}^{\left(i_{N'}+2\right)}\left(t_\f^*\right)&\cdots&\vx_{1:\abs{s_{i_{N'}}}}^{\left(N\right)}\left(t_\f^*\right)\\
        \end{array}\right],
        $}
    \end{equation}
    where $\vJ\leftrightarrow\vJ_1\leftrightarrow\vJ_2$.

    According to \eqref{eq:R3C2_constraint1}, $\forall i=1,2,\dots,N$, $k=1,2,\dots,n$, $x^*_{i,k}=\sum_{j=1}^{k}\frac{1}{\left(k-j\right)!}x^*_{i-1,j}{t_i^*}^{k-j}+\frac{1}{k!}u_i{t_i^*}^k$. So $x^*_{i-1,k}=\sum_{j=1}^{k}\frac{1}{\left(k-j\right)!}x^*_{i,j}\left(-t_i^*\right)^{k-j}+\frac{1}{k!}u_i\left(-t_i^*\right)^k$.
    
    In other words,
    \[
        x^{\left(i-1\right)}_k\left(\sum_{j=1}^{i-1}t_j^*\right)=\sum_{j=1}^{k}\frac{\left(-t_i^*\right)^{k-j}}{\left(k-j\right)!}x^{\left(i\right)}_k\left(\sum_{j=1}^{i}t_j^*\right),\,\forall i=1,2,\dots,N,\,k=2,3,\dots,n.
    \]
    So $\forall i=1,2,\dots,N$,
    \[
        \begin{aligned}
            \vx^{\left(i\right)}\left(\sum_{j=1}^{i-1}t_j^*\right)=&\vPhi_n\left(-t_i^*\right)\vx^{\left(i\right)}\left(\sum_{j=1}^{i}t_j^*\right)=
            \left[\begin{array}{c}
                x^{\left(i\right)}_1\left(\sum_{j=1}^{i}t_j^*\right)\\
                \vx^{\left(i-1\right)}_{2:n}\left(\sum_{j=1}^{i-1}t_j^*\right)
            \end{array}\right]\\
            =&\vx^{\left(i-1\right)}\left(\sum_{j=1}^{i-1}t_j^*\right)+\left(x^{\left(i\right)}_1\left(\sum_{j=1}^{i}t_j^*\right)-x^{\left(i-1\right)}_1\left(\sum_{j=1}^{i-1}t_j^*\right)\right)\ve_1.
        \end{aligned}
    \]
    In other words, $\forall i=1,2,\dots,N$,
    \begin{equation}\label{eq:3_2_xi_xi-1_Deltaui}
        \vx^{\left(i\right)}\left(\sum_{j=1}^{i-1}t_j^*\right)-\vx^{\left(i-1\right)}\left(\sum_{j=1}^{i-1}t_j^*\right)=\Delta u_i\ve_1=\Delta u_i\vphi_n\left(0\right),
    \end{equation}
    where $\Delta u_i\triangleq x^{\left(i\right)}_1\left(\sum_{j=1}^{i}t_j^*\right)-x^{\left(i-1\right)}_1\left(\sum_{j=1}^{i-1}t_j^*\right)=u_i-u_{i-1}\not=0$, $\forall i=2,3,\dots,N$, since $u_i$ and $u_{i-1}$ are the control of the different stages $s_i$ and $s_{i-1}$, respectively. Specifically, if $\abs{s_i}\not=0$, then $u_{i}=0$, and $\vx_{1:\abs{s_i}}^{\left(i-1\right)}\left(\sum_{j=1}^{i}t_j^*\right)=\vzero_{\abs{s_i}\times1}$; hence
    \begin{equation}\label{eq:3_2_xi_xi-1_Deltaui_}
        \vx^{\left(i+1\right)}\left(\sum_{j=1}^{i}t_j^*\right)=u_i\vphi_n\left(0\right),\,\abs{s_i}\not=0.
    \end{equation}
    
    According to \eqref{eq:3_2_xi_xi-1_Deltaui}, \eqref{eq:3_2_xi_xi-1_Deltaui_}, and Lemma \ref{lemma:3_2_phi}, let
    \begin{equation}\label{eq:3_2_J3}
        \resizebox{0.9\textwidth}{!}{$
        \begin{aligned}
            &\vJ_3\\
            =&\left(\prod_{j=i_1+1}^{N}\left[\begin{array}{cccc}
                \vPhi_{\abs{s_{i_1}}}\left(-t_j\right)&&&\\
                &\vPhi_{\abs{s_{i_2}}}\left(-t_j\right)&&\\
                &&\ddots&\\
                &&&\vPhi_{\abs{s_{i_{N'}}}}\left(-t_j\right)
            \end{array}\right]\right)
            \vJ_2
            \left(\prod_{j=i_1+1}^{N-1}\left[\begin{array}{cccc}
                \vI_{N-1-j}&&&\\
                &1&-1&\\
                &&1&\\
                &&&\vI_{j-i_1-1}
            \end{array}\right]\right)
            \left[\begin{array}{cccc}
                \frac{1}{\Delta u_{i_1+1}}&&&\\
                &\frac{1}{\Delta u_{i_1+2}}&&\\
                &&\ddots&\\
                &&&\frac{1}{\Delta u_{N}}
            \end{array}\right]\\
            =&\left[\begin{array}{cccccccccccc}
                \vphi_{\abs{s_{i_1}}}\left(-T_{i_1+1}\right)&\vphi_{\abs{s_{i_1}}}\left(-T_{i_1+2}\right)&\cdots&\vphi_{\abs{s_{i_1}}}\left(-T_{i_2}\right)&\vphi_{\abs{s_{i_1}}}\left(-T_{i_2+1}\right)&\vphi_{\abs{s_{i_1}}}\left(-T_{i_2+2}\right)&\cdots&\vphi_{\abs{s_{i_1}}}\left(-T_{i_{N'}}\right)&\vphi_{\abs{s_{i_1}}}\left(-T_{i_{N'}+1}\right)&\vphi_{\abs{s_{i_1}}}\left(-T_{i_{N'}+2}\right)&\cdots&\vphi_{\abs{s_{i_1}}}\left(-T_{N}\right)\\
                &&&&\vphi_{\abs{s_{i_2}}}\left(-T_{i_2+1}\right)&\vphi_{\abs{s_{i_2}}}\left(-T_{i_2+2}\right)&\cdots&\vphi_{\abs{s_{i_2}}}\left(-T_{i_{N'}}\right)&\vphi_{\abs{s_{i_2}}}\left(-T_{i_{N'}+1}\right)&\vphi_{\abs{s_{i_2}}}\left(-T_{i_{N'}+2}\right)&\cdots&\vphi_{\abs{s_{i_2}}}\left(-T_{N}\right)\\
                &&&&&&\ddots&\ddots&\vdots&\vdots&\ddots&\vdots\\
                &&&&&&&&\vphi_{\abs{s_{i_{N'}}}}\left(-T_{i_{N'}+1}\right)&\vphi_{\abs{s_{i_{N'}}}}\left(-T_{i_{N'}+2}\right)&\cdots&\vphi_{\abs{s_{i_{N'}}}}\left(-T_{N}\right)\\
            \end{array}\right]
        \end{aligned}
        $}
    \end{equation}
    where $T_i=\sum_{j=i_1+1}^{i-1}t_j^*$. It is evidently that $\vJ\leftrightarrow\vJ_2\leftrightarrow\vJ_3$. Note that \eqref{eq:3_2_J3} successfully avoid the complex background of the time-optimal problem. Left-multiplying matrix $\vJ_3$ by
    \[
        \left[\begin{array}{cccc}
            \mathrm{diag}\left(\left(\left(-1\right)^{j-1}\left(j-1\right)!\right)_{j=1}^{\abs{s_{i_1}}}\right)&&&\\
            &\mathrm{diag}\left(\left(\left(-1\right)^{j-1}\left(j-1\right)!\right)_{j=1}^{\abs{s_{i_2}}}\right)&&\\
            &&\ddots&\\
            &&&\mathrm{diag}\left(\left(\left(-1\right)^{j-1}\left(j-1\right)!\right)_{j=1}^{\abs{s_{i_{N'}}}}\right)\\
        \end{array}\right],
    \]
    we have $\vJ\leftrightarrow\vJ_3\leftrightarrow\vJ_4$, where
    \begin{equation}\label{eq:3_2_J4}
        \resizebox{0.9\textwidth}{!}{$
        \vJ_4=\left[\begin{array}{cccccccccccc}
            \hat{\vphi}_{\abs{s_{i_1}}}\left(T_{i_1+1}\right)&\hat{\vphi}_{\abs{s_{i_1}}}\left(T_{i_1+2}\right)&\cdots&\hat{\vphi}_{\abs{s_{i_1}}}\left(T_{i_2}\right)&\hat{\vphi}_{\abs{s_{i_1}}}\left(T_{i_2+1}\right)&\hat{\vphi}_{\abs{s_{i_1}}}\left(T_{i_2+2}\right)&\cdots&\hat{\vphi}_{\abs{s_{i_1}}}\left(T_{i_{N'}}\right)&\hat{\vphi}_{\abs{s_{i_1}}}\left(T_{i_{N'}+1}\right)&\hat{\vphi}_{\abs{s_{i_1}}}\left(T_{i_{N'}+2}\right)&\cdots&\hat{\vphi}_{\abs{s_{i_1}}}\left(T_{N}\right)\\
            &&&&\hat{\vphi}_{\abs{s_{i_2}}}\left(T_{i_2+1}\right)&\hat{\vphi}_{\abs{s_{i_2}}}\left(T_{i_2+2}\right)&\cdots&\hat{\vphi}_{\abs{s_{i_2}}}\left(T_{i_{N'}}\right)&\hat{\vphi}_{\abs{s_{i_2}}}\left(T_{i_{N'}+1}\right)&\hat{\vphi}_{\abs{s_{i_2}}}\left(T_{i_{N'}+2}\right)&\cdots&\hat{\vphi}_{\abs{s_{i_2}}}\left(T_{N}\right)\\
            &&&&&&\ddots&\ddots&\vdots&\vdots&\ddots&\vdots\\
            &&&&&&&&\hat{\vphi}_{\abs{s_{i_{N'}}}}\left(T_{i_{N'}+1}\right)&\hat{\vphi}_{\abs{s_{i_{N'}}}}\left(T_{i_{N'}+2}\right)&\cdots&\hat{\vphi}_{\abs{s_{i_{N'}}}}\left(T_{N}\right)\\
        \end{array}\right].
        $}
    \end{equation}

    In \eqref{eq:3_2_J4}, $T_{i_1+1}<T_{i_1+2}<\dots<T_N$, since $\forall i=1,2,\dots,N$, $t_i>0$. $\vJ_4$ in \eqref{eq:3_2_J4} has a form similar to the Vandermonde matrix. Now, we only need to proof that $\vJ_4$ has full row rank based on Lemma \ref{lemma:3_2_difference}.
    
    Denote $\vD_{m,p}^{a:b}\triangleq\left(D_{l,m,p}\right)_{l=a}^b$, where $D_{l,m,p}=D_{l,m}\left(\left(T_k\right)_{k=p}^{m+p}\right)$ is defined in \eqref{eq:3_2_definedifference}. Then, $\hat{\vphi}_l\left(T_p\right)=\vD_{0,p}^{0:\left(l-1\right)}$. Therefore,
    \begin{equation}
        \resizebox{0.9\textwidth}{!}{$
        \vJ_4=\left[\begin{array}{cccccccccccc}
            \vD_{0,i_1+1}^{0:\left(\abs{s_{i_1}}-1\right)}&\vD_{0,i_1+2}^{0:\left(\abs{s_{i_1}}-1\right)}&\cdots&\vD_{0,i_2}^{0:\left(\abs{s_{i_1}}-1\right)}&\vD_{0,i_2+1}^{0:\left(\abs{s_{i_1}}-1\right)}&\vD_{0,i_2+2}^{0:\left(\abs{s_{i_1}}-1\right)}&\cdots&\vD_{0,i_{N'}}^{0:\left(\abs{s_{i_1}}-1\right)}&\vD_{0,i_{N'}+1}^{0:\left(\abs{s_{i_1}}-1\right)}&\vD_{0,i_{N'}+2}^{0:\left(\abs{s_{i_1}}-1\right)}&\cdots&\vD_{0,N}^{0:\left(\abs{s_{i_1}}-1\right)}\\
            &&&&\vD_{0,i_2+1}^{0:\left(\abs{s_{i_2}}-1\right)}&\vD_{0,i_2+2}^{0:\left(\abs{s_{i_2}}-1\right)}&\cdots&\vD_{0,i_{N'}}^{0:\left(\abs{s_{i_2}}-1\right)}&\vD_{0,i_{N'}+1}^{0:\left(\abs{s_{i_2}}-1\right)}&\vD_{0,i_{N'}+2}^{0:\left(\abs{s_{i_2}}-1\right)}&\cdots&\vD_{0,N}^{0:\left(\abs{s_{i_2}}-1\right)}\\
            &&&&&&\ddots&\ddots&\vdots&\vdots&\ddots&\vdots\\
            &&&&&&&&\vD_{0,i_{N'}+1}^{0:\left(\abs{s_{i_{N'}}}-1\right)}&\vD_{0,i_{N'}+2}^{0:\left(\abs{s_{i_{N'}}}-1\right)}&\cdots&\vD_{0,N}^{0:\left(\abs{s_{i_{N'}}}-1\right)}
        \end{array}\right].
        $}
    \end{equation}

    We prove $\vJ_4$ has full row rank through row operations. At each step, we first subtract the preceding column from the succeeding one and then divide by the corresponding time difference. Take the first step as example:
    \begin{equation*}
        \resizebox{\textwidth}{!}{$
        \begin{aligned}
            &\vJ\leftrightarrow\vJ_4\\
            =&\left[\begin{array}{cccccccccccc}
                1&1&\cdots&1&1&1&\cdots&1&1&1&\cdots&1\\
                \vD_{0,i_1+1}^{1:\left(\abs{s_{i_1}}-1\right)}&\vD_{0,i_1+2}^{1:\left(\abs{s_{i_1}}-1\right)}&\cdots&\vD_{0,i_2}^{1:\left(\abs{s_{i_1}}-1\right)}&\vD_{0,i_2+1}^{1:\left(\abs{s_{i_1}}-1\right)}&\vD_{0,i_2+2}^{1:\left(\abs{s_{i_1}}-1\right)}&\cdots&\vD_{0,i_{N'}}^{1:\left(\abs{s_{i_1}}-1\right)}&\vD_{0,i_{N'}+1}^{1:\left(\abs{s_{i_1}}-1\right)}&\vD_{0,i_{N'}+2}^{1:\left(\abs{s_{i_1}}-1\right)}&\cdots&\vD_{0,N}^{1:\left(\abs{s_{i_1}}-1\right)}\\
                &&&&1&1&\cdots&1&1&1&\cdots&1\\
                &&&&\vD_{0,i_2+1}^{1:\left(\abs{s_{i_2}}-1\right)}&\vD_{0,i_2+2}^{1:\left(\abs{s_{i_2}}-1\right)}&\cdots&\vD_{0,i_{N'}}^{1:\left(\abs{s_{i_2}}-1\right)}&\vD_{0,i_{N'}+1}^{1:\left(\abs{s_{i_2}}-1\right)}&\vD_{0,i_{N'}+2}^{1:\left(\abs{s_{i_2}}-1\right)}&\cdots&\vD_{0,N}^{1:\left(\abs{s_{i_2}}-1\right)}\\
                &&&&&&\ddots&\ddots&\vdots&\vdots&\ddots&\vdots\\
                &&&&&&&&1&1&\cdots&1\\
                &&&&&&&&\vD_{0,i_{N'}+1}^{1:\left(\abs{s_{i_{N'}}}-1\right)}&\vD_{0,i_{N'}+2}^{1:\left(\abs{s_{i_{N'}}}-1\right)}&\cdots&\vD_{0,N}^{1:\left(\abs{s_{i_{N'}}}-1\right)}
            \end{array}\right]\\
            \leftrightarrow&\left[\begin{array}{cccccccccccc}
                1&0&\cdots&0&0&0&\cdots&0&0&0&\cdots&0\\
                \vD_{0,i_1+1}^{1:\left(\abs{s_{i_1}}-1\right)}&\vD_{1,i_1+1}^{1:\left(\abs{s_{i_1}}-1\right)}&\cdots&\vD_{1,i_2-1}^{1:\left(\abs{s_{i_1}}-1\right)}&\vD_{1,i_2}^{1:\left(\abs{s_{i_1}}-1\right)}&\vD_{1,i_2+1}^{1:\left(\abs{s_{i_1}}-1\right)}&\cdots&\vD_{1,i_{N'}-1}^{1:\left(\abs{s_{i_1}}-1\right)}&\vD_{1,i_{N'}}^{1:\left(\abs{s_{i_1}}-1\right)}&\vD_{1,i_{N'}+1}^{1:\left(\abs{s_{i_1}}-1\right)}&\cdots&\vD_{1,N-1}^{1:\left(\abs{s_{i_1}}-1\right)}\\
                &&&&1&0&\cdots&0&0&0&\cdots&0\\
                &&&&\vD_{0,i_2+1}^{1:\left(\abs{s_{i_2}}-1\right)}&\vD_{1,i_2+1}^{1:\left(\abs{s_{i_2}}-1\right)}&\cdots&\vD_{1,i_{N'}-1}^{1:\left(\abs{s_{i_2}}-1\right)}&\vD_{1,i_{N'}}^{1:\left(\abs{s_{i_2}}-1\right)}&\vD_{1,i_{N'}+1}^{1:\left(\abs{s_{i_2}}-1\right)}&\cdots&\vD_{1,N-1}^{1:\left(\abs{s_{i_2}}-1\right)}\\
                &&&&&&\ddots&\ddots&\vdots&\vdots&\ddots&\vdots\\
                &&&&&&&&1&0&\cdots&0\\
                &&&&&&&&\vD_{0,i_{N'}+1}^{1:\left(\abs{s_{i_{N'}}}-1\right)}&\vD_{1,i_{N'}+1}^{1:\left(\abs{s_{i_{N'}}}-1\right)}&\cdots&\vD_{1,N-1}^{1:\left(\abs{s_{i_{N'}}}-1\right)}
            \end{array}\right]\\
            \leftrightarrow&\left[\begin{array}{ccccccccc}
                \vD_{1,i_1+1}^{1:\left(\abs{s_{i_1}}-1\right)}&\cdots&\vD_{1,i_2-1}^{1:\left(\abs{s_{i_1}}-1\right)}&\vD_{1,i_2+1}^{1:\left(\abs{s_{i_1}}-1\right)}&\cdots&\vD_{1,i_{N'}-1}^{1:\left(\abs{s_{i_1}}-1\right)}&\vD_{1,i_{N'}+1}^{1:\left(\abs{s_{i_1}}-1\right)}&\cdots&\vD_{1,N-1}^{1:\left(\abs{s_{i_1}}-1\right)}\\
                &&&\vD_{1,i_2+1}^{1:\left(\abs{s_{i_2}}-1\right)}&\cdots&\vD_{1,i_{N'}-1}^{1:\left(\abs{s_{i_2}}-1\right)}&\vD_{1,i_{N'}+1}^{1:\left(\abs{s_{i_2}}-1\right)}&\cdots&\vD_{1,N-1}^{1:\left(\abs{s_{i_2}}-1\right)}\\
                &&&&\ddots&\ddots&\vdots&\ddots&\vdots\\
                &&&&&&\vD_{1,i_{N'}+1}^{1:\left(\abs{s_{i_{N'}}}-1\right)}&\cdots&\vD_{1,N-1}^{1:\left(\abs{s_{i_{N'}}}-1\right)}
            \end{array}\right].
        \end{aligned}
        $}
    \end{equation*}
    At the second step, elements $\vD_{1,i_j}^{1:\left(\abs{s_{i_l}}-1\right)}$ does not exist. Fortunately, the difference between $\vD_{1,i_j}^{1:\left(\abs{s_{i_l}}\pm1\right)}$ can be eliminated, resulting in a matrix with similar structures to the above one. At the second step, it is noteworthy that $D_{1,1,p}=1$. Then,
    \begin{equation*}
        \resizebox{\textwidth}{!}{$
        \begin{aligned}
            &\vJ\\
            \leftrightarrow&\left[\begin{array}{cccccccccccc}
                1&1&\cdots&1&1&1&\cdots&1&1&1&\cdots&1\\
                \vD_{1,i_1+1}^{2:\left(\abs{s_{i_1}}-1\right)}&\vD_{1,i_1+2}^{2:\left(\abs{s_{i_1}}-1\right)}&\cdots&\vD_{1,i_2-1}^{2:\left(\abs{s_{i_1}}-1\right)}&\vD_{1,i_2+1}^{2:\left(\abs{s_{i_1}}-1\right)}&\vD_{1,i_2+2}^{2:\left(\abs{s_{i_1}}-1\right)}&\cdots&\vD_{1,i_{N'}-1}^{2:\left(\abs{s_{i_1}}-1\right)}&\vD_{1,i_{N'}+1}^{2:\left(\abs{s_{i_1}}-1\right)}&\vD_{1,i_{N'}+2}^{2:\left(\abs{s_{i_1}}-1\right)}&\cdots&\vD_{1,N-1}^{2:\left(\abs{s_{i_1}}-1\right)}\\
                &&&&1&1&\cdots&1&1&1&\cdots&1\\
                &&&&\vD_{1,i_2+1}^{2:\left(\abs{s_{i_2}}-1\right)}&\vD_{1,i_2+2}^{2:\left(\abs{s_{i_2}}-1\right)}&\cdots&\vD_{1,i_{N'}-1}^{2:\left(\abs{s_{i_2}}-1\right)}&\vD_{1,i_{N'}+1}^{2:\left(\abs{s_{i_2}}-1\right)}&\vD_{1,i_{N'}+2}^{2:\left(\abs{s_{i_2}}-1\right)}&\cdots&\vD_{1,N-1}^{2:\left(\abs{s_{i_2}}-1\right)}\\
                &&&&&&\ddots&\ddots&\vdots&\vdots&\ddots&\vdots\\
                &&&&&&&&1&1&\cdots&1\\
                &&&&&&&&\vD_{1,i_{N'}+1}^{2:\left(\abs{s_{i_{N'}}}-1\right)}&\vD_{1,i_{N'}+2}^{2:\left(\abs{s_{i_{N'}}}-1\right)}&\cdots&\vD_{1,N-1}^{2:\left(\abs{s_{i_{N'}}}-1\right)}
            \end{array}\right]\\
            \leftrightarrow&\left[\begin{array}{cccccccccccc}
                1&0&\cdots&0&0&0&\cdots&0&0&0&\cdots&0\\
                \vD_{1,i_1+1}^{2:\left(\abs{s_{i_1}}-1\right)}&\vD_{2,i_1+1}^{2:\left(\abs{s_{i_1}}-1\right)}&\cdots&\vD_{2,i_2-2}^{2:\left(\abs{s_{i_1}}-1\right)}&*&\vD_{2,i_2+1}^{2:\left(\abs{s_{i_1}}-1\right)}&\cdots&\vD_{2,i_{N'}-2}^{2:\left(\abs{s_{i_1}}-1\right)}&*&\vD_{2,i_{N'}+1}^{2:\left(\abs{s_{i_1}}-1\right)}&\cdots&\vD_{2,N-2}^{2:\left(\abs{s_{i_1}}-1\right)}\\
                &&&&1&0&\cdots&0&0&0&\cdots&0\\
                &&&&\vD_{1,i_2+1}^{2:\left(\abs{s_{i_2}}-1\right)}&\vD_{2,i_2+1}^{2:\left(\abs{s_{i_2}}-1\right)}&\cdots&\vD_{2,i_{N'}-2}^{2:\left(\abs{s_{i_2}}-1\right)}&*&\vD_{2,i_{N'}+1}^{2:\left(\abs{s_{i_2}}-1\right)}&\cdots&\vD_{2,N-2}^{2:\left(\abs{s_{i_2}}-1\right)}\\
                &&&&&&\ddots&\ddots&\vdots&\vdots&\ddots&\vdots\\
                &&&&&&&&1&0&\cdots&0\\
                &&&&&&&&\vD_{1,i_{N'}+1}^{2:\left(\abs{s_{i_{N'}}}-1\right)}&\vD_{2,i_{N'}+1}^{2:\left(\abs{s_{i_{N'}}}-1\right)}&\cdots&\vD_{2,N-2}^{2:\left(\abs{s_{i_{N'}}}-1\right)}
            \end{array}\right]\\
            \leftrightarrow&\left[\begin{array}{cccccccccc}
                \vD_{2,i_1+1}^{2:\left(\abs{s_{i_1}}-1\right)}&\cdots&\vD_{2,i_2-2}^{2:\left(\abs{s_{i_1}}-1\right)}&\vD_{2,i_2+1}^{2:\left(\abs{s_{i_1}}-1\right)}&\cdots&\vD_{2,i_{N'}-2}^{2:\left(\abs{s_{i_1}}-1\right)}&\vD_{2,i_{N'}+1}^{2:\left(\abs{s_{i_1}}-1\right)}&\cdots&\vD_{2,N-2}^{2:\left(\abs{s_{i_1}}-1\right)}\\
                &&&\vD_{2,i_2+1}^{2:\left(\abs{s_{i_2}}-1\right)}&\cdots&\vD_{2,i_{N'}-2}^{2:\left(\abs{s_{i_2}}-1\right)}&\vD_{2,i_{N'}+1}^{2:\left(\abs{s_{i_2}}-1\right)}&\cdots&\vD_{2,N-2}^{2:\left(\abs{s_{i_2}}-1\right)}\\
                &&&&\ddots&\ddots&\vdots&\ddots&\vdots\\
                &&&&&&\vD_{2,i_{N'}+1}^{2:\left(\abs{s_{i_{N'}}}-1\right)}&\cdots&\vD_{2,N-2}^{2:\left(\abs{s_{i_{N'}}}-1\right)}
            \end{array}\right].
        \end{aligned}
        $}
    \end{equation*}

    It can be proved by Condition \ref{condition:3_2_switchinglaw} that the above process can be continued until the last step, i.e., the matrix is empty. Therefore, $\vJ$ has full row rank. Theorem \ref{thm:3_2_fullrowrank} holds.
\end{proof}

\begin{theorem}
    If Condition \ref{condition:3_2_switchinglaw} holds, then the switching law $S$ induces a manifold of $\dim S=N-\sum_{i=1}^{N}\abs{s_i}$ dimension.
\end{theorem}

\begin{proof}
    At $\vX^*$, $\vf\left(\vX^*\right)=\vzero$. Theorem \ref{thm:3_2_fullrowrank} implies that $\vJ^*=\where{\frac{\mathrm{d}\vf}{\mathrm{d}\vX}}{\vX^*}=\left(\vj_1^*,\vj_2^*,\dots,\vj_{N\left(n+1\right)}^*\right)$ has full row rank, i.e., $\rank\vJ^*=Nn+\sum_{i=1}^{N}\abs{s_i}$. Assume that $\left\{\vj_{i_k}\right\}_{k=1}^{\rank\vJ^*}$ are linearly independent. Denote $\left\{i_k'\right\}_{k=1}^{N\left(n+1\right)-\rank\vJ^*}=\Setminus{\left\{k\right\}_{k=1}^{N\left(n+1\right)}}{\left\{\vj_{i_k}\right\}_{k=1}^{\rank\vJ^*}}$. Applying the implicit function theorem, $\vf\left(\vX\right)=\vzero$ induces a smooth map $$\hat{\vf}:\left(X_{i_k'}\right)_{k=1}^{N\left(n+1\right)-\rank\vJ^*}\mapsto\left(X_{i_k}\right)_{k=1}^{\rank\vJ^*}$$ near $\vX^*$. Therefore, the map of the smooth function $\hat{\vf}$, i.e., the set induced by $\vf\left(\vX\right)=\vzero$ near $\vX^*$, is a manifold of dimension $N\left(n+1\right)-\rank\vJ^*=N-\sum_{i=1}^{N}\abs{s_i}$. The theorem holds.
\end{proof}

\end{document}